\newif\ifieee
\ieeefalse

\def\showauthornotes{1}

\ifieee
  \documentclass[10pt, conference, compsocconf]{IEEEtran}
\else
  \documentclass[11pt,a4paper]{article}
  \usepackage{fullpage}
\fi
\usepackage[utf8]{inputenc}
\usepackage{enumerate}
\usepackage{amsmath}
\usepackage{amssymb}
\usepackage{graphicx}
\usepackage{ae}
\usepackage[breaklinks,bookmarks=false]{hyperref}
\usepackage[usenames,dvipsnames,svgnames,table]{xcolor}
\usepackage{bbm}
\usepackage[amsthm,amsmath,thmmarks,hyperref]{ntheorem}
\usepackage[capitalise,noabbrev]{cleveref}
\usepackage{mathtools}
\usepackage{ifthen}

\usepackage{tikz}
\usepackage{tkz-berge}
\usetikzlibrary{positioning}
\usetikzlibrary{snakes}
\usetikzlibrary{backgrounds}
\usetikzlibrary{calc}

\usetikzlibrary{external}

\usepackage{caption}
\usepackage{subfig}
\newtheorem{theorem}{Theorem}[section]
\newtheorem{corollary}[theorem]{Corollary}
\newtheorem{lemma}[theorem]{Lemma}

\newtheorem{proposition}[theorem]{Proposition}
\newtheorem{definition}[theorem]{Definition}
\newtheorem{claim}{Claim} 
\newtheorem{fact}[theorem]{Fact}

\newcommand{\qedmanual}{\hfill\ensuremath{\square}}
\qedsymbol{\openbox}

\DeclareMathOperator{\supp}{supp}
\DeclareMathOperator{\poly}{poly}

\newcommand{\cT}{\mathcal{T}}

\newcommand{\cS}{\mathcal{S}}
\newcommand{\cL}{\mathcal{L}}
\newcommand{\cW}{\mathcal{W}}
\newcommand{\cWf}{\cW}

\newcommand{\bR}{\mathbb{R}}
\newcommand{\one}{\mathbbm{1}}

\newcommand{\ab}[1]{\left<{#1}\right>} 
\newcommand{\rb}[1]{\left( #1 \right)} 
\newcommand{\cb}[1]{\left\{ #1 \right\} } 
\newcommand{\abs}[1]{\left| #1 \right|} 

\ifnum\showauthornotes=1
\newcommand{\Authornote}[2]{{\sffamily\small\color{red}{[#1: #2]}}}
\newcommand{\Authornotecolored}[3]{{\sffamily\small\color{#1}{[#2: #3]}}}
\newcommand{\Authorcomment}[2]{{\sffamily\small\color{gray}{[#1: #2]}}}
\newcommand{\Authorstartcomment}[1]{\sffamily\small\color{gray}[#1: }

\newcommand{\Authorfnote}[2]{\footnote{\color{red}{#1: #2}}}
\newcommand{\Authorfixme}[1]{\Authornote{#1}{\textbf{??}}}
\newcommand{\Authormarginmark}[1]{\marginpar{\textcolor{red}{\fbox{\Large #1:!}}}}
\else
\newcommand{\Authornote}[2]{}
\newcommand{\Authornotecolored}[3]{}
\newcommand{\Authorcomment}[2]{}
\newcommand{\Authorstartcomment}[1]{}

\newcommand{\Authorfnote}[2]{}
\newcommand{\Authorfixme}[1]{}
\newcommand{\Authormarginmark}[1]{}
\fi

\DeclareMathOperator{\PM}{PM}

\DeclareMathOperator{\argmin}{argmin}
\DeclareMathOperator{\Span}{span}

\DeclareMathOperator{\NW}{NW}
\DeclareMathOperator{\incoming}{in}
\DeclareMathOperator{\outgoing}{out}
\DeclareMathOperator{\lcm}{lcm}
\DeclareMathOperator{\mymod}{mod}
\newcommand{\suppo}{E}
\newcommand{\tight}{\mathcal{S}}
\newcommand{\Fmid}{F_{\textsf{\tiny mid}}}
\newcommand{\wmid}{w_{\textsf{\tiny mid}}}

\newcommand{\Fin}{\ensuremath{F_{\textsf{\tiny in}}}}
\newcommand{\Lin}{\ensuremath{\cL_{\textsf{\tiny in}}}}

\newcommand{\Fout}{\ensuremath{F_{\textsf{\tiny out}}}}
\newcommand{\Lout}{\ensuremath{\cL_{\textsf{\tiny out}}}}
\newcommand{\wout}{\ensuremath{w_{\textsf{\tiny out}}}}

\newcommand{\cM}{\mathcal{M}}

\newcommand{\bZ}{\mathbb{Z}}
\newcommand{\cY}{\mathcal{Y}}
\newcommand{\cZ}{\mathcal{Z}}

\newcommand{\pmind}[1]{(\pm \one)_{#1}}
\newcommand{\facemin}[2]{\ensuremath{#1 [ #2 ]}}
\newcommand{\quasinc}{\text{\sf quasi-}{\sf NC}}
\newcommand{\rnc}{{\sf RNC}}
\newcommand{\nc}{{\sf NC}}

\begin{document}

\ifieee
\else
	\title{The Matching Problem in General Graphs is in Quasi-NC}
	\author{
	Ola Svensson\thanks{School of Computer and Communication Sciences, EPFL. \newline
	Email: \texttt{\{ola.svensson,jakub.tarnawski\}@epfl.ch}. \newline
	Supported by ERC Starting Grant 335288-OptApprox.}
	\and
	Jakub Tarnawski\footnotemark[1]
	}
	
	\date{\today}
	\clearpage\maketitle
	\thispagestyle{empty}
\fi

\begin{abstract}
  We show that the perfect matching problem in general graphs is in
  $\quasinc$.  That is, we give a deterministic parallel algorithm which runs
  in $O(\log^3 n)$ time on $n^{O(\log^2 n)}$ processors.  The result is
  obtained by a derandomization of the Isolation Lemma for perfect matchings,
  which was introduced in the classic paper by  Mulmuley, Vazirani and
  Vazirani [1987] to obtain a Randomized $\nc$ algorithm.
  
  Our proof extends the framework of Fenner,
  Gurjar and Thierauf [2016], who proved the analogous result in the special
  case of bipartite graphs. Compared to that setting, several new ingredients are
  needed due to the significantly more complex structure of perfect
  matchings in general graphs. In particular, our proof heavily relies on the
  laminar structure of the faces of the perfect matching polytope. 
\end{abstract}

\ifieee
\else
	\clearpage
\fi

\section{Introduction} \label{sec:intro}

The perfect matching problem is a fundamental question in graph theory.
Work on this problem has contributed to the development of many core concepts of modern computer science,
including linear-algebraic, probabilistic and parallel algorithms.
Edmonds~\cite{edm65matching} was the first to give a polynomial-time algorithm for it.
However, half a century later, we still do not have full understanding of the deterministic parallel complexity of the perfect matching problem.
In this paper we make progress in this direction.

We consider a problem to be efficiently solvable in parallel if it has an algorithm which uses polylogarithmic time and polynomially many processors.
More formally, a problem is in the class $\nc$ if it has uniform circuits of polynomial size and polylogarithmic depth.
The class $\rnc$ is obtained if we also allow randomization.

We study the decision version of the problem: given an undirected simple graph, determine whether it has a perfect matching --
and the search version: find and return a perfect matching if one exists.
The decision version was first shown to be in $\rnc$ by Lovász~\cite{Lovasz79}.
The search version has proved to be more difficult and it was found to be in $\rnc$ several years later by Karp, Upfal and Wigderson~\cite{KarpUW86} and Mulmuley, Vazirani and Vazirani~\cite{MulmuleyVV87}.
All these algorithms are randomized, and it remains a major open problem to determine whether randomness is required, i.e., whether either version is in $\nc$.
\ifieee\else
Intuitively, one difficulty is that
a graph may contain super-polynomially many perfect matchings
and it is necessary to somehow coordinate the processors so that they locate the same one. 
Our vocabulary for this is that we want to \emph{isolate} one matching.
\fi

A successful approach to the perfect matching problem has been the linear-algebraic one.
It involves the {Tutte matrix} associated with a graph $G = (V,E)$, which is a $|V| \times |V|$
matrix defined as follows (see \cref{fig:tutte} for an example):

\begin{figure}
	\centering
	\begin{tabular}{cc}
	\begin{minipage}{\ifieee\else 0.2\textwidth \fi}
		\begin{tikzpicture}
		  [scale=1.2,every node/.style={circle,fill=gray!40}]
		  \node (n3) at (0,0) {3};
		  \node (n1) at (0,1) {1};
		  \node (n4) at (1,0) {4};
		  \node (n2) at (1,1) {2};
		  \foreach \from/\to in {n1/n2,n2/n4,n3/n4,n4/n1,n1/n3}
	    	\draw (\from) -- (\to);
		\end{tikzpicture}
	\end{minipage}
		&
		$T(G) = \begin{pmatrix}
			0 & X_{12} & X_{13} & X_{14}\\
			-X_{12} & 0 & 0 & X_{24}\\
			-X_{13} & 0 & 0 & X_{34}\\
			-X_{14} & -X_{24} & -X_{34} & 0
		\end{pmatrix}$
	\end{tabular}
	\caption{Example of a Tutte matrix of an undirected graph.}
	\label{fig:tutte}
\end{figure}
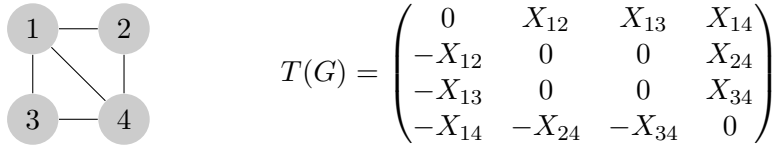

\[ T(G)_{u,v} = \begin{cases} X_{(u,v)} & \text{if $(u,v) \in E$ and $u < v$}, \\ -X_{(v,u)} & \text{if $(u,v) \in E$ and $u > v$}, \\ 0 & \text{if $(u,v) \not \in E$}, \end{cases} \]
where $X_{(u,v)}$ for $(u,v) \in E$ are variables.
Tutte's theorem~\cite{Tutte} says that $\det T(G) \ne 0$ if and only if $G$ has a perfect matching.
This is great news for parallelization, as computing determinants is in $\nc$~\cite{Csanky76,Berkowitz1984}.
However, the matrix is defined over a ring of indeterminates, so randomness is normally used in order to test if the determinant is nonzero.
One approach is to replace each indeterminate by a random value from a large field.
This leads, among others, to the fastest known (single-processor) running times for dense graphs~\cite{MuchaS04,Harvey09}.

\newcommand{\specialm}{M} 
A second approach, adopted by Mulmuley, Vazirani and Vazirani~\cite{MulmuleyVV87} for the search version, is to replace the indeterminates by randomly chosen powers of two.
Namely, for each edge $(u,v)$, a random weight $w(u,v) \in \{1, 2, ..., 2|E|\}$ is selected, and we substitute $X_{(u,v)} := 2^{w(u,v)}$.
Now, let us make the crucial assumption that one perfect matching $\specialm$ is \textbf{isolated},
in the sense that it is the \emph{unique minimum-weight perfect matching} (minimizing $w(M)$).
Then $\det T(G)$ remains nonzero after the substitution:
one can show that
$\specialm$ contributes a term $\pm 2^{2w(\specialm)}$ to $\det T(G)$, whereas
all other terms are
multiples of $2^{2w(\specialm) + 1}$
and thus they cannot cancel $2^{2w(\specialm)}$ out.
The determinant can still be computed in $\nc$ as all entries $2^{w(u,v)}$ of the matrix are of polynomial bit-length,
and so we have a parallel algorithm for the decision version.
An algorithm for the search version also follows:
for every edge in parallel, test whether removing it causes this least-significant digit $2^{2w(M)}$ in the determinant to disappear;
output those edges for which it does.

The fundamental claim in~\cite{MulmuleyVV87} is that assigning random weights to edges does indeed isolate one matching with high probability.
This is known as the Isolation Lemma and turns out to be true in the much more general setting of arbitrary set families:

\begin{lemma}[Isolation Lemma] \label{isolation_lemma}
Let $\cM \subseteq 2^E$ be any nonempty family of subsets of a universe $E = \{1, 2, ..., |E|\}$.
Suppose we define a weight function $w : E \to \{1, 2, ..., 2|E|\}$
by selecting each $w(e)$ for $e \in E$ independently and uniformly at random.
Then with probability at least $1/2$, there is a \emph{unique} set $M \in \cM$ which minimizes the weight $w(M) = \sum_{e \in M} w(e)$.
\end{lemma}

We call such a weight function $w$ \emph{isolating}.
We take $\cM$ in \cref{isolation_lemma} to be the set of all perfect matchings.

Since \cref{isolation_lemma} is the only randomized ingredient of the $\rnc$ algorithm,
a natural approach to showing that the perfect matching problem is in $\nc$
is the derandomization of the Isolation Lemma.
That is, we would like a set of polynomially many weight functions
(with polynomially bounded values)
which would be guaranteed to contain an isolating one.
To get an $\nc$ algorithm, we should be able to generate this set efficiently in parallel; then we can try all weight functions simultaneously.

However, derandomizing the Isolation Lemma turns out to be a challenging open question.
It has been done for certain classes of graphs:
strongly chordal~\cite{DahlhausK98},
planar bipartite~\cite{DattaKR10,TewariV12},
or
graphs with a small number of perfect matchings~\cite{GrigorievK87,AgrawalHT07}. 
	More generally, there has been much interest in obtaining $\nc$ algorithms for the perfect matching problem on
	restricted graph classes
	(not necessarily using the Isolation Lemma),
	e.g.:
	regular bipartite \cite{LevPV81},
	$P_4$-tidy \cite{Parfenoff98},
	dense \cite{DahlhausHK93},
	convex bipartite \cite{DekelS84},
	claw-free \cite{ChrobakNN89},
	incomparability graphs \cite{KozenVV85}.
The general set-family setting of the Isolation Lemma
is also related to circuit lower bounds
and polynomial identity testing~\cite{ArvindM08}.

Recently, in a major development, Fenner, Gurjar and Thierauf~\cite{FennerGT16}
have almost derandomized the Isolation Lemma for bipartite graphs.
Namely, they define a family of weight functions which can be computed obliviously
(only using the number of vertices $n$)
and prove that for any bipartite graph, one of these functions is isolating.
Because their family has quasi-polynomial size and the weights are quasi-polynomially large,
this has placed the perfect bipartite matching problem in the class $\quasinc$.

Nevertheless,
the
general-graph
setting of the derandomization question
(either using the Isolation Lemma or not)
remained open.
Even in the planar case,
with $\nc$ algorithms for bipartite planar and small-genus graphs having been known for a long time \cite{MillerN89,MahajanV00}, 
we knew no $\quasinc$ algorithm for non-bipartite graphs. \ifieee\else
	Curiously, an $\nc$ algorithm to count the number of perfect matchings in a planar graph is known \cite{Kasteleyn67,Vazirani89},
	which implies an algorithm for the decision version.
	However, it is open to give an $\nc$ algorithm for the search version.
\fi
In general,
the best known upper bound 
on the size of uniform circuits with polylogarithmic depth
was exponential.

We are able to
nearly bridge this gap in understanding.
The main result of our paper is the following:

\begin{theorem} \label{mainer}
For any number $n$, we can in $\quasinc$ construct $n^{O(\log^2 n)}$ weight functions on $\{1, 2, ..., {n \choose 2}\}$ with weights bounded by $n^{O(\log^2 n)}$
such that for any graph on $n$ vertices,
one of these weight functions isolates a perfect matching (if one exists).
\end{theorem}

The results of~\cite{MulmuleyVV87} and~\cref{mainer} together imply that the
perfect matching problem (both the decision and the search variant) in general graphs is in $\quasinc$. 
See \cref{sec:complexity} for more details on this.
We remark that the implied
algorithm is very simple.  The complexity lies in the analysis, i.e., proving
that one of the weight functions is isolating (see \cref{exists_isolating}).

In what
follows, we first give an overview of the framework in~\cite{FennerGT16} for bipartite
graphs.  We then explain how we extend the framework to general graphs. Due to the more complex structure of perfect matchings
in general graphs, we need several new ideas. In particular, we exploit structural properties of the perfect matching polytope.

\subsection{Isolation in bipartite graphs} \label{sec:bipartite}

In this section we shortly discuss the elegant framework introduced by Fenner, Gurjar and Thierauf~\cite{FennerGT16}, which we extend to obtain our result.

If a weight function $w$ is \emph{not} isolating,
then there exist two minimum-weight perfect matchings,
and their symmetric difference consists of alternating cycles.
In each such cycle, the total weight of edges from the first matching
must be equal to the total weight of edges from the second matching
(as otherwise we could obtain another matching of lower weight).
The difference between these two total weights is called the \emph{circulation} of the cycle.
By the above, if all cycles have nonzero circulation,
then $w$ is isolating.
It is known how to obtain weight functions which satisfy a polynomial number of such non-equalities (see \cref{lem23}). However, 
a graph may have an exponential number of cycles.

A key idea of~\cite{FennerGT16} is to build the weight function in $\log n$ rounds.
In the first round, we find a weight function with the property
that each cycle \emph{of length 4} has nonzero circulation.
This is possible since there are at most $n^4$ such cycles.
We apply this function and from now on consider only those edges
which belong to a minimum-weight perfect matching.
Crucially, it turns out that in the subgraph obtained this way,
all cycles of length $4$ have disappeared
-- this follows from the simple structure of the bipartite perfect matching polytope (a face is simply the bipartite matching polytope of a subgraph)
and fails to hold for general graphs.
In the second round, we start from this subgraph
and apply another weight function which ensures that
all even cycles of length up to $8$ have nonzero circulation
(one proves that there are again $\leq n^4$ many since the graph contains no $4$-cycles).
Again, these cycles disappear from the next subgraph, and so on.
After $\log n$ rounds, the current subgraph has no cycles, i.e., it is a perfect matching.
The final weight function is obtained by combining the $\log n$ polynomial-sized weight functions.
To get a parallel algorithm,
we need to simultaneously try each such possible combination,
of which there are quasi-polynomially many.

This result has later been generalized
by Gurjar and Thierauf~\cite{GurjarT16}
to the linear matroid intersection problem
-- a natural extension of bipartite matching.
From the work of Narayanan, Saran and Vazirani~\cite{NarayananSV94},
who gave an $\rnc$ algorithm for that problem
(also based on computing a determinant),
it again follows that
derandomizing the Isolation Lemma implies a $\quasinc$ algorithm.

\subsection{Challenges of non-bipartite graphs} \label{sec:challenges}

We find it useful to look at the method explained in the previous section
from a polyhedral perspective (also used by \cite{GurjarT16}).
We begin from the set of all perfect matchings, of which we take the convex hull: the perfect matching polytope.
After applying the first weight function, we want to consider only those perfect matchings which minimize the weight;
this is exactly the definition of a face of the polytope.
In the bipartite case, any face was characterized by just taking a subset of edges
(i.e., making certain constraints $x_e \ge 0$ tight),
so we could simply think about recursing on a smaller subgraph.
This was used to show that
any cycle whose circulation has been made nonzero
will not retain all of its edges in the next subgraph.
The progress we made
in the bipartite case 
could be measured by the girth (the minimum length of a cycle) of the current subgraph,
which doubled as we moved from face to subface.
\ifieee
\else
\ifieee
\begin{figure*}
\begin{center}
\else
\begin{figure}[t!]
\fi
  \begin{tikzpicture}
    \tikzstyle{vertex}=[circle, fill=black, minimum size=2,inner sep=1pt]

     \draw[rounded corners=5pt] (-0.5,-0.5) rectangle (8.7,1.5);
     \node at (4.2, -1) {Main difficulty};
    \node[vertex](u1) at (0,0) {};
    \node[vertex](u2) at (0.7,0.5) {};
    \node[vertex](u3) at (0,1) {};

    \node[vertex](v1) at (3,0) {};
    \node[vertex](v2) at (2.3,0.5) {};
    \node[vertex](v3) at (3,1) {};
    \draw (u1) edge[decorate,decoration={snake,amplitude=.3mm,segment length=4pt,post length=0mm}, thick] (v1);
    \draw (u2) edge[decorate,decoration={snake,amplitude=.3mm,segment length=4pt,post length=0mm}, thick] (v2);
    \draw (u3) edge (v3);
    \draw (u1) edge[thick] (u2) edge (u3);
    \draw (u2) edge (u3);
    \draw (v1) edge[thick] (v2) edge (v3);
    \draw (v2) edge (v3);

    \node at (3.9,0.75) {\tiny select $w$};
    \node at (3.9,0.5) {$\Longrightarrow$};
    \begin{scope}[xshift=5cm]
      \draw[fill = gray!30!white] (0.3, 0.5) ellipse (0.6cm and 0.85cm);
      \node[vertex](u1) at (0,0) {};
      \node[vertex](u2) at (0.7,0.5) {};
      \node[vertex](u3) at (0,1) {};

      \draw[fill = gray!30!white] (2.7, 0.5) ellipse (0.6cm and 0.85cm);
      \node[vertex](v1) at (3,0) {};
      \node[vertex](v2) at (2.3,0.5) {};
      \node[vertex](v3) at (3,1) {};
      \draw (u1) edge[decorate,decoration={snake,amplitude=.3mm,segment length=4pt,post length=0mm}, thick] node[fill=white,inner sep=1pt] {\scriptsize $1$} (v1);
      \draw (u2) edge[decorate,decoration={snake,amplitude=.3mm,segment length=4pt,post length=0mm}, thick] node[fill=white,inner sep=1pt] {\scriptsize $1$}(v2);
      \draw (u3) edge node[fill=white,inner sep=1pt] {\scriptsize $1$}(v3);
      \draw (u1) edge[thick] node[fill=gray!30!white,inner sep=1pt] {\scriptsize $0$}(u2) edge node[fill=gray!30!white, inner sep =1pt] {\scriptsize $0$}(u3);
      \draw (u2) edge node[fill=gray!30!white,inner sep=1pt] {\scriptsize $0$}(u3);
      \draw (v1) edge[thick] node[fill=gray!30!white,inner sep=1pt] {\scriptsize $0$}(v2) edge node[fill=gray!30!white, inner sep =1pt] {\scriptsize $0$}(v3);
      \draw (v2) edge node[fill=gray!30!white,inner sep=1pt] {\scriptsize $0$}(v3);
    \end{scope}
     \draw[rounded corners=5pt] (9.1,-0.5) rectangle (15.4,1.5);
     \node at (12.2, -1) {Minor difficulty};
    \begin{scope}[xshift=10cm]
      \node[vertex] (c1) at (0,1) {};
      \node[vertex] (c2) at (1,1) {};
      \node[vertex] (e1) at (-0.5,0.5) {};
      \node[vertex] (c3) at (1,0) {};
      \node[vertex] (c4) at (0,0) {};
      \node[vertex] (e2) at (1.5,0.5) {};
      \draw (c2) edge (c3);
      \draw (c1) edge[bend right=20] (c2);
      \draw (c3) edge[bend right=20] (c4);
      \draw (c1) edge[bend left=20,decorate,decoration={snake,amplitude=.3mm,segment length=3pt,post length=0mm}] (c2);
      \draw (c3) edge[bend left=20,decorate,decoration={snake,amplitude=.3mm,segment length=3pt,post length=0mm}] (c4);
      \draw (c4) edge (c1);
      \draw (e1) edge[decorate,decoration={snake,amplitude=.3mm,segment length=3pt,post length=0mm}] (c4) edge[decorate,decoration={snake,amplitude=.3mm,segment length=4pt,post length=0mm}] (c1);
      \draw (e2) edge[decorate,decoration={snake,amplitude=.3mm,segment length=3pt,post length=0mm}] (c3) edge[decorate,decoration={snake,amplitude=.3mm,segment length=4pt,post length=0mm}] (c2);
    \end{scope}
    \node at (12.3,0.75) {\tiny $C_1 \Delta C_2$};
    \node at (12.3,0.5) {$\Longrightarrow$};
    \begin{scope}[xshift=13.5cm]
      \node[vertex] (c1) at (0,1) {};
      \node[vertex] (c2) at (1,1) {};
      \node[vertex] (e1) at (-0.5,0.5) {};
      \node[vertex] (c3) at (1,0) {};
      \node[vertex] (c4) at (0,0) {};
      \node[vertex] (e2) at (1.5,0.5) {};
      \draw (c2) edge (c3);
      \draw (c4) edge (c1);
      \draw (e1) edge[decorate,decoration={snake,amplitude=.3mm,segment length=3pt,post length=0mm}] (c4) edge[decorate,decoration={snake,amplitude=.3mm,segment length=4pt,post length=0mm}] (c1);
      \draw (e2) edge[decorate,decoration={snake,amplitude=.3mm,segment length=3pt,post length=0mm}] (c3) edge[decorate,decoration={snake,amplitude=.3mm,segment length=4pt,post length=0mm}] (c2);
    \end{scope}
  \end{tikzpicture}
  \caption{An illustration of the difficulties of derandomizing the Isolation Lemma for general graphs as compared to bipartite graphs.
  \\
  On the left: in trying to remove the bold cycle, we select a weight function $w$ such that the circulation of the cycle is $1-0+1-0 \ne 0$.
  By minimizing over $w$ we obtain a new, smaller subface -- the convex hull of perfect matchings of weight $1$ -- but every edge of the cycle is still present in one of these matchings.
  The cycle has only been eliminated in the following sense:
  it can no longer be obtained in the symmetric difference of two matchings in the new face
  (since none of them select both swirly edges).
  The vertex sets drawn in gray
  represent the new tight odd-set constraints that describe the new face
  (indeed: for a matching to have weight $1$, it must take only one edge from the boundary of a gray set).
  We will say that the cycle does not \emph{respect} the gray vertex sets (see \cref{sec:alternating}).
  \\
  On the right: two even cycles whose symmetric difference contains no even cycle.
  }
  \label{fig:difficulty}
\ifieee
\end{center}
\end{figure*}
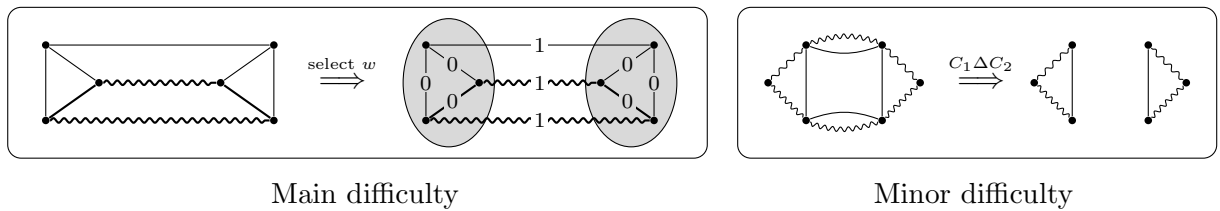
\else
\end{figure}
\fi
\fi
Unfortunately,
in the non-bipartite case,
the description of the perfect matching polytope is more involved
(see \cref{sec:pmpolytope}).
Namely, moving to a new subface
may also cause new tight \emph{odd-set constraints}
to appear.
These, also referred to as odd cut constraints, require that, for an odd set $S \subseteq V$ of vertices,
exactly one edge of a matching should cross the cut defined by $S$.
This complicates our task, as depicted in  the left part of \cref{fig:difficulty} 
(the same example was  given by~\cite{FennerGT16} to demonstrate the difficulty of the general-graph case).
Now a face is described by not only a subset of edges,
but also a family of tight odd-set constraints.
Thus we can no longer guarantee that
any cycle whose circulation has been made nonzero
will disappear from the support of the new face, i.e., the set of edges that appear in at least one perfect matching in this face.  
Our idea of what it means to remove a cycle thus needs to be refined
(see \cref{sec:alternating}),
as well as the measure of progress we use to prove that
a single matching is isolated after $\log n$ rounds
(see \cref{sec:lambda-goodness}).
We need several new ideas,
which we outline in \cref{sec:ourapproach}.

Another difficulty, of a more technical nature, concerns the counting argument
used to prove that a graph with no cycles of length at most $\lambda$
contains only polynomially many cycles of length at most $2 \lambda$.
In the bipartite case, the symmetric difference of two even cycles
contains a simple cycle, which is also even. In addition, one can show that if the two cycles share many vertices, then the symmetric difference must contain one such even cycle that is short (of length at most $\lambda$) and thus should not exist. 
This enables a simple checkpointing argument  to bound the number of cycles of length at most $2\lambda$, assuming that no cycle of length at most $\lambda$ exists.
Now, in the general case we are still only interested in removing \emph{even} cycles,
but the symmetric difference of two even cycles may not contain an even simple cycle (see the right part of \cref{fig:difficulty}).
This forces us to remove not only even simple cycles,
but all even walks, which may contain repeated edges
(we call these \emph{alternating circuits} -- see \cref{def:alternating_circuit}),
and to rework the counting scheme, obtaining a bound of $n^{17}$ rather than $n^4$%
\ifieee
\else
\space (see \cref{counting}).
\fi
Moreover, instead of simple graphs, we work on node-weighted multigraphs,
which arise by contracting certain tight odd-sets.

\subsection{Our approach} \label{sec:ourapproach}

This section is a high-level, idealized explanation of how to deal with the main difficulty (see the left part of \cref{fig:difficulty});
we ignore the more technical one in this description.

\paragraph{Removing cycles which do not cross a tight odd-set}
As discussed in~\cref{sec:challenges}, when moving from face to subface
we cannot guarantee that,
for each even cycle whose circulation we make nonzero,
one of its edges will be absent from the support of the new face.
However, this will at least be true for cycles that do not cross any odd-set tight for the new face.
	This is because if there are no tight odd-set constraints,
	then our faces behave as in the bipartite case.
	So, intuitively,
	if we only consider those cycles which do not cross any tight set,
	then we can remove them using the same arguments as in that case.
This implies, by the same argument as in~\cref{sec:bipartite}, that
if we apply $\log n$ weight functions in succession,
then the resulting face will not contain
in its support
any even cycle that crosses no tight odd-set.
This is less than we need, but a good first step.
If, at this point, there were no tight sets,
then we would be done, as we would have removed all cycles.
However, in general there will still be cycles crossing tight sets,
which make our task more difficult.

\paragraph{Decomposition into two subinstances}
To deal with the tight odd-sets,
we will make use of two crucial properties.
The first property is easy to see:
once we fix the single edge $e$ in the matching which crosses a tight set $S$,
the instance breaks up into two \emph{independent} subinstances.
	That is,
	every perfect matching
	which contains $e$
	is the union of:
	the edge $e$,
	a perfect matching on the vertex set $S$
	(ignoring the $S$-endpoint of $e$),
	and a perfect matching on the vertex set $V \setminus S$
	(ignoring the other endpoint of $e$).

This will allow us to employ a divide-and-conquer strategy:
to isolate a matching in the entire graph,
we will take care of both subinstances
and of the cut separating them.
We formulate the task of dealing with such a subinstance
(a subgraph induced on an odd-cardinality vertex set)
as follows:
we want that,
once the (only) edge of a matching which lies on the boundary of the tight odd-set is fixed,
the entire matching inside the set is uniquely determined.
We will then call this set \emph{contractible} (see \cref{def:contractable}).
This can be seen as a generalization of our isolation objective to subgraphs with an odd number of vertices.
If we can get that for the tight set and for its complement,
then each edge from the cut separating them
induces a unique perfect matching in the graph.
Therefore there are at most $n^2$ perfect matchings left in the current face.
Now, in order to isolate the entire graph,
we only need a weight function $w$
which assigns different weights to all these matchings.
This demand can be written as a system of $n^4$ linear non-equalities on $w$,
and we can generate a weight function $w$ satisfying all of them (see \cref{lem23}).

While it is not clear how to continue this scheme beyond the first level
or why we could hope to have a low depth of recursion,
we will soon explain how we utilize this basic strategy. 

\ifieee
\ifieee
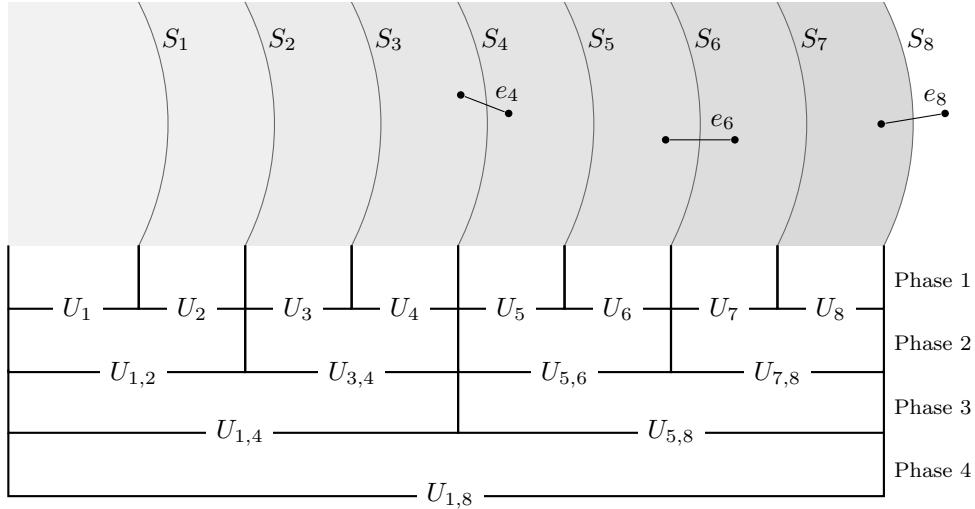
\begin{figure*}
\begin{center}
\else
\begin{figure}[t]
\fi
  \begin{center}
    \ifieee \begin{tikzpicture}[scale=0.6]
      \else 
      \begin{tikzpicture}[scale=0.7]
    \fi

        \tikzstyle{vertex}=[circle, fill=black, minimum size=2,inner sep=1pt]
        \tikzset{
          ncbar angle/.initial=-90,
          ncbar/.style={
              to path=(\tikztostart)
              -- ($(\tikztostart)!#1!\pgfkeysvalueof{/tikz/ncbar angle}:(\tikztotarget)$)
              -- ($(\tikztotarget)!($(\tikztostart)!#1!\pgfkeysvalueof{/tikz/ncbar angle}:(\tikztotarget)$)!\pgfkeysvalueof{/tikz/ncbar angle}:(\tikztostart)$)
              -- (\tikztotarget)
          },
         ncbar/.default=0.5cm,
        }

      \tikzset{square left brace/.style={ncbar=1.2cm}}
      \tikzset{square right brace/.style={ncbar=-0.5cm}}

      \begin{scope}
        \begin{scope}
          \clip(3,-2.3) rectangle (22.5,2.3);
          \draw[fill=gray!30!white, draw=gray!80!black] (14,0) ellipse (6cm and 5.5cm);
          \draw[fill=gray!27!white, draw=gray!80!black] (12,0) ellipse (6cm and 5.5cm);
          \draw[fill=gray!24!white, draw=gray!80!black] (10,0) ellipse (6cm and 5.5cm);
          \draw[fill=gray!21!white, draw=gray!80!black] (8,0) ellipse (6cm and 5.5cm);
          \draw[fill=gray!19!white, draw=gray!80!black] (6,0) ellipse (6cm and 5.5cm);
          \draw[fill=gray!16!white, draw=gray!80!black] (4,0) ellipse (6cm and 5.5cm);
          \draw[fill=gray!13!white, draw=gray!80!black] (2,0) ellipse (6cm and 5.5cm);
          \draw[fill=gray!10!white, draw=gray!80!black] (0,0) ellipse (6cm and 5.5cm);
        \end{scope}
        \foreach \val in {1,2,3,4,5,6,7,8} {
          \node at (\val*2 -2 + 6.15, 1.6) {\small $S_\val$};  
          \ifthenelse{\val = 1}{
            \draw [thick] (\val*2  +1,-2.30) to [square left brace ]  (\val*2  + 3.45,-2.30);
            \node[fill=white] at (\val*2 + 2.30, -3.5) {\small $U_\val$};
          }
          {
            \draw [thick] (\val*2  +1.45,-2.30) to [square left brace ]  (\val*2  + 3.45,-2.30);
            \node[fill=white] at (\val*2 + 2.45, -3.5) {\small $U_\val$};
          }
        }
        \begin{scope}[yshift=0.75cm]
        \foreach \val/\end in {1/2,3/4,5/6,7/8} {
          \ifthenelse{\val = 1}{
            \draw [thick] (\val*2  +1,-4.25) to [square left brace ]  (\val*2  + 5.45,-4.25);
            \node[fill=white] at (\val*2 + 3.35, -5.5) {\small $U_{\val, \end}$};
          }
          {
            \draw [thick] (\val*2  +1.45,-4.25) to [square left brace ]  (\val*2  + 5.45,-4.25);
            \node[fill=white] at (\val*2 + 3.45, -5.5) {\small $U_{\val, \end}$};
          }
        }
        \foreach \val/\end in {1/4,5/8} {
          \ifthenelse{\val = 1}{
            \draw [thick] (\val*2  +1,-5.4) to [square left brace ]  (\val*2  + 9.45,-5.4);
            \node[fill=white] at (\val*2 + 5.35, -6.6) {\small $U_{\val, \end}$};
          }
          {
            \draw [thick] (\val*2  +1.45,-5.4) to [square left brace ]  (\val*2  + 9.45,-5.4);
            \node[fill=white] at (\val*2 + 5.45, -6.6) {\small $U_{\val, \end}$};
          }
        }
        \draw [thick] (2  +1,-6.60) to [square left brace ]  (2  + 17.45,-6.60);
        \node[fill=white] at (2 + 9.3, -7.8) {\small $U_{1, 8}$};
        \ifieee 
          \node at (23.0, -3.7) {\begin{minipage}{4cm}\scriptsize  Phase 1\end{minipage}};
          \node at (23.0, -4.9) {\begin{minipage}{4cm}\scriptsize  Phase 2\end{minipage}};
          \node at (23.0, -6.1) {\begin{minipage}{4cm}\scriptsize  Phase 3\end{minipage}};
          \node at (23.0, -7.3) {\begin{minipage}{4cm}\scriptsize  Phase 4\end{minipage}};
        \else
          \node at (22.5, -3.7) {\begin{minipage}{4cm}\scriptsize  Phase 1\end{minipage}};
          \node at (22.5, -4.9) {\begin{minipage}{4cm}\scriptsize  Phase 2\end{minipage}};
          \node at (22.5, -6.1) {\begin{minipage}{4cm}\scriptsize  Phase 3\end{minipage}};
          \node at (22.5, -7.3) {\begin{minipage}{4cm}\scriptsize  Phase 4\end{minipage}};
        \fi 
      \end{scope}

      \node[vertex] (e4) at (11.50,0.55) { };  
      \node[vertex] (e4o) at (12.4,0.2) { };  
      \draw (e4) edge node[above right = -0.1cm and 0cm] {\small $e_{4}$} (e4o);
      \begin{scope}[xshift=4.0cm]
        \node[vertex] (e4) at (11.35,-0.3) { };  
        \node[vertex] (e4o) at (12.65,-0.3) { };  
        \draw (e4) edge node[above right = -0.0cm and 0cm] {\small $e_{6}$} (e4o);
      \end{scope}
      \begin{scope}[xshift=8.0cm,yshift=-0.6cm]
        \node[vertex] (e4) at (11.4,0.6) { };  
        \node[vertex] (e4o) at (12.6,0.8) { };  
        \draw (e4) edge node[above right = 0.05cm and 0cm] {\small $e_{8}$} (e4o);
      \end{scope}
      \end{scope}
      \node at (1, 0) {};
    \end{tikzpicture}
  \end{center}
  \caption{Example of a chain consisting of 8 tight sets, and our divide-and-conquer argument.}
  \label{fig:chain}
\ifieee
\end{center}
\end{figure*}
\else
\end{figure}
\fi
\else
\fi

\paragraph{Laminarity}
The second crucial property
that we utilize
is that
the family of odd-set constraints tight for a face
exhibits good structural properties.
Namely,
it is known that a \emph{laminar} family of odd sets is enough
to describe any face (see \cref{sec:pmpolytope}). Recall that a family of sets is  laminar if any two sets in the family are either disjoint or one is a subset of the other (see~\cref{fig:generalex} for an example).
This enables a scheme
where we use this family
to make progress in a bottom-up fashion.
This is still challenging
as the family
does not stay fixed
as we move from face to face.
The good news is that it can only increase:
whenever a new tight odd-set constraint appears
which is not spanned by the previous ones,
we can always add an odd-set to our laminar family.

\paragraph{Chain case}
To get started,
let us
first discuss the special case where
the laminar family of tight constraints is a chain,
i.e., an increasing sequence of odd-sets $S_1 \subsetneq S_2 \subsetneq ... \subsetneq S_{\ell}$.
\ifieee\else
We remark that this will be an informal and simplified description of the proof of \cref{making-2l-contractible}.
\fi
For this introduction, assume $\ell = 8$ as depicted in~\cref{fig:chain}.
Denote by $U_1, ..., U_{8}$ the \emph{layers} of this chain,
i.e., $U_1 = S_1$ and $U_p = S_p \setminus S_{p-1}$ for $p = 2, 3, ..., 8$.
Suppose this chain describes the face that was obtained by applying the $\log n$ weight functions as above that remove all even cycles that do not cross a tight set. Then there is no cycle that lies inside a single layer $U_p$.

\ifieee
\else
\ifieee
\begin{figure*}
\begin{center}
\else
\begin{figure}[t]
\fi
  \begin{center}
    \ifieee \begin{tikzpicture}[scale=0.6]
      \else 
      \begin{tikzpicture}[scale=0.7]
    \fi

        \tikzstyle{vertex}=[circle, fill=black, minimum size=2,inner sep=1pt]
        \tikzset{
          ncbar angle/.initial=-90,
          ncbar/.style={
              to path=(\tikztostart)
              -- ($(\tikztostart)!#1!\pgfkeysvalueof{/tikz/ncbar angle}:(\tikztotarget)$)
              -- ($(\tikztotarget)!($(\tikztostart)!#1!\pgfkeysvalueof{/tikz/ncbar angle}:(\tikztotarget)$)!\pgfkeysvalueof{/tikz/ncbar angle}:(\tikztostart)$)
              -- (\tikztotarget)
          },
         ncbar/.default=0.5cm,
        }

      \tikzset{square left brace/.style={ncbar=1.2cm}}
      \tikzset{square right brace/.style={ncbar=-0.5cm}}

      \begin{scope}
        \begin{scope}
          \clip(3,-2.3) rectangle (22.5,2.3);
          \draw[fill=gray!30!white, draw=gray!80!black] (14,0) ellipse (6cm and 5.5cm);
          \draw[fill=gray!27!white, draw=gray!80!black] (12,0) ellipse (6cm and 5.5cm);
          \draw[fill=gray!24!white, draw=gray!80!black] (10,0) ellipse (6cm and 5.5cm);
          \draw[fill=gray!21!white, draw=gray!80!black] (8,0) ellipse (6cm and 5.5cm);
          \draw[fill=gray!19!white, draw=gray!80!black] (6,0) ellipse (6cm and 5.5cm);
          \draw[fill=gray!16!white, draw=gray!80!black] (4,0) ellipse (6cm and 5.5cm);
          \draw[fill=gray!13!white, draw=gray!80!black] (2,0) ellipse (6cm and 5.5cm);
          \draw[fill=gray!10!white, draw=gray!80!black] (0,0) ellipse (6cm and 5.5cm);
        \end{scope}
        \foreach \val in {1,2,3,4,5,6,7,8} {
          \node at (\val*2 -2 + 6.15, 1.6) {\small $S_\val$};  
          \ifthenelse{\val = 1}{
            \draw [thick] (\val*2  +1,-2.30) to [square left brace ]  (\val*2  + 3.45,-2.30);
            \node[fill=white] at (\val*2 + 2.30, -3.5) {\small $U_\val$};
          }
          {
            \draw [thick] (\val*2  +1.45,-2.30) to [square left brace ]  (\val*2  + 3.45,-2.30);
            \node[fill=white] at (\val*2 + 2.45, -3.5) {\small $U_\val$};
          }
        }
        \begin{scope}[yshift=0.75cm]
        \foreach \val/\end in {1/2,3/4,5/6,7/8} {
          \ifthenelse{\val = 1}{
            \draw [thick] (\val*2  +1,-4.25) to [square left brace ]  (\val*2  + 5.45,-4.25);
            \node[fill=white] at (\val*2 + 3.35, -5.5) {\small $U_{\val, \end}$};
          }
          {
            \draw [thick] (\val*2  +1.45,-4.25) to [square left brace ]  (\val*2  + 5.45,-4.25);
            \node[fill=white] at (\val*2 + 3.45, -5.5) {\small $U_{\val, \end}$};
          }
        }
        \foreach \val/\end in {1/4,5/8} {
          \ifthenelse{\val = 1}{
            \draw [thick] (\val*2  +1,-5.4) to [square left brace ]  (\val*2  + 9.45,-5.4);
            \node[fill=white] at (\val*2 + 5.35, -6.6) {\small $U_{\val, \end}$};
          }
          {
            \draw [thick] (\val*2  +1.45,-5.4) to [square left brace ]  (\val*2  + 9.45,-5.4);
            \node[fill=white] at (\val*2 + 5.45, -6.6) {\small $U_{\val, \end}$};
          }
        }
        \draw [thick] (2  +1,-6.60) to [square left brace ]  (2  + 17.45,-6.60);
        \node[fill=white] at (2 + 9.3, -7.8) {\small $U_{1, 8}$};
        \ifieee 
          \node at (23.0, -3.7) {\begin{minipage}{4cm}\scriptsize  Phase 1\end{minipage}};
          \node at (23.0, -4.9) {\begin{minipage}{4cm}\scriptsize  Phase 2\end{minipage}};
          \node at (23.0, -6.1) {\begin{minipage}{4cm}\scriptsize  Phase 3\end{minipage}};
          \node at (23.0, -7.3) {\begin{minipage}{4cm}\scriptsize  Phase 4\end{minipage}};
        \else
          \node at (22.5, -3.7) {\begin{minipage}{4cm}\scriptsize  Phase 1\end{minipage}};
          \node at (22.5, -4.9) {\begin{minipage}{4cm}\scriptsize  Phase 2\end{minipage}};
          \node at (22.5, -6.1) {\begin{minipage}{4cm}\scriptsize  Phase 3\end{minipage}};
          \node at (22.5, -7.3) {\begin{minipage}{4cm}\scriptsize  Phase 4\end{minipage}};
        \fi 
      \end{scope}

      \node[vertex] (e4) at (11.50,0.55) { };  
      \node[vertex] (e4o) at (12.4,0.2) { };  
      \draw (e4) edge node[above right = -0.1cm and 0cm] {\small $e_{4}$} (e4o);
      \begin{scope}[xshift=4.0cm]
        \node[vertex] (e4) at (11.35,-0.3) { };  
        \node[vertex] (e4o) at (12.65,-0.3) { };  
        \draw (e4) edge node[above right = -0.0cm and 0cm] {\small $e_{6}$} (e4o);
      \end{scope}
      \begin{scope}[xshift=8.0cm,yshift=-0.6cm]
        \node[vertex] (e4) at (11.4,0.6) { };  
        \node[vertex] (e4o) at (12.6,0.8) { };  
        \draw (e4) edge node[above right = 0.05cm and 0cm] {\small $e_{8}$} (e4o);
      \end{scope}
      \end{scope}
      \node at (1, 0) {};
    \end{tikzpicture}
  \end{center}
  \caption{Example of a chain consisting of 8 tight sets, and our divide-and-conquer argument.}
  \label{fig:chain}
\ifieee
\end{center}
\end{figure*}
\else
\end{figure}
\fi
\fi
Notice that every layer $U_p$ is of even size
and it touches two boundaries of tight odd-sets: $S_{p-1}$ and $S_p$
(that is, $\delta(U_p) \subseteq \delta(S_{p-1}) \cup \delta(S_p)$).
Any perfect matching in the current face
will have one edge from $\delta(S_{p-1})$ and one edge from $\delta(S_p)$
(possibly the same edge),
therefore $U_p$ will have two (or zero) boundary edges in the matching.
An exception is $U_1$, which is odd, only touches $S_1$ and will have one boundary edge in the matching.
This motivates us to generalize our isolation objective to layers as follows:
we say that a layer $U_p$ is \emph{contractible}
if choosing an edge from $\delta(S_{p-1})$ and an edge from $\delta(S_p)$
induces a unique matching inside $U_p$.
This definition naturally extends to layers of the form
$S_r \setminus S_{p-1} = U_p \cup U_{p+1} \cup ... \cup U_r$,
which we will denote by $U_{p,r}$.

Recall that we have ensured that there is no cycle that lies inside a single layer $U_p = U_{p,p}$.
It follows that these layers are contractible.
	This is because two different matchings (but with the same boundary edges) in the current face
	would induce an alternating cycle in their symmetric difference.

Let us say that
this was the first phase of our approach (see \cref{fig:chain}).
In the second phase,
we want to ensure contractibility for double layers: $U_{1,2}$, $U_{3,4}$, $U_{5,6}$ and $U_{7,8}$.
In general,
we double our progress in each phase:
in the third one
we deal with the quadruple layers $U_{1,4}$ and $U_{5,8}$,
and in the fourth phase
we deal with the octuple layer $U_{1,8}$.

Let us now describe a single phase.
Take e.g. the layer $U_{5,8}$
and two boundary edges $e_{4} \in \delta(S_{4})$ and $e_8 \in \delta(S_8)$
(see \cref{fig:chain});
we want to have only a unique matching in $U_{5,8}$ including these edges.
Now we 
realize our divide-and-conquer approach.
Note that the layers $U_{5,6}$ and $U_{7,8}$ have already been dealt with
(made contractible)
in the previous phase.
Therefore, for each choice of boundary edge $e_6 \in \delta(S_6)$ for the matching,
there is a unique matching inside both of these layers.
Just like previously,
this implies that there are only $n^2$ matchings
using $e_4$ and $e_8$
in the layer $U_{5,8}$,
and we can select a weight function
that isolates one of them.
%
  We actually select only one function per phase,
	which works simultaneously for all layers $U_{p,r}$
	in this phase (here: $U_{1,4}$ and $U_{5,8}$)
	and all pairs of boundary edges
	$e_{p-1}$ and $e_r$.%

By generalizing this strategy
(from $\ell = 8$ to arbitrary $\ell$)
in the natural way,
we can deal with any chain in $\log \ell \le \log n$ phases,
even if it consists of $\Omega(n)$ tight sets.\ifieee\else \space
	We remark that, in the general proof, we do not quite use a binary tree structure like in the example.
	Instead, in the $t$-th phase,
	we deal with all layers $U_{p,r}$
	having $1 \le p \le r \le \ell$ with $r - p \le 2^{t-1} - 1$.
	This makes our proof simpler if $\ell$ is not a power of two.
 \fi

\paragraph{General case}
\begin{figure}[t]
  \begin{center}
  
    \ifieee\else \newcommand{\myscaleforfigure}{0.7} \fi
	\tikzset{external/export next=false} 
	
    \begin{tikzpicture}[scale=\myscaleforfigure]
    \tikzstyle{vertex}=[circle, fill=black, minimum size=2,inner sep=1pt]

      \draw[fill=white,very thick] (4.0,1.7) ellipse (8.5cm and 4.7cm); 
    \begin{scope}
      \draw[fill=gray!00!white,very thick] (0.0,1.7) ellipse (4.0cm and 3.1cm); 

        \draw[fill=gray!60!white] (4.0,-1.5) ellipse (1.3cm and 0.9cm);

      \begin{scope}[rotate=-45, xshift=-2.0cm, yshift=-1.5cm]
        \draw[fill=gray!60!white] (-1.25,1.5) ellipse (0.5cm and 0.9cm);
      \end{scope}

      \begin{scope}[xshift=1.5cm,yshift=1.5cm]
        \draw[fill=gray!60!white] (0,0.25) ellipse (0.7cm and 0.4cm);
      \end{scope}
      \begin{scope}[xshift=-0.2cm,yshift=0.5cm, rotate=30]
        \draw[fill=gray!60!white] (0,0.25) ellipse (0.7cm and 0.4cm);
      \end{scope}
      \begin{scope}[xshift=-1.8cm,yshift=-0.5cm, rotate=30,scale=0.4]
        \draw[fill=gray!60!white] (0,0.25) ellipse (0.6cm and 0.4cm);
      \end{scope}
      \begin{scope}[xshift=1.9cm,yshift=-0.0cm, rotate=50,scale=0.4]
        \draw[fill=gray!60!white] (0,0.25) ellipse (0.6cm and 0.4cm);
      \end{scope}
      \begin{scope}[xshift=1.3cm,yshift=3.0cm, rotate=50,scale=0.7]
        \draw[fill=gray!60!white] (0,0.25) ellipse (0.4cm and 0.7cm);
      \end{scope}

      \begin{scope}[xshift=-0.9cm,yshift=3.8cm, rotate=70,scale=0.8]
        \draw[fill=gray!60!white] (0,0.25) ellipse (0.6cm and 0.4cm);
      \end{scope}
    \end{scope}

    \begin{scope}[xshift=8.5cm]
      \draw[fill=gray!25!white,dashed] (0.0,1.7) ellipse (3.2cm and 2.2cm); 
      \draw[fill=gray!25!white,dashed] (-0.6,1.7) ellipse (2.6cm and 1.9cm); 
      \draw[fill=gray!25!white,dashed] (-1.3,1.7) ellipse (1.9cm and 1.6cm); 

      \begin{scope}[rotate=-45, xshift=-1.4cm, yshift=-1.5cm]
        \draw[fill=gray!60!white] (-1.25,1.5) ellipse (0.5cm and 0.9cm);
      \end{scope}

      \begin{scope}[xshift=-1.8cm,yshift=0.5cm, rotate=30,scale=0.4]
        \draw[fill=gray!60!white] (0,0.25) ellipse (0.6cm and 0.4cm);
      \end{scope}
      \begin{scope}[xshift=1.2cm,yshift=0.7cm, rotate=40,scale=0.4]
        \draw[fill=gray!60!white] (0,0.25) ellipse (0.6cm and 0.4cm);
      \end{scope}

      \begin{scope}[xshift=2.1cm,yshift=0.3cm, rotate=50,scale=0.4]
        \draw[fill=gray!60!white] (0,0.25) ellipse (0.6cm and 0.4cm);
        \draw[fill=gray!60!white] (4,1.85) ellipse (0.3cm and 0.4cm);
        \draw[fill=gray!60!white] (4.5,3.25) ellipse (0.4cm and 0.3cm);
        \draw[fill=gray!60!white] (3.0,0.75) ellipse (0.2cm and 0.3cm);
      \end{scope}
      \begin{scope}[xshift=1.3cm,yshift=2.0cm, rotate=50,scale=0.7]
        \draw[fill=gray!60!white] (0,0.25) ellipse (0.4cm and 0.7cm);
      \end{scope}

      \begin{scope}[xshift=-0.2cm,yshift=1.8cm, rotate=70,scale=0.8]
        \draw[fill=gray!60!white] (0,0.25) ellipse (0.6cm and 0.4cm);
      \end{scope}
    \end{scope}
    \end{tikzpicture}
  \end{center}
  \caption{
  Example of a general laminar family.\\
  \textbf{Dark-gray sets}
  are of size at most $\lambda$ and thus contractible.\\
  \textbf{Dashed sets}
  are of size more than $\lambda$ but at most $2 \lambda$;
  they must form chains (due to the cardinality constraints).
  We make them contractible in the first step.
  Then we contract them
  (so now all light-gray and dark-gray sets are contracted).\\
  \textbf{Thick sets}
  are of size more than $2 \lambda$.
  For the second step, we erase the edges on their boundaries.
  Then we remove cycles of length up to $2 \lambda$ from the resulting instance (the \emph{contraction}),
  which has no tight odd-sets (and no cycles of length up to $\lambda$).
  }
  \label{fig:generalex}
\end{figure}
Of course, there is no reason to expect that the laminar family
of tight cuts
we obtain after applying the initial $\log n$ weight functions
will be a chain.
It also does not seem easy to directly generalize our inductive scheme
from a chain to an arbitrary family.
Therefore we put forth a different progress measure,
which allows us to make headway
even in the absence of such a favorable odd-set structure.

Since a laminar family can be represented as a tree,
we might think about a bottom-up strategy based on it;
however, we cannot deal with its nodes level-by-level,
since it may have height $\Omega(n)$
and we can only afford $\poly(\log n)$ many phases.
Instead, we will first deal with all tight odd-sets of size up to $4$,
then up to $8$,
then up to $16$
and so on,
by making them contractible.
At the same time,
we also remove all even cycles of length up to $4$,
then up to $8$
and so on.\ifieee\else\footnote{
	As discussed in \cref{sec:challenges,fig:difficulty},
	the meaning of the term \emph{remove}
	needs to be refined,
	as we cannot hope to always delete an edge of the cycle
	from the support of the current face.
} \fi
These two components of our progress measure,
which we call $\lambda$-goodness,
are mutually beneficial,
as we will see below.

Making odd-sets contractible enables us 
not only to achieve progress,
but also to simplify our setting.
A contractible tight set
can be, for our purposes,
thought of as a single vertex
-- much like a blossom in Edmonds' algorithm.
This is because such a set has exactly one boundary edge in a perfect matching
(as does a vertex),
and choosing that edge determines the matching in the interior.
As the name suggests, we will contract such sets.

Suppose that our current face is already $\lambda$-good.
Roughly, this means that we have
made odd-sets of size up to $\lambda$ (which we will call small) contractible
and
removed cycles of length up to $\lambda$.
Now we want to obtain a face which is $2 \lambda$-good.

The first step is to make odd-sets of size up to $2 \lambda$ contractible.
Let us zoom in on one such odd-set -- a maximal set of size at most $2 \lambda$
(see the largest dashed set in~\cref{fig:generalex}).
Once we have contracted all the small sets into single vertices,
all interesting sets are now of size more than $\lambda$ but at most $2 \lambda$,
and any laminar family consisting of such sets must be a chain,
since a set of such size cannot have two disjoint subsets of such size
(see \cref{fig:generalex}).
But this is the chain case that we have already solved!

Having made odd-sets of size up to $2 \lambda$ contractible,
we can contract them.
The second step is now to remove cycles of length up to $2 \lambda$.
However, here we do not need to care about those cycles which cross an odd-set $S$ of size larger than $2 \lambda$
-- the reason being,
intuitively,
that in our technical arguments
we define the length of a cycle based on the sizes of sets that it crosses,
and thus such a cycle actually becomes longer than $2 \lambda$.
In other words,
we can think about removing cycles
of length up to $2 \lambda$
from a version of the input graph where
all small odd-sets have been contracted
and
all larger ones have had their boundaries erased (see \cref{fig:generalex}).
We call this version the \emph{contraction} (see \cref{def:contraction}).
Our $\lambda$-goodness progress measure (see \cref{def:face-laminar})
is actually defined in terms of cycles in the contraction.

Now the second step is easy:
we just need to remove all cycles of length up to $2 \lambda$
from the contraction,
which has no tight odd-sets
and no cycles of length up to $\lambda$
--
a simple scenario, already known from the bipartite case.
Applying one weight function is enough to do this.

Finally, what does it mean for us to remove a cycle?
When we make a cycle's circulation nonzero,
it is then eliminated from the new face in the following sense:
either one of its edges disappears from the support of the face
(recall that this is what always happened in the bipartite case),
or a new tight odd-set appears, with the following property:
the cycle crosses the set with fewer (or more) even-indexed edges than odd-indexed edges
(see the example in~\cref{fig:difficulty}).
In short, we say that the cycle does not \emph{respect} the new face (see \cref{sec:alternating}).
This notion of removal makes sense when viewed in tandem with the contraction,
because once a cycle crosses a set in the laminar family,
there are two possibilities in each phase:
either this set is large -- then its boundary is not present in the contraction,
which cancels the cycle,
or it is small -- then
it is contracted and the cycle also disappears
(for somewhat more technical reasons).

To reiterate,
our strategy is to simultaneously remove cycles up to a given length
and make odd-sets up to a given size contractible.
We can do this in $\log n$ phases.
In each such phase we need to apply a sequence of $\log n$ weight functions
in order to deal with a chain of tight odd-sets (as outlined above).
In all, we are able to isolate a perfect matching in the entire graph
using a sequence $O(\log^2 n)$ weight functions
with polynomially bounded weights.

\subsection{Future work}

The most immediate open problem left by our work
is to get down from $\quasinc$ to $\nc$ for the perfect matching problem.
Even for the bipartite case,
this will require new insights or methods,
as it is not clear how we could e.g. reduce the number of weight functions from $\log n$ to only a constant.

Proving that
the search version of the perfect matching problem
in planar graphs
is in $\nc$
is also open.
While the $\quasinc$ result of \cite{FennerGT16}
gives rise to a new $\nc$ algorithm for \emph{bipartite} planar graphs,
which proceeds by verifying at each step whether the chosen weight function has removed the wanted cycles
(it computes the girth of the support of the current face in $\nc$),
our $\lambda$-goodness progress measure seems to be difficult to verify in $\nc$.

A related problem which has resisted derandomization so far is exact matching \cite{PapadimitriouY82}.
Here we are given a graph, whose some edges are colored red, and an integer $k$;
the question is to find a perfect matching containing exactly $k$ red edges.
The problem is in $\rnc$ \cite{MulmuleyVV87},
but not known to even be in ${\sf P}$.

Finally, our polyhedral approach motivates the question of
what other zero-one polytopes
admit such a derandomization of
the Isolation Lemma.
One class that comes to mind
are totally unimodular polyhedra.
We remark that,
for that class,
this question has been resolved~\cite{GurjarTV17}
subsequent to the first version of this paper.

\subsection{Outline}

The rest of the paper is organized as follows.
In \cref{sec:preliminaries} we introduce notation and define basic notions related to the perfect matching polytope and to the weight functions that we use.
In \cref{sec:alternating} we define alternating circuits (our generalization of alternating cycles), discuss what it means for such a circuit to respect a face, and develop our tools for circuit removal.
In \cref{sec:lambda-goodness} we introduce our measure of progress ($\lambda$-goodness), contractible sets and the contraction multigraph.
We also state \cref{exists_isolating}, which implies our main result.
\ifieee
\else
	Finally, in~\cref{proof_of_main} we prove
\fi
our key technical theorem: that applying $\log_2 n + 1$ weight functions allows us to make progress from $\lambda$-good to $2 \lambda$-good.

\section{Preliminaries} \label{sec:preliminaries}
Throughout the paper we consider a fixed graph $G=(V,E)$ with $n$ vertices.
We remark that the isolating weight functions whose existence we prove can be generated without knowledge of the graph.
For notational convenience, we assume that $\log_2 n$ evaluates to an integer; otherwise simply replace $\log_2 n$ by $\lceil \log_2 n\rceil$.
We also assume that $n$ is sufficiently large.

We use the following notation.
For a subset $S \subseteq V$ of the vertices, let $\delta(S) = \{e\in E : |e\cap S| = 1\}$ denote the edges crossing the cut $(S, V\setminus S)$ and $E(S) = \{ e \in E : |e \cap S| = 2\}$ denote the edges inside $S$.
We shorten $\delta(\{v\})$ to $\delta(v)$ for $v \in V$. 
For a vector $(x_e)_{e\in E} \in \mathbb{R}^{|E|}$, we  define $x(\delta(S)) = \sum_{e\in \delta(S)} x_e$, as well as $\supp(x) = \{e \in E : x_e > 0 \}$.
For a subset $X \subseteq E$ we define $\one_X$ to be the vector with $1$ on coordinates in $X$ and $0$ elsewhere.
We again shorten $\one_{\{e\}}$ to $\one_e$ for $e \in E$.
Sometimes we identify matchings $M$ with their indicator vectors $\one_M$.

A matching is a set of edges $M \subseteq E$ such that no two edges in $M$ share an endpoint.
A matching $M$ is perfect if $|M| = \frac{n}{2}$.

\subsection{Parallel complexity}
\label{sec:complexity}

The complexity class $\quasinc$ is defined as $\quasinc = \bigcup_{k \ge 0} \quasinc^k$, where $\quasinc^k$ is the class of problems having uniform circuits of quasi-polynomial size $2^{\log^{O(1)} n}$ and polylogarithmic depth $O(\log^k n)$ \cite{Barrington92}.
Here by ``uniform'' we mean that the circuit can be generated in polylogarithmic space.

By the results of~\cite{MulmuleyVV87},
\cref{mainer} implies that the perfect matching problem
(both the decision and the search variant)
in general graphs is in $\quasinc$.
The same can be said about maximum cardinality matching,
as well as minimum-cost perfect matching for small costs (given in unary);
see Section~5 of~\cite{MulmuleyVV87}.

Some care is required to obtain our postulated running time,
i.e., that the perfect matching problem has
uniform circuits of size $n^{O(\log^2 n)}$ and depth $O(\log^3 n)$.
We could get a $\quasinc^4$ algorithm
by applying the results of \cite[Section~6.1]{MahajanV97}
to compute the determinant(s).
To shave off one $\log n$ factor,
we use the following Chinese remaindering method,
pointed out to us by Rohit Gurjar
(it will also appear in the full version of~\cite{FennerGT16}).
We first compute determinants modulo small primes;
since the determinant has $2^{O(\log^3 n)}$ bits,
we need as many primes (each of $O(\log^3 n)$ bits).
For one prime this can be done in $\nc^2$ \cite{Berkowitz1984}.
Then we reconstruct the true value from the remainders.
Doing this
for an $n$-bits result
would be in $\nc^1$ \cite{BeameCH86},
and thus for a result with $2^{O(\log^3 n)}$ bits it is in $\quasinc^3$.

\subsection{Perfect matching polytope} \label{sec:pmpolytope}

Edmonds~\cite{Edm65} showed that the following set of equalities and
inequalities on the variables $(x_e)_{e\in E}$ determines the perfect matching
polytope (i.e.,  the convex hull of indicator vectors of all perfect
matchings): 
\begin{equation*}
\begin{aligned}
\arraycolsep=1.4pt\def\arraystretch{1.2}
\begin{array}{rll}
 x(\delta(v)) = & 1 & \quad \text{for $v \in V$}, \\  
 x(\delta(S)) \geq & 1 & \quad \text{for $S \subseteq V$ with $|S|$ odd,} \\
x_e \geq & 0 & \quad \text{for $e \in E$.}
\end{array}
\end{aligned}
\end{equation*}

Note that the constraints imply that $x_e \le 1$ for any $e \in E$.
We refer to the perfect matching polytope of the graph $G=(V,E)$ by $\PM(V,E)$
or simply by $\PM$.  Our approach
exploits the special structure of faces of the perfect matching polytope.
Recall that a face of a polytope is obtained by setting a subset of the inequalities
to equalities.
We follow the definition of a face from the book of Schrijver \cite{Schrijver03} -- in particular, every face is nonempty.

Throughout the paper, we will only consider the perfect matching polytope and so the term ``face'' will always refer to a face of $\PM$.
We sometimes abuse notation and say that a perfect matching $M$ is in a face $F$ if its indicator vector is in $F$. 
When talking about faces, we also use the following notation:
\begin{definition}
  For a face $F$ we define
  \ifieee
\else
  \begin{align*}
    \suppo(F) = \{ e \in E : (\exists x \in F) \ x_e > 0 \} \quad \mbox{and} \quad \tight(F) = \{S \subseteq V : |S|\mbox{ odd and  $(\forall x \in F)\  x(\delta(S)) = 1$}\}\,.
  \end{align*}
  \fi
  In other words, $\suppo(F)$ contains the edges that appear in a perfect matching in $F$ and $\tight(F)$ contains the tight cut constraints of $F$.
\end{definition}

Notice that if a set is tight for a face, then it is also tight for any of its subfaces.

Standard uncrossing techniques imply that faces
can be defined using laminar families of tight constraints. This is proved using
\cref{face_structure} below, which is also  useful in our approach.

Two
subsets  $S,T
\subseteq V$  of vertices are said to be crossing if they intersect and none is
contained in the other, i.e., $S\cap T, S\setminus T, T\setminus S \neq
\emptyset$. A family $\cL$ of subsets 
of vertices is \emph{laminar} if no two sets $S, T \in \cL$ are crossing. Furthermore, we say that $\cL$ is a \emph{maximal laminar subset} of a family $\cS$ if no set in $\cS \setminus \cL$ can be added to $\cL$ while maintaining laminarity.

Note that any single-vertex set is tight for any face, and therefore a maximal laminar family contains all these sets. The laminar families in our arguments will always contain all singletons.

The following lemma is known.
\ifieee
\else
For completeness,  its proof is included
in~\cref{proof_of_face_structure}.
\fi

\begin{lemma} \label{face_structure}
  Consider a face $F$. For any maximal laminar subset $\cL$ of $\tight(F)$ we have
    \begin{align*}
      \Span(\cL) = \Span(\tight(F))\,,
    \end{align*}
    where for a subset $\cT \subseteq \tight(F)$, $\Span(\cT)$ denotes the
    linear subspace of $\bR^{E}$ spanned by the boundaries of sets in $\cT$,
    i.e., $\Span(\cT) = \Span \{ \one_{\delta(S)} : S \in \cT \}$.
\end{lemma}

Intuitively, \cref{face_structure} implies that a maximal laminar family $\cL$ of $\tight(F)$ is enough to describe a face $F$ (together with the edge set $\suppo(F)$).
Furthermore, given a subface $F' \subseteq F$, we can extend $\cL$ to a larger laminar family $\cL' \supseteq \cL$ which describes $F'$.

As the perfect matching polytope $\PM$ is defined as the convex hull of the indicator vectors of all perfect matchings, it is an integral polytope. In particular, it follows that every face of $\PM$ is also integral.

\subsection{Weight functions} \label{weight_functions}

For our derandomization of the Isolation Lemma
we will use families of weight functions which are possible to generate obliviously, i.e., by only using the number of vertices in $G$.
We define them below.

\newcommand{\largenumber}{(4n^2+1)}
\begin{definition} \label{def:W}
Given $t \ge 7$, we define the family of weight functions $\cW(t)$ as follows.
Number the edge set $E = \{ e_1, ..., e_{|E|} \}$ arbitrarily.
Let $w_k : E \to \bZ$ be given by $w_k(e_j) = \largenumber^j \mymod k$ for $j = 1, ..., |E|$ and $k = 2, ..., t$.
We define $\cW(t) = \{ w_k : k = 2, ..., t \}$.
\end{definition}

For brevity, we write $\cW := \cW(n^{20})$.

In our argument we will obtain a decreasing sequence of faces.
Each face arises from the previous by minimizing over a linear objective (given by a weight function).

\begin{definition}
Let $F$ be a face and $w$ a weight function. The subface of $F$ minimizing $w$ will be called \facemin{F}{w}:
\[ \facemin{F}{w} := \argmin \{ \ab{w,x} : x \in F \}. \]
\end{definition}

Instead of minimizing over one weight function and then over another,
we can \emph{concatenate} them in such a way that minimizing over the concatenation yields the same subface.
In particular, we will argue that one just needs to try all possible concatenations of $O(\log^2 n)$ weight functions from $\cW$
in order
to find one which isolates a unique perfect matching in $G$ (i.e., it produces a single extreme point as the minimizing subface).

\begin{definition} \label{def:concatenation}
For two weight functions $w$ and $w'$, where $w : E \to \bZ$ and $w' \in \cW$, we define their \emph{concatenation} $w \circ w' := n^{21} w + w'$, i.e.,
\[ \rb{w \circ w'}(e) := n^{21} \cdot w(e) + w'(e). \]
We also define $\cW^k$ to be the set of all concatenations of $k$ weight functions from $\cW$, i.e.,
\[ \cW^k := \{ w_1 \circ w_2 \circ ... \circ w_k : w_1, w_2, ..., w_k \in \cW \}\ifieee\else \,, \fi \]
\ifieee\else where by $w_1 \circ w_2 \circ ... \circ w_k$ we mean $\rb{\rb{w_1 \circ w_2} \circ ...} \circ w_k$.\fi
\end{definition}

\begin{fact} \label{concatenation_and_subface}
We have $\facemin{\facemin{F}{w}}{w'} = \facemin{F}{w \circ w'}$.
\end{fact}
\ifieee
\else
\begin{proof}
Both faces are integral and so we only need to show that $\facemin{\facemin{F}{w}}{w'} \cap \mathbb{Z}^E = \facemin{F}{w \circ w'} \cap \mathbb{Z}^E$.
The first  set consists of matchings in $F$ minimizing $w$ and, among such matchings, minimizing $w'$.
The second set consists of matchings in $F$ minimizing $w \circ w'$.
These two sets are equal because for any $M$:
\begin{itemize}
	\item $(w \circ w')(M) = n^{21} \cdot w(M) + w'(M)$,
	\item $w' \in \cW = \cW(n^{20})$ implies that $0 \le w'(M) < n^{20} \cdot \frac{n}{2} < n^{21}$ for any matching $M$,
	\item $w(M) \in \mathbb{Z}$, so that for any two matchings $M_1$ and $M_2$, $w(M_1) > w(M_2)$ implies $(w \circ w')(M_1) - (w \circ w')(M_2) = n^{21} \rb{w(M_1) - w(M_2)} + \rb{w'(M_1) - w'(M_2)} > 0$.
\end{itemize}
Hence, the ordering given by $w \circ w'$ is the same as the lexicographic ordering given by $(w,w')$.
\end{proof}
\fi

\section{Alternating circuits and respecting a face} \label{sec:alternating}

In this section we introduce two notions which are vital for our approach.
Before giving the formal definitions, we give an informal motivation. 

Our argument is centered around \emph{removing} even cycles.
As discussed in \cref{sec:challenges,fig:difficulty},
the meaning of this term in the non-bipartite case
needs to be more subtle than just ``removing an edge of the cycle''.

In order to deal with a cycle,
we find a weight function $w$ which assigns it a nonzero circulation.
Formally,
given an even cycle $C$ with edges numbered in order,
define a vector $\pmind{C} \in \{-1,0,1\}^E$ as having $1$ on even-numbered edges of $C$, $-1$ on odd-numbered edges of $C$, and $0$ elsewhere.
Then, nonzero circulation means that $\ab{\pmind{C}, w} \ne 0$.
Now, in the bipartite case,
if such a cycle survived in the new face $\facemin{F}{w}$,
that is, $C \subseteq \suppo(\facemin{F}{w})$,
then the vector $\pmind{C}$ could be used to
obtain a point in the face $F$
with lower $w$-weight
than the points in $\facemin{F}{w}$,
a contradiction.
This argument is possible because of the simple structure of the bipartite perfect matching polytope.

In the non-bipartite case, it is not enough that $C \subseteq \suppo(\facemin{F}{w})$ in order to obtain such a point (and a contradiction).
It is also required that, if the cycle $C$ enters a tight odd-set $S$ on an even-numbered edge,
it exits it on an odd-numbered edge (and vice versa).
This makes intuitive sense:
if $C$ were obtained from the symmetric difference of two perfect matchings
which both have exactly one edge crossing $S$,
then $C$ would have this property.
Formally, we require that $\ab{\pmind{C}, \one_{\delta(S)}} = 0$ for each $S \in \tight(\facemin{F}{w})$.
If $C$ satisfies these two conditions, i.e.,
that $C \subseteq \suppo(\facemin{F}{w})$
and that
 $\ab{\pmind{C}, \one_{\delta(S)}} = 0$ for every $S \in \tight(\facemin{F}{w})$,
then we say that $C$ \emph{respects} the face $\facemin{F}{w}$.
The notion of respecting a face  exactly formalizes what is required to obtain a contradictory point as above (see the proof of \cref{disc_no_respect}).

In other words,  if we assign a nonzero circulation to a cycle,
then it will not respect the new face,
and this is what is now meant by removing a cycle.

To deal with the second, more technical difficulty discussed in \cref{sec:challenges},
we need to remove not only simple cycles of even length,
but also walks with repeated edges.
However, we would run into problems
if we allowed all such walks (up to a given length).
Consider for example a walk $C$ of length $2$;
such a walk traverses an edge back and forth.
It is impossible to assign a nonzero circulation to $C$,
because its vector $\pmind{C}$ is zero.
We overcome this technicality by defining alternating circuits to be those even walks whose vector $\pmind{C}$ is nonzero (see Figure~\ref{fig:altcircuit} for an example).
For generality, we also
formulate the definition of respect in terms of the vector $\pmind{C}$.

\begin{figure}[t]
  \begin{center}
    \begin{tikzpicture}
      \tikzstyle{vertex}=[circle, fill=black, minimum size=4,inner sep=1pt]
      \node[vertex](u1) at (0.7,0) {};
      \node[vertex](u2) at (0.7,2) {};
      \node[vertex](u3) at (-0.7,1) {};

      \draw[fill = gray!30!white] (3.1, 1) ellipse (1.2cm and 0.85cm);
      \node[vertex](v2) at (2.3,1) {};
      \node[vertex](v3) at (4,1) {};
      \draw (u1) edge[thick] node[fill=white,inner sep=2pt] {\scriptsize $e_4$} (v2);
      \draw (u2) edge[decorate,decoration={snake,amplitude=.3mm,segment length=4pt,post length=0mm}, thick] node[fill=white,inner sep=2pt] {\scriptsize $e_1$}(v2);
      \draw (u1)  edge[decorate,decoration={snake,amplitude=.3mm,segment length=4pt,post length=0mm}, thick] node[fill=white, inner sep=2pt] {\scriptsize $e_3$}(u3);
      \draw (u2) edge[thick] node[fill=white,inner sep=2pt] {\scriptsize $e_2$}(u3);
      \draw (v2) edge[thick, bend left=20] node[fill=gray!30!white,inner sep=2pt] {\scriptsize $e_0$}(v3);
      \draw (v2) edge[decorate,decoration={snake,amplitude=.3mm,segment length=4pt,post length=0mm}, thick, bend right=20] node[fill=gray!30!white,inner sep=2pt] {\scriptsize $e_5$}(v3);
    \end{tikzpicture}
  \end{center}
  \caption{An example of an alternating circuit $C$ of length $6$ with indicator vector $\pmind{C} = \sum_{i=0}^{5} (-1)^i \one_{e_i} = - \one_{e_1} + \one_{e_2} - \one_{e_3} + \one_{e_4}$ (since $\one_{e_0}$ and $\one_{e_5}$ cancel each other).  Also note that $\ab{\pmind{C}, \one_{\delta(S)}} = 0$ for the tight set $S$ depicted in gray.}
  \label{fig:altcircuit}
\end{figure}

\begin{definition} \label{def:alternating_circuit}
Let $C = (e_0, ..., e_{k-1})$ be a nonempty cyclic walk of even length $k$.
\begin{itemize}
  \item We define the \emph{alternating indicator vector $\pmind{C}$ of $C$} to be $\pmind{C} = \sum_{i=0}^{k-1} (-1)^i \one_{e_i}$, where $\one_e \in \mathbb{R}^E$ is the indicator vector having $1$ on position $e$ and $0$ elsewhere.
  \item We say that $C$ is an \emph{alternating circuit} if its alternating
    indicator vector is nonzero. We also refer to $C$ as an \emph{alternating
    (simple) cycle} if it is an alternating circuit that visits every vertex at most once. 
	\item When talking about a graph with node-weights, the node-weight of an alternating circuit is the sum of all node-weights of visited vertices (with multiplicities if visited multiple times).
\end{itemize}
\end{definition}
\ifieee\else We remark that $\pmind{C}$ does not need to have all entries $-1$, $0$ or $1$ since edges can repeat in $C$. \fi

\begin{definition} \label{def:respect}
We say that a vector $y \in \bZ^E$ \emph{respects} a face $F$ if:
\begin{itemize}
	\item $\supp(y) \subseteq \suppo(F)$, and
	\item for each $S \in \tight(F)$ we have $\ab{y, \one_{\delta(S)}} = 0$.
\end{itemize}
	Furthermore, we say that an alternating circuit $C$ respects a face $F$ if its alternating indicator vector $\pmind{C}$ respects $F$.
\end{definition}

Clearly, if $F' \subseteq F$ are faces and a vector respects $F'$, then it also respects $F$.

Now we argue that
we can remove an alternating circuit
by assigning it a nonzero circulation.
The proof of this lemma (which generalizes Lemma 3.2 of \cite{FennerGT16}) motivates \cref{def:respect}.

\begin{lemma} \label{disc_no_respect}
Let $y \in \bZ^E$ be a vector
and $F$ a face.
If $w : E \to \bR$ is such that $\ab{y, w} \ne 0$, then
$y$ does not respect
the face $F' = \facemin{F}{w}$.
\end{lemma}
\begin{proof}
Suppose towards a contradiction that $y$ respects $F'$.
Assume that $\ab{w,y} < 0$ (otherwise use $-y$ in place of $y$).
We pick $x \in F'$ to be the average of all extreme points of $F'$,
so that the constraints of $\PM$ which are tight for $x$ are exactly those which are tight for $F'$.
Select $\varepsilon > 0$ very small.
Then $\ab{x + \varepsilon y, w} < \ab{x, w}$, which will contradict the definition of $F' = \argmin \{ \ab{w,x} : x \in F \}$ once we show that $x + \varepsilon y \in F$.
We show that $x + \varepsilon y \in F' \subseteq F$ by verifying:
\begin{itemize}
	\item If $e \in \suppo(F')$ (i.e., $e$ is an edge with $x_e > 0$), then $(x + \varepsilon y)_e = x_e + \varepsilon y_e \ge 0$ if $\varepsilon$ is chosen small enough.
	\item If $e \in E \setminus \suppo(F')$ (i.e., $e$ is an edge with $x_e = 0$), then from $y$ respecting $F'$ we get $e \not \in \supp(y)$ and so $(x + \varepsilon y)_e = 0$.
	\item If $S \not \in \tight(F')$ is an odd set not tight for $F'$, i.e., $\ab{x, \one_{\delta(S)}} > 1$, then $\ab{x + \varepsilon y, \one_{\delta(S)}} = \ab{x, \one_{\delta(S)}} + \varepsilon \ab{y, \one_{\delta(S)}} \ge 1$ if $\varepsilon$ is chosen small enough.
	\item If $S \in \tight(F')$ is an odd set tight for $F'$ (this includes all singleton sets), then from $y$ respecting $F'$ we get $\ab{y, \one_{\delta(S)}} = 0$ and thus $\ab{x + \varepsilon y, \one_{\delta(S)}} = \ab{x, \one_{\delta(S)}} = 1$.
\end{itemize}
\end{proof}

The following lemma says that we can assign nonzero circulation to many vectors at once
using an oblivious choice of weight function from $\cW$.
It is a minor generalization of Lemma 2.3 of \cite{FennerGT16} and the proof remains similar. \ifieee \else \space We give it for completeness.\fi

\begin{lemma} \label{lem23}
For any number $s$ and for any set of $s$ vectors $y_1, ..., y_s \in \mathbb{Z}^E \setminus \{0\}$ with the boundedness property $\| y_i \|_1 \le 4n^2$, there exists $w \in \cW(n^3 s)$ with $\ab{y_i, w} \ne 0$ for each $i = 1, ..., s$.
\end{lemma}
We usually invoke \cref{lem23} with vectors $y_i$ being the alternating indicator vectors of alternating circuits.
Then
the quantities
$\ab{y_i, w}$ are the circulations of these circuits.
\ifieee
\else
\begin{proof}
Let $w' : E \to \bZ$ be given by $w'(e_j) = \largenumber^j$ for $j = 1, ..., |E|$.
Then we have $\ab{y_i, w'} \ne 0$ for each $i$ because the highest nonzero coefficient dominates the expression.
Formally, let $j'$ be maximum index with $y_i(e_{j'}) \ne 0$ and suppose $y_i(e_{j'}) > 0$ (the other case is analogous).  Then,
because $\| y_i \|_{\infty} \le \| y_i \|_1 \le 4n^2$, we have
\[
\ab{y_i, w'} =
y_i(e_{j'}) \largenumber^{j'} + \sum_{j < j'} y_i(e_j) \largenumber^j >
\largenumber^{j'} + \sum_{j = -\infty}^{j' - 1} (-4n^2) \largenumber^j
= 0 \,. 
\]
Let $t = n^3 s$.
We want to show that there exists $k \in  \{2, ..., t\}$ such that for all $i = 1, ..., s$, $\ab{y_i, w_k} \ne 0$. Recalling the definition of $w_k$ (see \cref{def:W}), $\ab{y_i, w_k} \ne 0$ is equivalent to $\abs{\ab{y_i, w'}} \ne 0 \mod k$. 

This will be implied if there exists $k \in  \{2, ..., t\}$ such that $\prod_i \abs{\ab{y_i, w'}} \ne 0 \mod k$.
So there should be some $k \in  \{2, ..., t\}$ not dividing $\prod_i \abs{\ab{y_i, w'}}$ -- equivalently, $\lcm(2,...,t)$ should not divide  $\prod_i \abs{\ab{y_i, w'}}$.
Knowing that $\prod_i \abs{\ab{y_i, w'}} \ne 0$, this will follow if we have $\prod_i \abs{\ab{y_i, w'}} < \lcm(2,...,t)$.
This is true because
\[
\prod_{i=1}^s \abs{\ab{y_i, w'}} < \rb{\largenumber^{|E| + 1}}^s < \largenumber^{n^2s} = 2^{n^2s \log \largenumber} < 2^{n^3s} = 2^t < \lcm(2,...,t)
\]
where we used that $\lcm(2,...,t) > 2^t$ for $t \ge 7$ \cite{Nair82}.
\end{proof}
\fi

\cref{lem23,disc_no_respect} together imply the following:

\begin{corollary} \label{killing_set_of_cycles}
Let $F$ be a face.
For any finite set of vectors $\cY \subseteq \mathbb{Z}^E \setminus \{0\}$ with the boundedness property $\| y \|_1 \le 4n^2$ for every $y \in \cY$, there exists $w \in \cW(n^3 \cdot |\cY|)$ such that each $y \in \cY$ does not respect 
the face $F' = \facemin{F}{w}$.
\ifieee \else \qedmanual\fi
\end{corollary}

\section{Contractible sets and $\lambda$-goodness} \label{sec:lambda-goodness}

We will make progress by ensuring that larger and larger parts of the graph are ``isolated'' in our current face $F$.
By ``parts of the graph'' we mean sets $S$ which are tight for $F$.
As discussed in \cref{sec:ourapproach},
for such a set $S$, the following isolation property is desirable: once the (only) edge of a matching which lies on the boundary of $S$ is fixed,
the entire matching inside $S$ is uniquely determined.
This motivates the following definition:
\begin{definition}
  \label{def:contractable}
Let $F$ be a face and let $S\in \tight(F)$ be a tight set for $F$. We say that $S$ is \emph{$F$-contractible} if for every $e \in \delta(S)$ there are no two perfect matchings in $F$ which both contain $e$ and are different inside $S$.
\end{definition}
Note that, in the above definition, there could be no such perfect matching for certain edges $e\in \delta(S)$ (this is the case if and only if $e \not \in \suppo(F)$).
Intuitively, a contractible set can be thought of as a single vertex with respect to the structure of the current face of the perfect matching polytope.
The notion of contractibility enjoys the following two natural monotonicity properties:

\begin{fact}
  \label{sub-face-contractable}
Let $F' \subseteq F$ be two faces. If $S$ is $F$-contractible, then it is also $F'$-contractible.
\ifieee \else \qedmanual\fi
\end{fact}

\begin{lemma} \label{contractability_downward_closed}
Let $F$ be a face and $S \subseteq T$ two sets tight for $F$, i.e., $S, T \in \tight(F)$. If $T$ is $F$-contractible, then so is $S$.
\end{lemma}
\ifieee
\else
\begin{proof}
Let $e \in \delta(S)$.
Suppose that $M_1$ and $M_2$ are two perfect matchings in $F$ which contain $e$ but are different inside $S$.
We will argue that in that case there also exist two perfect matchings $M_1$ and $M_{12}$ in $F$ which contain $e$, are different inside $S$, and are equal outside of $S$.

Once we have that, we conclude as follows.
Let $f$ be the (only) edge in $\delta(T) \cap M_1$ (perhaps $f = e$); then also $f \in M_{12}$.
Then $M_1$ and $M_{12}$ are two perfect matchings in $F$ which contain $f \in \delta(T)$ but are different inside $T$, contradicting that $T$ is $F$-contractible.

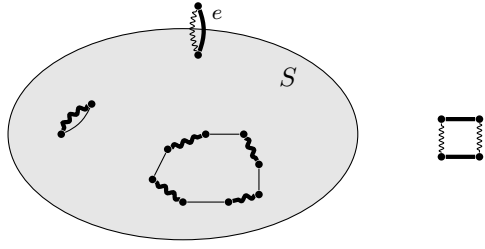
\begin{figure}[t!]
  \begin{center}
    \tikzset{external/export next=false} 
    \begin{tikzpicture}
    \tikzstyle{vertex}=[circle, fill=black, minimum size=2,inner sep=1pt]

        \draw[fill=gray!20!white] (3.1,-0.8) ellipse (2.3cm and 1.4cm) node[above right = 0.50cm and 1.13cm] {\small $S$};
        \node[vertex] (e1) at (1.9,-0.4) { };  
        \node[vertex] (e1o) at (1.5,-0.8) { };  
        \draw (e1) edge[bend left = 20]  (e1o);
        \draw(e1) edge[decorate,decoration={snake,amplitude=.3mm,segment length=4pt,post length=0mm}, ultra thick, bend right=20] (e1o);

        \node[vertex] (a1) at (2.7,-1.4) { };  
        \node[vertex] (a2) at (2.9,-1.0) { };  
        \node[vertex] (a3) at (3.4,-0.8) { };  
        \node[vertex] (a4) at (3.9,-0.8) { };  
        \node[vertex] (a5) at (4.1,-1.2) { };  
        \node[vertex] (a6) at (4.1,-1.6) { };  
        \node[vertex] (a7) at (3.7,-1.7) { };  
        \node[vertex] (a8) at (3.1,-1.7) { };  
        \draw (a1) edge (a2);
        \draw (a2) edge[decorate,decoration={snake,amplitude=.3mm,segment length=4pt,post length=0mm}, ultra thick] (a3);
        \draw (a3) edge (a4);
        \draw (a4) edge[decorate,decoration={snake,amplitude=.3mm,segment length=4pt,post length=0mm}, ultra thick] (a5);
        \draw (a5) edge (a6);
        \draw (a6) edge[decorate,decoration={snake,amplitude=.3mm,segment length=4pt,post length=0mm}, ultra thick] (a7);
        \draw (a7) edge (a8);
        \draw (a8) edge[decorate,decoration={snake,amplitude=.3mm,segment length=4pt,post length=0mm}, ultra thick] (a1);
        \node[vertex] (b1) at (3.3,0.25) { };  
        \node[vertex] (b2) at (3.3,0.90) { };  
        \draw (b1) edge[ultra thick, bend right = 20] node[above right= 0cm and -0.05cm] {\scriptsize $e$}  (b2);
        \draw(b1) edge[decorate,decoration={snake,amplitude=.3mm,segment length=2pt,post length=0mm}, bend left=20] (b2);
        \begin{scope}[xshift=2.5cm, yshift = -3cm]
          \node[vertex] (b1) at (4.0, 2.4) { };  
          \node[vertex] (b2) at (4.5, 2.4) { };  
          \node[vertex] (b3) at (4.5, 1.9) { };  
          \node[vertex] (b4) at (4.0, 1.9) { };  
          \draw (b1) edge[ultra thick] (b2);
          \draw (b2) edge[decorate,decoration={snake,amplitude=.3mm,segment length=2pt,post length=0mm}] (b3);
          \draw (b3) edge[ultra thick] (b4);
          \draw (b4) edge[decorate,decoration={snake,amplitude=.3mm,segment length=2pt,post length=0mm}] (b1);
        \end{scope}
    \end{tikzpicture}
  \end{center}
  \caption{Illustration of the matching $M_{12}$ constructed in the proof of Lemma~\ref{contractability_downward_closed}. Straight and swirly edges denote $M_1$ and $M_2$ respectively. The thick edges denote $M_{12}$, which agrees with $M_1$ outside $S$ and with $M_2$ inside $S$.}
  \label{fig:contractability_downward_closed}
\end{figure}

To get the outstanding claim, we define 
\[ M_{12} = (M_1 \setminus E(S)) \cup (M_2 \cap E(S)) \]
 to be the perfect matching that agrees with $M_1$ on
 all edges not in $E(S)$ and agrees with $M_2$ on all edges in $E(S)$ (see Figure~\ref{fig:contractability_downward_closed}).
  To see that $M_{12}$ is a perfect matching, notice that both
  $M_1$ and $M_2$ are in $F$ and contain $e$.
  Furthermore, as $e \in \delta(S)$ for the tight set $S \in \tight(F)$,
  we have that $M_1 \cap E(S)$ and $M_2 \cap E(S)$ are both perfect matchings
  on the vertex set $S$ where we ignore the vertex incident to $e$.
  We can thus ``replace'' $M_1 \cap E(S)$ by  $M_2 \cap E(S)$  to obtain
  the perfect matching $M_{12}$. 

  We now show that $M_{12}$ is in the face
  $F$. Suppose the contrary.  Since $M_1$ and $M_2$ are both in $F$, we
  have  $M_{12} \subseteq \suppo(F)$. Therefore, if $M_{12}$ is not in $F$,  we must have 
  $|\delta(R) \cap M_{12}| > 1$ for some tight set $R \in \tight(F)$. Since $|R|$ is odd,
  also $|\delta(R) \cap M|$ is  odd
  for any perfect matching $M$. In particular, $|\delta(R) \cap M_{12}| \geq 3$, which contradicts
  \begin{align*}
    |\delta(R) \cap M_{12}| \leq |\delta(R) \cap M_1| + |\delta(R) \cap M_2| = 2\,,
  \end{align*}
  where the equality holds because $M_1$ and $M_2$ are perfect matchings in
  $F$ and $R \in \tight(F)$ is a tight set.
\end{proof}
\fi

In our proof, we will be working with faces and laminar families which are compatible in the following sense:

\begin{definition}
Let $F$ be a face and $\cL$ a laminar family. If $\cL \subseteq \tight(F)$, i.e., all sets $S \in \cL$ are tight for $F$, then we say that $(F, \cL)$ is a \emph{face-laminar pair}.
\end{definition}

Given a face-laminar pair $(F,\cL$), we will often work with a multigraph obtained from $G$ by contracting all small sets, i.e., those with size being at most some parameter $\lambda$ (which is a measure of our progress). This multigraph will be called the \emph{contraction} (see Figure~\ref{fig:contractiongraph} for an example).

In the contraction, we will also remove all boundaries of larger sets (i.e., those with size larger than $\lambda$).
This is done to simulate working inside each such large set independently, because the contraction then decomposes into a collection of disconnected components, one per each large set.
Because, in the contraction, each set in $\cL$ has either been contracted or has had its boundary removed, our task is reduced to dealing with instances having no laminar sets.

Moreover, we only include those edges which are still in the support of the current face $F$, i.e., the set $\suppo(F)$.

\ifieee
\begin{figure*}
\else
\begin{figure}[t]
\fi
  \begin{center}
  \subfloat[A graph $G$ and a laminar family $\cL$. We only draw the edges in $\suppo(F)$. We also do not draw ellipses for the singleton sets in $\cL$. The dark-gray sets are $F$-contractible.]{
    \ifieee \begin{tikzpicture}[scale=0.7]
    \else
    \begin{tikzpicture}[scale=0.9]
      \fi
    \tikzstyle{vertex}=[circle, fill=black, minimum size=2,inner sep=1pt]

        \draw[fill=gray!05!white] (0.0,1.7) ellipse (3.0cm and 2.1cm); 

        \node[vertex] (a1) at (0,0) { };  
  
        \begin{scope}[rotate=-45, xshift=-1.5cm, yshift=-1cm]
          \draw[fill=gray!40!white] (-1.25,1.5) ellipse (0.7cm and 1.0cm);
          \node[vertex] (b1) at (-1,1) { };  
          \node[vertex] (b2) at (-1.5,1.5) { };  
          \node[vertex] (b3) at (-1,2) { };  
        \end{scope}

        \begin{scope}[xshift=1.5cm,yshift=1.5cm]
          \draw[fill=gray!40!white] (0,0.25) ellipse (1cm and 0.7cm);
          \node[vertex] (c1) at (0,0) { };  
          \node[vertex] (c3) at (0.5,0.5) { };  
          \node[vertex] (c2) at (-0.5,0.5) { };  
        \end{scope}

        \begin{scope}[rotate=10,xshift=5.0cm,yshift=1.8cm]
          \draw[fill=gray!40!white] (0,0.25) ellipse (1cm and 0.7cm);
          \node[vertex] (d1) at (0,0) { };  
          \node[vertex] (d3) at (0.5,0.5) { };  
          \node[vertex] (d2) at (-0.5,0.5) { };  
          \node[vertex] (e) at (-0.25,-1.5) { };  
          \node[vertex] (f) at (-1.4,-2.6) { };  
        \end{scope}

        \draw (a1) edge[thick] (b1) edge[thick] (c2);
        \draw (b2) edge (b1) edge (b3);
        \draw (b3) edge[thick] (c2);
        \draw (c1) edge (c3);

        \draw (e) edge[thick] (f) edge[thick] (d1);
        \draw (d3) edge (d2) edge (d1);
        \draw (f) edge[bend left=12] (b3);
        \draw (f) edge (a1);

        \draw (d2) edge (c2);

        \draw[rounded corners=5pt] (-3.3,-1) rectangle (6,4);
 
    \end{tikzpicture}
  }
  \subfloat[The $(F, \cL, 4)$-contraction of $G$. Its vertices are labeled by their node-weights.]{
    \ifieee \begin{tikzpicture}[scale=0.7]
    \else
    \begin{tikzpicture}[scale=0.9]
      \fi
      
    \tikzstyle{vertex}=[circle, fill=black, minimum size=2,inner sep=1pt]

          \node[vertex] (a1) at (0,0) { };  
          \node at (0.3, -0.1) { $1$ };
    
          \begin{scope}[rotate=-45, xshift=-1.5cm, yshift=-1cm]
            \node[vertex] (b1) at (-1,1) { };  
            \node[vertex] (b2) at (-1.5,1.5) { };  
            \node[vertex] (b3) at (-1,2) { };  
            \draw[fill=gray!40!white] (-1.25,1.5) ellipse (0.7cm and 1.0cm);
            \node at (-1.25, 1.5) {$3$};
          \end{scope}

          \begin{scope}[xshift=1.5cm,yshift=1.5cm]
            \node[vertex] (c1) at (0,0) { };  
            \node[vertex] (c3) at (0.5,0.5) { };  
            \node[vertex] (c2) at (-0.5,0.5) { };
            \draw[fill=gray!40!white] (0,0.25) ellipse (1cm and 0.7cm);
            \node at (0, 0.25) {$3$};
          \end{scope}

          \begin{scope}[rotate=10,xshift=5.0cm,yshift=1.8cm]
            \node[vertex] (d1) at (0,0) { };  
            \node[vertex] (d3) at (0.5,0.5) { };  
            \node[vertex] (d2) at (-0.5,0.5) { };  
            \node[vertex] (e) at (-0.25,-1.5) { };  
            \node at (0.1, -1.5) { $1$ };
            \node[vertex] (f) at (-1.4,-2.6) { };  
            \node at (-1.7, -2.6) { $1$ };
            \draw[fill=gray!40!white] (0,0.25) ellipse (1cm and 0.7cm);
            \node at (0, 0.25) {$3$};
          \end{scope}

          \begin{scope}[on background layer]
            \draw (a1) edge[thick] (b1) edge[thick] (c2);
            \draw (b3) edge[thick] (c2);

            \draw (e) edge[thick] (f) edge[thick] (d1);
          \end{scope}

          \draw[rounded corners=5pt] (-3.0,-1) rectangle (6,4);

    \end{tikzpicture}
    }
  \end{center}
  \vspace{-1.2em}
  \caption{
  An example of the $(F, \cL, \lambda)$-contraction of $G$. 
  }
  \label{fig:contractiongraph}
\ifieee
\end{figure*}
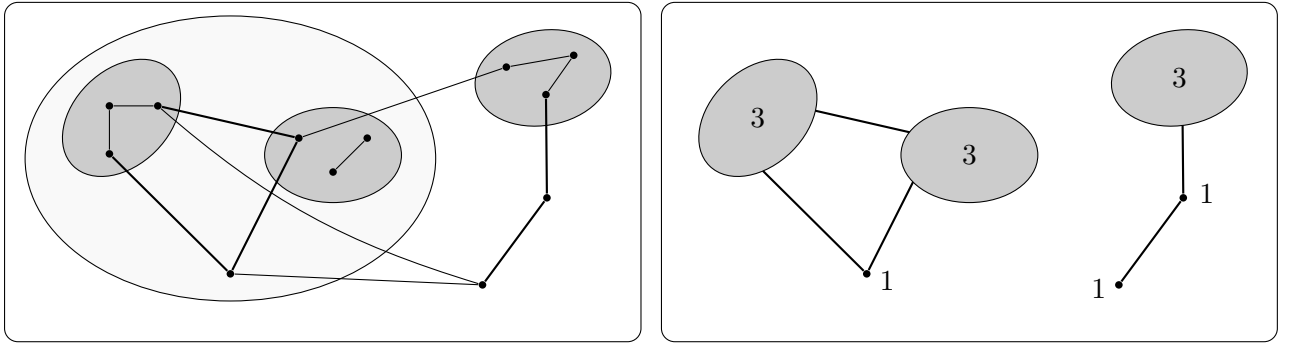
\else
\end{figure}
\fi

\begin{definition} \label{def:contraction}
Given a face-laminar pair $(F,\cL$) and a parameter $\lambda$ (with $1 \le \lambda \le 2n$), we define the
\emph{$(F, \cL, \lambda)$-contraction} of $G$ as a node-weighted
multigraph as follows:
\begin{itemize}
	\item the node set is the set of maximal sets of size (cardinality) at most $\lambda$ in $\cL$,
	\item each node has a node-weight equal to the size of the corresponding set,
	\item
	  the edge set is obtained from $\suppo(F) \setminus \bigcup_{T\in \cL: |T| > \lambda} \delta(T)$ by contracting each of these maximal sets.
	  That is, an edge of $G$ maps to an edge of the contraction if it is in $\suppo(F)$, it is not inside any of the contracted sets and it does not cross any cut defined by a set $T\in \cL: |T| > \lambda$. \ifieee \else Sometimes we identify edges of the contraction with their preimages in $G$. \fi
\end{itemize}
\end{definition}

In the $(F, \cL, \lambda)$-contractions arising in our arguments,
we will always only contract sets $S \in \cL$ which are $F$-contractible
(i.e., the vertices of a contraction will always correspond to $F$-contractible sets).
Then, a very useful property is that alternating circuits in the contraction can be lifted to alternating circuits in the entire graph $G$ in a  canonical way.
\ifieee\else
This is done in the proofs of \cref{killing_cycle,adding_new_laminar_sets}.
\fi

Finally, we need the following extension of \cref{def:respect} for vectors defined on the contraction. 

\begin{definition} \label{def:respect_contraction}
Denote the $(F, \cL, \lambda)$-contraction of $G$ as $H$, and let $z \in \bZ^{E(H)}$ be a vector on the edges of $H$.
We say that $z$ respects a subface $F' \subseteq F$ if~\footnote{In the following conditions we abuse notation and think of $z$ as a vector in $\bZ^{E}$ obtained by identifying each edge of $H$ with its preimage in $G$ and letting those edges of $G$ without a preimage in $H$ have value $0$.}
\begin{itemize}
	\item $\supp(z) \subseteq \suppo(F')$, and
	\item for each $S \in \tight(F')$ which is a union of sets corresponding to vertices in $V(H)$,\ifieee \else\footnote{That is, the maximal sets of size at most $\lambda$ in $\cL$.} \fi we have $\ab{z, \one_{\delta(S)}} = 0$.
\end{itemize}
As before, we say that an alternating circuit $C$ in $H$ respects a subface $F'$ if its alternating indicator vector $\pmind{C} \in \bZ^{E(H)}$ respects $F'$.
\end{definition}

Now we are able to define our measure of progress.
On one hand, we want to make larger and larger laminar sets contractible.
On the other hand, there could very well be no laminar sets, so we also proceed as in the bipartite case: remove longer and longer alternating circuits.

\begin{definition}
  \label{def:face-laminar}
Let $(F,\cL)$ be a face-laminar pair and $\lambda$ a parameter (with $1 \le \lambda \le 2n$). We say that $(F,\cL)$ is \emph{$\lambda$-good} if
$\cL$ is a maximal laminar subset of $\tight(F)$ and:
\begin{enumerate}
	\item[(i)] each $S \in \cL$ with $|S| \le \lambda$ is $F$-contractible,
	\item[(ii)] in the $(F, \cL, \lambda)$-contraction of $G$, there is no alternating circuit of node-weight at most $\lambda$.
\end{enumerate}
\end{definition}

We begin with $\lambda = 1$, which is trivial, and then show that by concatenating enough weight functions we can obtain face-laminar families which are $2$-good, $4$-good, $8$-good, and so on.
We are done once we have a $\lambda$-good family with $\lambda \ge n$.
The components of this proof strategy are given in the following three claims.
The first step is clear:

\begin{fact} \label{initialization}
Let $\cL_0$ be a maximal laminar subset of $\tight(PM)$.
Then the face-laminar pair $(\PM, \cL_0)$ is $1$-good.
\ifieee \else \qedmanual\fi
\end{fact}
\ifieee\else
Note that $\cL_0$ contains all singleton sets, i.e., $\{ \{v\} : v \in V \} \subseteq \cL_0$.
\fi

We then proceed iteratively in $\log_2 n$ rounds using the following theorem.
Its proof, which constitutes the bulk of our argument,
\ifieee
\else
is given in \cref{proof_of_main}.
\fi

\begin{theorem} \label{main}
Let $(F,\cL)$ be a $\lambda$-good face-laminar pair.
Then there exists a weight function $w \in \cW^{\log_2 n + 1}$
and a laminar family $\cL' \supseteq \cL$ such that
$(\facemin{F}{w}, \cL')$ is a $2 \lambda$-good face-laminar pair.
\end{theorem}

We are done once $\lambda$ exceeds $n$:

\begin{lemma} \label{finalization}
Suppose $(F,\cL)$ is $\lambda$-good for some $\lambda \ge n$. Then $|F| = 1$.
\end{lemma}

\ifieee
\else
We think that the proof of this lemma is instructive. It serves to understand and motivate \cref{def:face-laminar},
and more involved versions of this argument appear in the sequel.

\begin{proof}
Let $H$ be the $(F, \cL, \lambda)$-contraction of $G$.
Also let $S_1, S_2, ..., S_k$ be all maximal sets in $\cL$.
As $\cL$ contains all singletons, their disjoint union is $V$ and we have $V(H) = \{ S_1, ..., S_k \}$ and $E(H) = \bigcup_{i=1}^k \delta(S_i)$.
Since $(F, \cL)$ is $\lambda$-good with $\lambda \geq n$, each set $S \in \cL$ is $F$-contractible, and $H$ contains no alternating circuit of node-weight at most $\lambda$ --
in particular, $H$ contains no alternating simple cycle. Indeed, an upper bound on the node-weight of any alternating simple cycle is $|S_1| + ... + |S_k| = n \le \lambda$.

Now we show that there is only one perfect matching in $F$.
One direction is easy: since $F$ is a face, it is nonempty by definition.
For the other direction, let $M_1$ and $M_2$ be two perfect matchings in $F$. We show that $M_1 = M_2$.

Because the sets $S_1, ..., S_k$ are tight for $F$ and no edge can possibly cross a tight odd-set of cardinality at least $\lambda \geq n$, any perfect matching (in the face $F$) in $G$ induces a perfect matching in $H$.
If the matchings induced by $M_1$ and $M_2$ were different, then their symmetric difference would contain an alternating simple cycle in $H$, which is impossible.
So the induced matchings must be equal, i.e., $M_1 \cap \bigcup_i \delta(S_i) = M_2 \cap \bigcup_i \delta(S_i)$.
Moreover, the sets $S_1, ..., S_k$ are $F$-contractible, which means that, given the boundary edges, there is a unique perfect matching in $F$ inside each $S_i$.
This yields $M_1 = M_2$.
\end{proof}
\fi

\ifieee
\else
Before we proceed to the proof of \cref{main}, let
\fi
us see how \cref{initialization,main,finalization} together give our desired  result:

\begin{theorem} \label{exists_isolating}
There exists an isolating weight function $w \in \cW^{\rb{\log_2 n + 1} \log_2 n}$, i.e., one with $|\facemin{\PM}{w}| = 1$.
\end{theorem}
\begin{proof}
Let $\ell = \log_2 n$.
We iteratively construct a sequence of face-laminar pairs $(F_i, \cL_i)$ for $i = 0, 1, ..., \ell$ such that $(F_i, \cL_i)$ is $2^i$-good and $F_i = \facemin{F_{i-1}}{w_i}$ for some weight function $w_i \in \cW^{\ell + 1}$. We begin by setting $F_0 = \PM$ and $\cL_0$ to be a maximal laminar subset of $\tight(PM)$.  By \cref{initialization}, $(F_0,\cL_0)$ is $1$-good. Then for $i = 1, ..., \ell$ we use \cref{main} to obtain the wanted weight function $w_i$ along with a laminar family $\cL_i \supseteq \cL_{i-1}$. Finally, we have that $(F_{\ell}, \cL_{\ell})$ is $2^{\ell}$-good, so that by \cref{finalization}, $|F_{\ell}| = 1$.

It remains to argue that $F_{\ell} = \facemin{\PM}{w}$ for some $w \in \cW^{\rb{\ell + 1} \ell}$.
To do this, we proceed as in \cref{weight_functions}: define the concatenation $w' \bullet w'' := n^{21 \rb{\ell + 1}} w' + w''$ for two weight functions $w'$ and $w''$, where $w'' \in \cW^{\ell + 1}$.
\ifieee\else (We need to use a padding term $n^{21 \rb{\ell + 1}}$ which is larger than the $n^{21}$ of \cref{def:concatenation} because the right-hand weight functions are now from $\cW^{\ell + 1}$ rather than from $\cW$.) \fi
By the same reasoning as for \cref{concatenation_and_subface} we get that $F_{\ell} = \facemin{\facemin{\facemin{\PM}{w_1}}{w_2} ...}{w_{\ell}} = \facemin{\PM}{w_1 \bullet w_2 \bullet ... \bullet w_{\ell}}$.
We put $w = w_1 \bullet w_2 \bullet ... \bullet w_{\ell} \in \cW^{(\ell + 1) \ell}$.
\end{proof}

\cref{exists_isolating} implies \cref{mainer}
because we have $|\cW^{\rb{\log_2 n + 1} \log_2 n}| = |\cW|^{\rb{\log_2 n + 1} \log_2 n} \le n^{20 \rb{\log_2 n + 1} \log_2 n}$,
the values of any $w \in \cW^{\rb{\log_2 n + 1} \log_2 n}$ are bounded by $n^{21 \rb{\log_2 n + 1} \log_2 n}$,
and the functions $w \in \cW$ can be generated obliviously using only the number of vertices $n$.

\ifieee\else
\section{Proof of the key \cref{main}: from $\lambda$-good to $2\lambda$-good} \label{proof_of_main}

In this section we show how to make progress (measured by the $\lambda$ parameter of $\lambda$-goodness) by applying a new weight function to the current face.
Our objective is to make larger sets contractible (by doubling the size threshold from $\lambda$ to $2 \lambda$) and to ensure that in the new contracted graph, alternating circuits of an increased node-weight are not present.
We do this by moving from the current face-laminar pair, which we call $(\Fin,\Lin)$, to a new face-laminar pair $(\Fout,\Lout)$.
Both pairs have the property that the laminar family is a maximal laminar family of sets tight for the face. The new family extends the previous, i.e., $\Lout \supseteq \Lin$.

Our main technical tools are~\cref{circuit-removal} and~\cref{making-2l-contractible}.~\cref{circuit-removal} is used to ensure that certain alternating circuits are not present in the new contraction. 
It says that if our current contraction has no alternating circuits of at most some node-weight, then a single weight function $w \in \cW$ is enough to guarantee that all alternating circuits of at most twice that node-weight do not respect the new face obtained by applying $w$.
We call this \emph{removing} these circuits. \cref{making-2l-contractible} is used to make sure that sets in our laminar family that are of size at most $2 \lambda$ become contractible. 
Later, new sets will be added to the laminar family in \cref{adding_new_laminar_sets}, in such a way that these properties are maintained
and that the removed alternating circuits indeed do not survive in the new contraction.

The formal structure of the proof is as follows.
We begin from a $\lambda$-good face-laminar pair $(\Fin, \Lin)$. Then, using our technical tools~\cref{circuit-removal,making-2l-contractible}, 
we show in \cref{good-weight-fn}  the existence of a weight function $\wout \in \cW^{\log_2(n) + 1}$ such that the face $\Fout = \facemin{\Fin}{\wout}$ satisfies two conditions which make progress on conditions (i) and (ii) of $\lambda$-goodness:
  \begin{itemize}
    \item[(i)'] For each $S \in \Lin$ with $|S| \le 2 \lambda$, $S$ is $\Fout$-contractible.
    \item[(ii)'] In the $(\Fout, \Lin, 2 \lambda)$-contraction of $G$, there is no $\Fout$-respecting alternating circuit of node-weight at most $2 \lambda$.
  \end{itemize}
This gives us the wanted face $\Fout$ and weight function $\wout$.
Finally, in \cref{adding_new_laminar_sets} we show that extending the laminar family $\Lin$ to a maximal laminar family $\Lout$ (of sets tight for the new face) yields a $2 \lambda$-good pair $(\Fout, \Lout)$.
This finishes the proof of \cref{main}.

\subsection{Removing alternating circuits}

This section is devoted to the proof of \cref{circuit-removal},
which is a technical tool we use to remove alternating circuits of size between $\lambda$ and $2 \lambda$
from the contraction.

\begin{theorem}
  Consider a face-laminar pair $(F, \cL)$ such that each $S\in \cL$ with $|S| \le \beta$ is
  $F$-contractible (for a parameter $\beta$). Denote by $H$ the  $(F, \cL, \beta)$-contraction of $G$.
  If $H$ has no
  alternating circuit of node-weight at most $\lambda$,
  then
  there exists $w\in \cW$ such that $H$ has no $\facemin{F}{w}$-respecting alternating
  circuit of node-weight at most $2 \lambda$.
  \label{circuit-removal}
\end{theorem}

We begin with a simple technical fact: to verify that a vector respects a face, it is enough to check this for a maximal laminar family of tight constraints.

\begin{lemma} \label{respect_L_respect_F}
Consider a face $F$
and a vector $y \in \bZ^E$.
Let $\cL$ be a maximal laminar subset of $\tight(F)$.
If for each $S \in \cL$ we have $\ab{y, \one_{\delta(S)}} = 0$,
then the same holds for all $S \in \tight(F)$.
\end{lemma}
\begin{proof}
As $\cL$ is
a maximal laminar subset of $\tight(F)$, Lemma~\ref{face_structure} says
that $\Span(\cL) = \Span(\tight(F))$. In other words, for any $S\in
\tight(F)$ we can write $\one_{\delta(S)}$ as a linear combination
$\sum_{L \in \cL} \mu_L \one_{\delta(L)}$ for some coefficients
$(\mu_L)_{L\in \cL}$. Hence
\begin{align*}
	\ab{y, \one_{\delta(S)}} = \ab{y, \sum_{L\in \cL} \mu_L \one_{\delta(L)}} = \sum_{L \in \cL} \mu_L \ab{y, \one_{\delta(L)}} = 0.
\end{align*}
\end{proof}

Now we prove a lemma which reduces the task of removing an alternating circuit in $H$ to that of removing a vector defined on the edges of $G$, which we can do using \cref{killing_set_of_cycles}.
Throughout this section, $F$, $\cL$ and $H$ are as in the statement of \cref{circuit-removal}. 
Recall that the vertices of $H$ are elements of $\cL$, i.e., sets of vertices, and so by $S\in V(H)$ we mean the set $S\in \cL$ that corresponds to a vertex in $H$. 
\begin{lemma} \label{killing_cycle}
Let $z \in \bZ^{E(H)}$ be a nonzero vector on the edges of $H$ satisfying $\ab{z, \one_{\delta(S)}} = 0$ for each $S \in V(H)$.
Then there exists a nonzero vector $y \in \bZ^E$ such that for any face $F' \subseteq F$ we have: if $z$ respects $F'$, then $y$ respects $F'$.
We also have $\| y \|_1 \le n \| z \|_1$.
\end{lemma}
We remark that $y$ does not depend on $F'$.
\begin{proof}
We consider $z$ as a vector $z \in \bZ^{E}$ by identifying each edge of $H$ with its preimage in $G$ and letting those edges of $G$ without a preimage in $H$ have value $0$. We may assume $\supp(z) \subseteq \suppo(F)$; otherwise $z$ cannot respect $F'$ (see~\cref{def:respect_contraction}) and thus we are done by outputting any $y$.

The proof idea is to extend $z$ to a vector $y \in \bZ^E$ which resembles an alternating indicator vector. We do this in a canonical way so that if this extension does not respect $F'$, then it must be because $z$ itself does not respect $F'$.

To this end, we do the following for each $S \in V(H)$: pair up the boundary edges $e \in \delta(S)$ which have $z_e > 0$ with boundary edges $e$ which have $z_e < 0$, respecting their multiplicities as given by $z$. For example, if we had $\delta(S) = \{e_1, e_2, e_3\}$ with $z(e_1) = 3$, $z(e_2) = -2$ and $z(e_3) = -1$, we would get the pairs $\{(e_1, e_2), (e_1, e_2), (e_1, e_3)\}$. Such a pairing is always possible because $\ab{z, \one_{\delta(S)}} = 0$. Let $\cb{ (e^+_i, e^-_i) }_i$ be the multiset of pairs of edges obtained in this way across all $S \in V(H)$, and let $S_i \in V(H)$ be the set for which the pair $(e^+_i, e^-_i)$ has been introduced. Also denote by $v^+_i, v^-_i$ the $S_i$-endpoints of edges $e^+_i, e^-_i$.

Now, for each $i$ we have $e^+_i, e^-_i \in \supp(z) \subseteq \suppo(F)$, and $S$ is $F$-contractible, so there is a unique perfect matching $M_i^+$ on the vertex-induced subgraph $(S_i \setminus \{v^+_i\}, E(S_i \setminus \{v^+_i\}))$ in $F$ (more precisely, $M_i^+$ is the unique perfect matching on that subgraph that extends to a matching in $F$), as well as a unique perfect matching $M_i^-$ on $(S_i \setminus \{v^-_i\}, E(S_i \setminus \{v^-_i\}))$ in $F$. We let
\[ y := z + \sum_i \rb{\one_{M_i^+} - \one_{M_i^-}}. \]
What remains now is to prove the following claim:
\begin{claim} \label{claim:y_respects}
Let $F' \subseteq F$ be such that $z$ respects $F'$. Then $y$ respects $F'$.
\end{claim}
\begin{proof}
Since $z$ respects $F'$ (see~\cref{def:respect_contraction}), we have $\supp(z) \subseteq \suppo(F')$. This implies that for each $i$, $M_i^+ \subseteq \suppo(F')$. Indeed, since $e_i^+ \in \supp(z) \subseteq \suppo(F')$, there is a perfect matching on $(S_i \setminus \{v^+_i\}, E(S_i \setminus \{v^+_i\}))$ in $F'$. However, $S_i$ is $F$-contractible and thus $M_i^+$ is the \emph{only} such matching in $F$ (thus also in $F'$). Therefore $M_i^+ \subseteq \suppo(F')$ and analogously $M_i^- \subseteq \suppo(F')$.

Now we check the conditions for $y$ to respect $F'$ (see~\cref{def:respect}):
\begin{itemize}
	\item We have $\supp(y) = \supp(z) \cup \bigcup_i \rb{M_i^+ \cup M_i^-} \subseteq \suppo(F')$.
	\item Let $T \in \tight(F')$. We need to verify that $\ab{y, \one_{\delta(T)}} = 0$. Let $\cL'$ be a maximal laminar subfamily of $\tight(F')$ extending $\cL$, i.e., $\cL \subseteq \cL' \subseteq \tight(F')$. By \cref{respect_L_respect_F}, it is enough to verify that $\ab{y, \one_{\delta(T)}} = 0$ for $T \in \cL'$.
	For a set $T$ belonging to a laminar family which extends $\cL$, it is not hard to see that there are two possibilities: either $T \subsetneq S$ for some $S \in V(H)$, or $T$ is a union of sets in $V(H)$ (for recall that $V(H)$ is a partitioning of $V$ that consists of sets in $\cL$). In the latter case, $\ab{y, \one_{\delta(T)}} = \ab{z, \one_{\delta(T)}} = 0$ because
$y$ equals $z$ on edges crossing sets in $V(H)$
and because $z$ respects $F'$. In the former case, we have
	\begin{align*}
	\ab{y, \one_{\delta(T)}} &= \ab{\sum_{i : S_i = S} \rb{ \one_{e_i^+} - \one_{e_i^-} + \one_{M_i^+} - \one_{M_i^-} }, \one_{\delta(T)}}\\ &= \sum_{i : S_i = S} \ab{\one_{e_i^+} - \one_{e_i^-} + \one_{M_i^+} - \one_{M_i^-}, \one_{\delta(T)}}
	\end{align*}
	because these are the only edges in $y$'s support that have an endpoint in $S$ (other edges cannot possibly cross $T \subseteq S$). Now it is enough to show that each summand is $0$.
	
	For this, we know that $M_i^+ \cup \{e_i^+\}$ and $M_i^- \cup \{e_i^-\}$ are (partial) matchings in $F'$ and that $T$ is tight for $F'$.
	Therefore we have $|\delta(T) \cap \rb{M_i^+ \cup \{e_i^+\}}| = 1$,\footnote{Formally, consider a perfect matching $M^+$ (on $G$) in $F'$ which is a superset of $M_i^+ \cup \{e_i^+\}$. Then we have $|\delta(T) \cap M^+| = 1$. But $\delta(T) \cap \rb{M_i^+ \cup \{e_i^+\}} = \delta(T) \cap M^+$ because $T \subseteq S$.} and the same holds for $M_i^- \cup \{e_i^-\}$.
	Therefore $\ab{\one_{M_i^+ \cup \{e_i^+\}}, \one_{\delta(T)}} = 1 = \ab{\one_{M_i^- \cup \{e_i^-\}}, \one_{\delta(T)}}$.
\end{itemize}
\end{proof}
Regarding the norm: every edge (with multiplicity) in $z$ causes less than $n/2$ new edges (a~partial matching) to appear in $y$. Therefore $\| y \|_1 \le (n/2 + 1) \| z \|_1 \le n \|z\|_1$.
\end{proof}

Our second lemma gives a bound on the number of alternating circuits we need to remove.
%
Its proof resembles that of Lemma 3.4 in~\cite{FennerGT16}, but it is  somewhat more complex, as we are dealing
with a node-weighted multigraph, as well as
with alternating circuits instead of simple cycles (see \cref{sec:alternating}).
We have made no attempt to minimize the exponent $17$.
\begin{lemma} \label{counting}
There are polynomially many alternating circuits of node-weight at most $2 \lambda$ in $H$, up to identifying circuits with equal alternating indicator vectors. More precisely, the cardinality of the set
\[ \{ \pmind{C} : C \text{ is an alternating circuit in $H$ of node-weight at most $2 \lambda$} \} \]
is at most $n^{17}$.
\end{lemma}
\begin{proof}
We will associate a small \emph{signature} with each alternating circuit in $H$ of node-weight at most $2 \lambda$, with the property that alternating circuits with different alternating indicator vectors are assigned different signatures. This will prove that the considered cardinality is at most the number of possible signatures, which is polynomially bounded.

\begin{figure}[t]
  \begin{center}
    \begin{tikzpicture}
      \begin{scope}[xscale=1, yscale=0.8]
        \tikzset{VertexStyle/.append style={minimum size=2pt, inner sep=1.5pt, fill=black!10!white}}
        \SetGraphUnit{2}
        \SetVertexNoLabel
        \Vertices{circle}{A,B,C,D,E, F, G, H, I, J}
        \SetVertexLabel
        \Vertex[Node,L={\small $1$}]{A}
        \Vertex[Node,L={\small$9$}]{B}
        \Vertex[Node,L={\small$1$}]{C}
        \Vertex[Node,L={\small$1$}]{D}
        \Vertex[Node,L={\small$1$}]{E}
        \Vertex[Node,L={\small$4$}]{F}
        \Vertex[Node,L={\small$5$}]{G}
        \Vertex[Node,L={\small$3$}]{H}
        \Vertex[Node,L={\small$1$}]{I}
        \Vertex[Node,L={\small$4$}]{J}
        \tikzset{EdgeStyle/.style={<-, above,  midway,font=\scriptsize}}
        \foreach \x/\y/\n in {A/B/3, B/C/2, C/D/1, D/E/10, E/F/9, F/G/8, G/H/7, H/I/6, I/J/5, J/A/4}
        {
          \Edge(\x)(\y)
        }
        \begin{scope}[rotate=20]
          \SetGraphUnit{2.12}
          \tikzset{VertexStyle/.append style={fill=none, draw=none,font=\scriptsize}}
          \Vertices{circle}{$e_2$,$e_1$,$e_0$,$e_{9}$,$e_8$, $e_7$, $e_6$, $e_5$, $e_4$, $e_3$}
        \end{scope}
        \SetGraphUnit{0.5}
        \tikzset{VertexStyle/.append style={fill=none, draw=none,font=\scriptsize}}
        \NO (D) {$i_0$}
        \SetGraphUnit{0.35}
        \NOEA (B) {$i_1$}
        \SetGraphUnit{0.5}
        \SO (H) {$i_2$}
        \SetGraphUnit{0.4}
        \WE (F) {$i_3$}
      \end{scope}

      \begin{scope}[xshift=7cm,xscale=1, yscale =0.8]
        \tikzset{VertexStyle/.append style={minimum size=2pt, inner sep=1.5pt, fill=black!10!white}}
        \SetGraphUnit{2}
        \SetVertexNoLabel
        \Vertices{circle}{A,B,C,D,E, F, G, H, I, J}
        \SetVertexLabel
        \Vertex[Node,L={\small $1$}]{A}
        \Vertex[Node,L={\small$9$}]{B}
        \Vertex[Node,L={\small$1$}]{C}
        \Vertex[Node,L={\small$1$}]{D}
        \Vertex[Node,L={\small$1$}]{E}
        \Vertex[Node,L={\small$4$}]{F}
        \Vertex[Node,L={\small$5$}]{G}
        \Vertex[Node,L={\small$3$}]{H}
        \Vertex[Node,L={\small$1$}]{I}
        \Vertex[Node,L={\small$4$}]{J}

        \SetVertexNoLabel
        \SetGraphUnit{0.7}
        \Vertex[x=0.8, y=0] {v1}
        \Vertex[x=-0.2, y=-0.5] {v2}
        \SetVertexLabel
        \Vertex[Node,L={\small $1$}]{v1}
        \Vertex[Node,L={\small $1$}]{v2}
        \tikzset{EdgeStyle/.style={<-, above,  midway,font=\scriptsize}}
        \foreach \x/\y/\n in {A/B/3, B/C/2, C/D/1, D/E/10, E/F/9, F/G/8, G/H/7, H/I/6}
        {
          \Edge(\x)(\y)
        }
        \tikzset{EdgeStyle/.style={<-,ultra thick,bend right,  above,  midway,font=\scriptsize}}
        \foreach \x/\y/\n in {I/J/5, J/A/4}
        {
          \Edge(\x)(\y)
        }
         \tikzset{EdgeStyle/.style={<-,ultra thick,dotted,bend left,  above,  midway,font=\scriptsize}}
        \Edge(J)(A)
        \tikzset{VertexStyle/.append style={fill=none, draw=none,font=\scriptsize}}
          \Vertex[x=1.5, y=-0.1] {$g_4$}
         \tikzset{EdgeStyle/.style={<-,ultra thick,dotted,  above,  midway,font=\scriptsize}}
         \Edge[label={$g_3$}, labelstyle={fill=none, below , pos=0.2}](v1)(J)
         \Edge[label={$g_2$}, labelstyle={fill=none, above, midway}](v2)(v1)
         \Edge[label={$g_1$}, labelstyle={fill=none, left}](I)(v2)

        \begin{scope}[rotate=20]
          \SetGraphUnit{2.37}
          \tikzset{VertexStyle/.append style={fill=none, draw=none,font=\scriptsize}}
          \Vertices{circle}{ , , , , ,  ,  ,  , $f_2$, $f_1$}
        \end{scope}
        \SetGraphUnit{0.5}
        \tikzset{VertexStyle/.append style={fill=none, draw=none,font=\scriptsize}}
        \NO (D) {$i_0$}
        \SetGraphUnit{0.35}
        \NOEA (B) {$i_1, a$}
        \EA (A) {$b$}
        \SetGraphUnit{0.5}
        \SO (I) {$c$}
        \SO (H) {$i_2, d$}
        \SetGraphUnit{0.4}
        \WE (F) {$i_3$}
      \end{scope}
    \end{tikzpicture}
  \end{center}
  \caption{Intuition of the signature vector definition and the proof of~\cref{counting}.
  \\
  On the left, each vertex is labeled by its node-weight, and the corresponding selection of $i_0, i_1, i_2, i_3$ is shown for $\lambda = 16$.  Notice that the selected vertices partition the alternating circuit into paths; the total node-weight of internal vertices on each path is at most $\lambda/2$.
  \\
  On the right we see two different alternating circuits with the same signature. They differ in that one uses $f_2$ and the other uses $g_3, g_2, g_1$. The thick edges illustrate the alternating circuit $B = (f_1, f_2, g_1, g_2, g_3, g_4)$ of node-weight at most $\lambda$ which leads  to the contradiction. We walk the dashed path ($P_D$) in reverse.}
  \label{fig:signature}
\end{figure}

Let $C = (e_0, e_1, ..., e_{k-1})$ be an alternating circuit in $H$ of node-weight at most $2 \lambda$. We want to define its signature $\sigma(C)$.
To streamline notation, we let $v_i$  be the tail of $e_i$ for $i=0, \dots,
k-1$.
Thus $C$  is of the form
\begin{align*}
  v_0 \xrightarrow{e_0} v_1 \xrightarrow{e_1} \dots \xrightarrow{e_{k-2}} v_{k-1}\xrightarrow{e_{k-1}} v_{0}\,.
\end{align*}
(See also~\cref{fig:signature} for an example.) 
We 
also let $\NW(v_i)$ denote
the node-weight of vertex $v_i$, and $\incoming(v_i) = e_{(i-1) \mymod k}$ and $\outgoing(v_i) = e_{i}$ be the incoming and outgoing edges of $v_i$ in $C$.\footnote{
The functions $\incoming(v)$ and $\outgoing(v)$ are not formally well-defined since they depend on the considered alternating circuit $C$ and on which occurrence of $v$ in the circuit we are considering,
but their values will be clear from the context.}
 We now define the signature $\sigma(C)$ as the output of the following procedure:
\begin{itemize}\itemsep0mm
  \item Let $i_0=0$ be the index of the first vertex in $C$. 
  \item  For $j=1,2,3$, select $i_j\leq k$ to be the largest index satisfying
    $\sum_{i=i_{j-1}+1}^{i_j-1} \NW(v_i) \leq \lambda/2$.
  \item Let $t = \max \{ j: i_j < k\}$  and output the signature  $\sigma(C) = ((-1)^{i_j}, \incoming(v_{i_j}), \outgoing(v_{i_j}))_{j=0,1,\dots, t}$.
  (Note that $t \le 3$.)
\end{itemize}

 The intuition of the signature is as follows (see also the left part 
 of~\cref{fig:signature}). The procedure starts at the first vertex $v_{i_0}
 = v_0$. It then selects the farthest (according to $C$) vertex $v_{i_1}$
 while guaranteeing that the total node-weight of the vertices visited
 in-between $v_{i_0}$ and $v_{i_1}$ is at most $\lambda/2$.   Similarly,
 $v_{i_2}$ is selected to be the farthest vertex such that the total
 node-weight of the vertices $v_{i_1+1}, \dots , v_{i_2 -1}$ is at most
 $\lambda/2$, and $i_3$ is selected in the same fashion. The indices $i_0, i_1, \dots, i_t$ thus partition $C$ into paths
 \begin{align*}
   C_0 &= v_{i_0} \xrightarrow{e_{i_0}}  v_{i_0+1}  \xrightarrow{e_{i_0+1}} \dots \xrightarrow{e_{i_1 -2}} v_{i_1 -1} \xrightarrow{e_{i_1 -1}} v_{i_1} \\
   C_1 &= v_{i_1} \xrightarrow{e_{i_1}}  v_{i_1+1}  \xrightarrow{e_{i_1+1}} \dots \xrightarrow{e_{i_2 -2}} v_{i_2 -1} \xrightarrow{e_{i_2 -1}} v_{i_2} \\
   & \vdots\\
   C_t &= v_{i_t} \xrightarrow{e_{i_t}}  v_{i_t+1}  \xrightarrow{e_{i_t+1}} \dots \xrightarrow{e_{i_0 -2}} v_{i_0-1} \xrightarrow{e_{i_0 -1}} v_{i_0} \\
 \end{align*}
 so that the total node-weight of the internal vertices on each path is at most $\lambda/2$. Indeed, for $C_j$ with $j < 3$ this follows from the
 selection of $i_j$. For $C_3$ (in the case $t = 3$), by maximality of $i_1$, $i_2$ and $i_3$ we have
 \begin{align*}
   \underbrace{\sum_{i=i_0+1}^{i_1} \NW(v_i)}_{\geq \lambda/2} + \underbrace{\sum_{i=i_1+1}^{i_2} \NW(v_i)}_{\geq \lambda/2}+ \underbrace{\sum_{i=i_2+1}^{i_3} \NW(v_i)}_{\geq \lambda/2} \ge \frac{3}{2} \lambda 
 \end{align*}
 and so the internal vertices of $C_3$ can have node-weight at most $\lambda/2$
 (the total node-weight of $C$ being at most $2\lambda$).

We now count the number of possible signature vectors.  As for each $j$ there are at most $n^2$ ways of choosing the
incoming edge, at most $n^2$ ways of choosing the outgoing edge, and $i_j$ can
have two different parities, the number of possible signatures is (summing over the choices of $t=0,1,2,3$) at most
$ \rb{2n^2\cdot n^2} + \rb{2n^2\cdot n^2}^2 + \rb{2n^2\cdot n^2}^3 + \rb{2n^2\cdot n^2}^4  < n^{17}$.

It remains to be shown that any
two alternating circuits $C$ and $D$ in $H$ of node-weight at most $2\lambda$ have different
signatures if $\pmind{C} \neq \pmind{D}$. Suppose that  $\pmind{C} \ne \pmind{D}$ but $\sigma(C) = \sigma(D)$.
We would like to derive a contradiction with the assumption (in \cref{circuit-removal}) that $H$ contains no alternating circuit of node-weight at most $\lambda$.
This will finish the proof.

As described above, $C$ can be partitioned into disjoint paths
$C_0, \dots, C_t$ using its indices $i_0, \dots, i_t$. 
Similarly we partition $D$ into $D_0, \dots, D_t$.
Since these
are disjoint unions, $\pmind{C} \ne \pmind{D}$ implies that at least one of the
four subpaths must be different between $C$ and $D$, in the sense that the part
of the alternating indicator vector arising from that subpath is different.
More formally,  let $b_j$ denote the parity (i.e., the first element) of the $j$-th tuple in the signatures $\sigma(C) = \sigma(D)$. Then we have $\pmind{C} = \sum_{j=0}^t b_j \cdot \pmind{C_j}$ and $\pmind{D} = \sum_{j=0}^t b_j \cdot \pmind{D_j}$. Therefore, as $\pmind{C} \neq \pmind{D}$, there must be a $j\in \{0,\dots, t\}$ such that $\pmind{C_j} \neq
\pmind{D_j}$.
We will ``glue'' together the paths $C_j$ and $D_j$ to obtain
another alternating circuit $B$.

First notice that  both $C_j$ and $D_j$ are
paths of the form $v_{i_j} = a \rightarrow  b  \rightarrow \dots \rightarrow
c \rightarrow  d = v_{i_{(j+1) \mymod t}}$, where the segment from $b$ to $c$ differs between them.\footnote{
Here again we slightly abuse notation since $i_j$ might differ between $C$ and $D$;
however, the vertex $v_{i_j}$ does not,
because it is the tail of $e_{i_j}$, which is part of the signature $\sigma(C) = \sigma(D)$.
The same applies to $v_{i_{j+1 \mymod t}}$.
}
This
follows from the assumption that $\sigma(C) = \sigma(D)$. Let $P_C$ denote the
path from $b$ to $c$ in $C_j$ and let $P_D$ denote the path from $b$ to $c$ in
$D_j$. As the parity fields of the signatures agree,  we have that  $|P_C|
+ |P_D|$ is even. Now let $B$ be the cyclic walk of even length obtained by walking from
$b$ to $c$ along the path $P_C$ and back from $c$ to $b$ along the path $P_D$ (in
reverse). That is, $B$ is of the form (see also the right part of~\cref{fig:signature})
\begin{align*}
  b \xrightarrow{f_1} \dots \xrightarrow{f_{|P_C|}} c \xrightarrow{g_1} \dots \xrightarrow{g_{|P_D|}} b\,,
\end{align*}
where we let $f_1, \dots, f_{|P_C|}$ denote the edges of the path $P_C$ and $g_1, \dots, g_{|P_D|}$ denote the edges of the reversed path $P_D$.  
To verify that $B$  is an alternating circuit  we need to show that its alternating indicator vector is nonzero:
\begin{align*}
  - \pmind{B} &= \sum_{i=1}^{|P_C|} (-1)^{i} \one_{f_i} + \sum_{i=1}^{|P_D|} (-1)^{|P_C| + i} \one_{g_i} \\
  & =  \underbrace{\left((-1)^0\one_{\outgoing(a)} + \sum_{i=1}^{|P_C|} (-1)^{i} \one_{f_i} + (-1)^{|P_C| + 1} \one_{\incoming(d)}\right)}_{= \pmind{C_j}}\\
  &  \qquad \qquad \qquad + \underbrace{\left( (-1)^{|P_C| + 2} \one_{\incoming(d)} + \sum_{i=1}^{|P_D|} (-1)^{|P_C| + 2 + i} \one_{g_i} + (-1)^{|P_C| + |P_D| + 3} \one_{\outgoing(a)} \right)}_{= -\pmind{D_j}}.
\end{align*}
The second equality is easiest to see
by mentally extending $B$ from a circuit
$b \to ... \to c \to ... \to b$
to
$a \to b \to ... \to c \to d \to c \to ... \to b \to a$.
Also recall that $|P_C|+|P_D|$ is even.
Thus we get $\pmind{B} = - \pmind{C_j} + \pmind{D_j}$, which is nonzero by the choice of $j$.
Finally, the node-weight of $B$ is at most the node-weight of the internal nodes of path $C_j$ plus the node-weight of the internal nodes of path $D_j$ and thus at most $\lambda/2 + \lambda/2 = \lambda$.

We have thus shown that $B$ is a nonempty cyclic walk of even length whose alternating indicator vector is nonzero -- thus an alternating circuit -- and whose node-weight is at most $\lambda$. This contradicts our assumption on $H$.
\end{proof}

Now we have all the tools needed to prove the main result of this section.

\begin{proof}[Proof of \cref{circuit-removal}]
Let us fix some $w \in \cW$.
We want to articulate conditions on $w$ which will make sure that the statement is satisfied.
Then we show that some $w \in \cW$ satisfies these conditions.

Let $C$ be any alternating circuit in $H$ of node-weight at most $2 \lambda$.
Our condition on $w$ will be that all such circuits $C$ should not respect $\facemin{F}{w}$, i.e., that all vectors from the set
\[ \cZ := \{ \pmind{C} : C \text{ is an alternating circuit in $H$ of node-weight at most $2 \lambda$} \} \]
should not respect $\facemin{F}{w}$.
We use \cref{killing_cycle} to transform each $z \in \cZ$ ($z \in \bZ^{E(H)}$) to a vector $y = y(z) \in \bZ^E$ such that if $y(z)$ does not respect $\facemin{F}{w}$, then $z$ does not respect $\facemin{F}{w}$.
Let $\cY = \{ y(z) : z \in \cZ \}$.
Clearly $|\cY| \le |\cZ|$ (actually $|\cY| = |\cZ|$ since the mapping $z \mapsto y(z)$ is one-to-one),
and $|\cZ| \le n^{17}$ by \cref{counting}.
Moreover, since the alternating circuits $C$ were of node-weight at most $2 \lambda \le 4 n$, we have $\|z \|_1 \le 4n$ for $z \in \cZ$ and $\| y \|_1 \le 4n^2$ for $y \in \cY$.
Now it is enough to apply \cref{killing_set_of_cycles} to obtain that there exists $w \in \cW(n^3 \cdot n^{17}) = \cW$ such that each $y \in \cY$ does not respect 
the face $\facemin{F}{w}$,
and thus each $z \in \cZ$ does not respect $\facemin{F}{w}$.
\end{proof}

\subsection{The existence of a good weight function}
In this section, we use \cref{circuit-removal} to prove the existence of
a weight function defining a face $\Fout$ with the desired properties (so
as to be  the face in our $2\lambda$-good face-laminar pair), namely:
\begin{theorem}
  Let $(\Fin, \Lin)$ be a $\lambda$-good face-laminar pair.  Then there exists a weight function $\wout \in \cWf^{\log_2(n) + 1 }$ such that
  the face $\Fout = \facemin{\Fin}{\wout}$ satisfies:
  \begin{itemize}
    \item[(i)'] For each $S \in \Lin$ with $|S| \le 2 \lambda$, $S$ is $\Fout$-contractible.
    \item[(ii)'] In the $(\Fout, \Lin, 2 \lambda)$-contraction of $G$, there is no $\Fout$-respecting alternating circuit of node-weight at most $2 \lambda$.
  \end{itemize}
  \label{good-weight-fn}
\end{theorem}

Throughout this section, $\Fin$ and $\Lin$ are as in the statement of \cref{good-weight-fn}.
The proof of~\cref{good-weight-fn} is based on the following technical lemma.
\begin{lemma}
 There exists a weight function
 $\wmid \in \cWf^{\log_2 n}$ such that the face  $\Fmid = \facemin{\Fin}{\wmid}$ satisfies: 
  \begin{itemize}
    \item[(i)'] For each $S \in \Lin$ with $|S| \le 2 \lambda$, $S$ is $\Fmid$-contractible.
  \end{itemize}
  \label{making-2l-contractible}
\end{lemma}
Before giving the proof of~\cref{making-2l-contractible}  let
us see how it, together with~\cref{circuit-removal}, readily implies~\cref{good-weight-fn}.
In short, once we have made sets of size up to $2 \lambda$ contractible
using \cref{making-2l-contractible},
we are only left with removing alternating circuits
of node-weight between $\lambda$ and $2 \lambda$.
One application of \cref{circuit-removal}
is enough to achieve this.

\begin{proof}[Proof of \cref{good-weight-fn}]
  \cref{making-2l-contractible} says that
  there is a weight function $\wmid \in \cWf^{\log_2(n)}$  such that
  the face $\Fmid = \facemin{\Fin}{\wmid}$ satisfies that every $S\in \Lin$ with $|S| \leq 2 \lambda$ is
  $\Fmid$-contractible.
  As we will obtain $\Fout$ as a subface of $\Fmid$,
  we have
  thus proved point $(i)'$ of \cref{good-weight-fn}, as any set that is $\Fmid$-contractible will remain
  contractible in any subface of $\Fmid$ (by \cref{sub-face-contractable}).

  By the above, every vertex in the $(\Fmid, \Lin,
  2\lambda)$-contraction of $G$ corresponds to an $\Fmid$-contractible set.
  Moreover, by the assumption that the face-laminar pair $(\Fin, \Lin)$ is
  $\lambda$-good, the $(\Fin, \Lin, \lambda)$-contraction of $G$ does not have any
  alternating circuits of node-weight at most $\lambda$.
  This implies
  that the $(\Fmid, \Lin, 2\lambda)$-contraction of $G$ does not have any such alternating
  circuits.
  For suppose $C$ were one.
  Let $S_1, ..., S_k$ be maximal sets of size at most $2 \lambda$ in $\Lin$, i.e., the vertices of the $(\Fmid, \Lin, 2 \lambda)$-contraction of $G$.
  Note that $C$ cannot cross a set $S_i$ with $|S_i| > \lambda$, because then its node-weight would be larger than $\lambda$.
  Therefore $C$ only crosses sets $S_i$ with $|S_i| \le \lambda$.
  Thus $C$ also appears in the $(\Fmid, \Lin, \lambda)$-contraction of $G$, with the same node-weight,
  and in the $(\Fin, \Lin, \lambda)$-contraction (of which the $(\Fmid, \Lin, \lambda)$-contraction is a subgraph) as well.
  This is a contradiction.
  
  We can thus apply~\cref{circuit-removal} with $\beta = 2 \lambda$ to
  the face-laminar pair $(\Fmid, \Lin)$. We
  get a weight function $w \in \cWf$ such that
  the $(\Fmid, \Lin, 2 \lambda)$-contraction
  has no $\facemin{\Fmid}{w}$-respecting
  alternating circuit of node-weight at most $2 \lambda$.
  Therefore, as the $(\facemin{\Fmid}{w}, \Lin, 2\lambda)$-contraction is a subgraph of the $(\Fmid, \Lin, 2\lambda)$-contraction,
  the face $\Fout = \facemin{\Fmid}{w}$
  satisfies $(ii)'$. Selecting $\wout = \wmid \circ
  w \in \cWf^{\log_2(n) + 1}$ completes the proof (by \cref{concatenation_and_subface}).  \qedmanual 
\end{proof}

The rest of this section is devoted to the proof of~\cref{making-2l-contractible}.
Recall that we need to prove the existence of a weight function $\wmid \in \cWf^{\log_2(n)}$  satisfying  $(i)'$, i.e., that
 \begin{align*}
   \label{eq:cond1} 
   \mbox{for each $S \in \Lin$ with $|S| \le 2 \lambda$, $S$ is $\Fmid$-contractible,}
 \end{align*}
 where $\Fmid = \facemin{\Fin}{\wmid}$.
First note that the statement will be true for every $S\in \Lin$ with $|S| \le
\lambda$, regardless of the choice of the weight function $\wmid$. Indeed, by
assumption $(\Fin, \Lin)$ is $\lambda$-good and so $S$ is $\Fin$-contractible. Thus,
by \cref{sub-face-contractable}, $S$ remains $\Fmid$-contractible for any
subface $\Fmid \subseteq\Fin$.  

It remains to deal with the sets $S\in \Lin$ with $\lambda < |S| \leq
2\lambda$.  Let $\cL = \{S\in \Lin :\lambda<  |S| \leq 2\lambda\}$ be the
laminar family $\Lin$ restricted to these sets. Notice that any set in $\cL$ can have
at most one child in $\cL$ due to the cardinality
constraints. In other words, $\cL$ consists of disjoint chains, as depicted
in~\cref{fig:laminar}.

\begin{figure}[t]
  \begin{center}
    \begin{tikzpicture}[scale=0.9]
    \tikzstyle{vertex}=[circle, fill=black, minimum size=2,inner sep=1pt]
      \draw (-5,0) ellipse (1.5cm and 1cm) node[above right = 0.63cm and 0.63cm] {\scriptsize $S^{(1)}_1$};
      \draw[fill=gray!20!white] (6.5,0) ellipse (2cm and 1.4cm)node[above right = 0.99cm and 0.63cm] {\scriptsize $S^{(3)}_2$};
      \draw[fill=white] (6,0) ellipse (1.4cm and 0.9cm)node[above right = 0.54cm and 0.54cm] {\scriptsize $S^{(3)}_1$};
      \draw[fill=gray!40!white] (0.9,0) ellipse (2.5cm and 1.7cm)node[above right = 1.17cm and 0.99cm] {\scriptsize $S^{(2)}_3$};
      \draw[fill=gray!20!white] (0.55,0) ellipse (2cm and 1.3cm)node[above right = 0.825cm and 0.72cm] {\scriptsize $S^{(2)}_2$};
      \draw[fill=white] (0,0) ellipse (1.3cm and 0.8cm)node[above right = 0.45cm and 0.54cm] {\scriptsize $S^{(2)}_1$};
      \node[vertex] (e1a) at (0.5,0) {};
      \node[vertex] (e1b) at (1.4,-0.5) {};
      \draw (e1a) edge node[above right = -0.1cm and -0.2cm] {\tiny $e_1$} (e1b);
      \node[vertex] (e3a) at (2.9,-0.6) {};
      \node[vertex] (e3b) at (3.8,-1.3) {};
      \draw (e3a) edge node[above right = -0.1cm and -0.2cm] {\tiny $e_3$} (e3b);
    \end{tikzpicture}
  \end{center}
  \caption{An example of the laminar family $\cL$, which consists of disjoint chains. The different shades of gray depict the sets $U_p^{(i)}$. }
  \label{fig:laminar}
\end{figure}
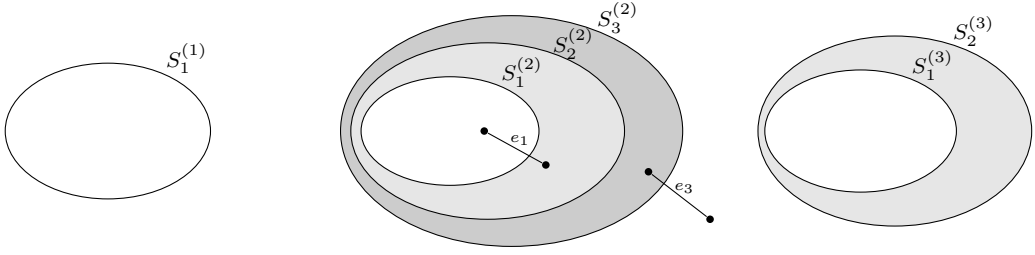

\paragraph{Notation.} We refer to the sets in the $i$-th chain  of $\cL$ by $\cL_{i}$. 
Let $\ell_i = |\cL_i|$ and index the sets of the chain $\cL_i = \{S^{(i)}_1, S^{(i)}_2,
\ldots, S^{(i)}_{\ell_i}\}$ so that  $S^{(i)}_1 \subseteq S^{(i)}_2 \subseteq \cdots \subseteq
S^{(i)}_{\ell_i}$. Let $U^{(i)}_1 = S^{(i)}_1$ and $U^{(i)}_{p}
= S^{(i)}_p\setminus S^{(i)}_{p-1}$ for $p=2, 3, \ldots, \ell_i$. Also define
$U^{(i)}_{p, r} = U^{(i)}_p \cup U^{(i)}_{p+1} \cup \cdots \cup U^{(i)}_{r}$.

\paragraph{}Recall that a set
$U^{(i)}_{1,r} = S^{(i)}_r$ is defined to be $F$-contractible if for every $e_r \in
\delta(U^{(i)}_{1,r})$ there are no two perfect matchings in $F$ which both
contain $e_r$ and are different inside $U^{(i)}_{1,r}$ (\cref{def:contractable}). 
For the proof of~\cref{making-2l-contractible}, we generalize this definition to
also include sets $U^{(i)}_{p,r}$ with $p\geq 2$. 
\begin{definition}
  Consider a face $F$. We say that a set
  $U^{(i)}_{p,r}$ with $2\leq p \leq r \leq \ell_i$ is \emph{$F$-contractible} if for every 
  $e_{p-1}\in \delta(S^{(i)}_{p-1})$ and $e_r \in \delta(S^{(i)}_{r})$ there are no two perfect matchings in $F$ which both contain $e_{p-1}$ and $e_r$ and are different inside $U^{(i)}_{p,r}$. (It is possible that $e_{p-1} = e_r$, in which case neither endpoint of this edge lies in $U^{(i)}_{p,r}$.)
\end{definition}
The intuition of this definition is similar to that of~\cref{def:contractable}. Consider the second chain in~\cref{fig:laminar} for an example.  If
we restrict our attention to perfect matchings that contain edges $e_1 \in
\delta(S^{(2)}_1)$ and $e_3 \in \delta(S^{(2)}_3)$, then, as $S^{(2)}_1$
are $S^{(2)}_3$ are tight sets, the task of selecting such a matching
decomposes into two \emph{independent} problems:  the problem of selecting
a perfect matching in $U_{2,3}^{(2)}$ (ignoring the vertices incident to $e_1$
and $e_3$) and the problem of selecting a perfect matching in $V\setminus
U_{2,3}^{(2)}$ (again ignoring the vertices incident to $e_1$ and $e_3$).  

The proof now proceeds iteratively as follows. 
\begin{itemize}
  \item First we select $w_1 \in \cWf$ such that $F_1 =  \facemin{\Fin}{w_1}$ satisfies: 
      \begin{align}
        \label{eq1}
        \mbox{$U^{(i)}_p$ is $F_1$-contractible
      for all chains $i$ and $1\leq p \leq \ell_i$.}
    \end{align}
  \item For $t=2, 3,\ldots,  \log_2(n)$ we select $w_t \in \cWf$ such that
$F_t = \facemin{F_{t-1}}{w_t}$ satisfies: 
\begin{align}
  \label{eqt}
  \mbox{$U^{(i)}_{p,r}$
is $F_t$-contractible for all chains $i$ and $1\leq p \leq r \leq \ell_i$ with
$r-p\leq 2^{t-1}-1$.}
\end{align}
\end{itemize}
We remark that, having selected $w_1, w_2, \ldots, w_{\log_2(n)}$ as above, if
we let $\wmid = w_1 \circ w_2 \circ \ldots \circ w_{\log_2(n)} \in \cWf^{\log_2(n)}$,  then the
face $\Fmid  = \facemin{\Fin}{\wmid}$ equals $F_{\log_2(n)}$ (by \cref{concatenation_and_subface}). To see
that this completes the proof of~\cref{making-2l-contractible}, note that  $\ell_i < n/2$ for any chain $i$
since $|S^{(i)}_1| > \lambda = (2 \lambda)/2 \geq |S^{(i)}_{\ell_i}|/2$ and $|S^{(i)}_{\ell_i}| \leq n$.
Therefore any set $S^{(i)}_r \in \cL$ has $r\leq n/2$ and so, by
\eqref{eqt}, $S^{(i)}_r = U^{(i)}_{1,r}$ is
$\Fmid$-contractible.  

In what follows, we complete the proof of~\cref{making-2l-contractible} with a description of how to select $w_1$, followed by the selection of $w_t$, $t=2, 3,\ldots,  \log_2(n)$, in  the iterative case.

\subsubsection{The selection of $w_1$}

The following claim allows us to use~\cref{circuit-removal} to show the existence of a weight function $w_1$ satisfying~\eqref{eq1}.
\begin{claim}
  \label{F1}
  If the $(\Fin, \Lin, \lambda)$-contraction of $G$ has no $F_1$-respecting
  alternating circuit of node-weight at most $2\lambda$, then every $U_p^{(i)}$
  is $F_1$-contractible.
\end{claim}

The claim together with~\cref{circuit-removal} completes the selection of $w_1$
as follows. Since $(\Fin, \Lin)$ is $\lambda$-good, we can
apply~\cref{circuit-removal} with $\beta = \lambda$ to obtain the existence of
a weight function $w_1\in \cWf$ such that in  the $(\Fin, \Lin,
\lambda)$-contraction of $G$ there is no $F_1$-respecting alternating circuit of
node-weight at most $2\lambda$, where $F_1 = \facemin{\Fin}{w_1}$.
Hence, by the above claim, every $U_p^{(i)}$ is $F_1$-contractible as
required. 

\begin{proof}[Proof of \cref{F1}]
This proof resembles that of \cref{finalization}.
Fix $U_p^{(i)}$ and let $e_{p-1} \in \delta(S^{(i)}_{p-1})$ and $e_p \in \delta(S^{(i)}_{p})$.\footnote{Here and in \cref{select_wt}, we abuse notation and assume $p\geq 2$. The only difference is that, given a set $U^{(i)}_{p,r}$ with $p=1$ (in this section $p=r$), we consider matchings containing one edge $e_r \in \delta(S^{(i)}_r)$ instead of matchings containing two edges $e_{p-1} \in \delta(S^{(i)}_{p-1})$ and $e_r \in \delta(S^{(i)}_r)$, since if $p=1$, the set $S^{(i)}_{p-1}$ is not defined.}
Suppose that $M_1$ and $M_2$ are two perfect matchings in $F_1$ that both contain $e_{p-1}$ and $e_p$.
We want to show that $M_1$ and $M_2$ are equal inside $U_p^{(i)}$.

\begin{figure}[t]
  \begin{center}
    \tikzset{external/export next=false} 
    \tikzstyle{vertex}=[circle, fill=black, minimum size=2,inner sep=1pt]
    \begin{tikzpicture}
      \begin{scope}
        \clip(1.5,-2) rectangle (5.2,3);
        \draw[fill=gray!10!white, draw=gray!80!black] (-0.7,-2.3) ellipse (5.7cm and 5cm);
      \end{scope}
      \begin{scope}
        \clip(1.49,-2.01) rectangle (5.2,2.01);
        \draw[fill=white, draw=gray!80!black] (-0.5,-2.5) ellipse (3cm and 2.7cm);
      \end{scope}
      \node at (2.6, -2.5) {\scriptsize $S^{(i)}_{p-1}$};
      \node at (5.0, -2.5) {\scriptsize $S^{(i)}_{p}$};
      \node[vertex] (e1) at (2.2,-0.8) { };  
      \node[vertex,fill=white, draw=white] (e1o) at (1.4,-1.4) { };  
      \draw (e1) edge[bend left = 10] node[below right = 0.1cm and -0.35cm] {\tiny $e_{p-1}$} (e1o);
      \draw(e1) edge[decorate,decoration={snake,amplitude=.3mm,segment length=2pt,post length=0mm}, bend right=10] (e1o);
      \draw[fill = gray!30!white, rotate around={115:(2.2,-0.8)}] (2.2, -0.8) ellipse (0.2cm and 0.15cm);

      \node[vertex] (e2) at (2.50,1.55) { };  
      \node[vertex,fill=white,draw=white] (e2o) at (3.4,2.1) { };  
      \draw (e2) edge[bend right = 10] node[below right = -0.2cm and 0cm] {\tiny $e_{p}$} (e2o);
      \draw(e2) edge[decorate,decoration={snake,amplitude=.3mm,segment length=2pt,post length=0mm}, bend left=10] (e2o);
      \draw[fill = gray!30!white, rotate around={115:(2.3,1.5)}]  (2.3, 1.5) ellipse (0.2cm and 0.35cm); 

      \begin{scope}[xshift=0.2cm,yshift=-0.3cm]
        \node[vertex] (a1) at (2,0.8) { };  
        \node[vertex] (a2) at (2.5,0.8) { };  
        \draw (a1) edge[bend right = 20]  (a2);
        \draw(a1) edge[decorate,decoration={snake,amplitude=.3mm,segment length=2pt,post length=0mm}, bend left=20] (a2);
        \draw[fill = gray!30!white, rotate around={0:(2.3,1.5)}]  (2, 0.8) ellipse (0.1cm and 0.15cm);
        \draw[fill = gray!30!white, rotate around={45:(2.5,0.8)}]  (2.5, 0.8) ellipse (0.12cm and 0.1cm);
      \end{scope}
      \begin{scope}[xshift=0.1cm, yshift=0.5cm]
        \node[vertex] (v1) at (3, -1) {};
        \node[vertex] (v2) at (3.7, -0.5) {};
        \node[vertex] (v3) at (3.9, -1.2) {};
        \node[vertex] (v5) at (3.65, -1.5) {};
        \node[vertex] (v6) at (3, -1.7) {};
        \draw (v1) edge[ultra thick] (v2);
        \draw (v2) edge[decorate,decoration={snake,amplitude=.3mm,segment length=4pt,post length=0mm}, ultra thick] (v3);
        \draw (v5) edge[ultra thick] (v6);
        \draw (v6) edge[decorate,decoration={snake,amplitude=.3mm,segment length=4pt,post length=0mm}, ultra thick] (v1);
        \draw[fill = gray!30!white, rotate around={53:(3.9,-1.4)}] (3.9, -1.4) ellipse (0.4cm and 0.25cm);
        \draw[fill = gray!30!white, rotate around={45:(3,-1.7)}]  (3, -1.7) ellipse (0.12cm and 0.1cm);
        \draw[fill = gray!30!white, rotate around={120:(3,-1)}]  (3, -1) ellipse (0.22cm and 0.14cm);
        \draw[fill = gray!30!white, rotate around={45:(3.7,-0.5)}]  (3.7, -0.5) ellipse (0.13cm and 0.15cm);
      \end{scope}
        \draw[rounded corners=5pt] (1.0,-3) rectangle (6,3);
        \node at (3.5, -3.3) {\small \cref{F1}};
      \begin{scope}[xshift=5.5cm]
        \begin{scope}
          \clip(1.5,-2) rectangle (5.8,3);
          \draw[fill=gray!30!white, draw=gray!80!black] (-0.7,-2.3) ellipse (6cm and 5.5cm);
          \draw[fill=gray!20!white, draw=gray!80!black, dashed] (-0.7,-2.3) ellipse (5.2cm and 4.7cm);
          \draw[fill=gray!10!white, draw=gray!80!black, dashed] (-0.7,-2.3) ellipse (4.2cm and 3.7cm);
        \end{scope}
        \begin{scope}
          \clip(1.49,-2.01) rectangle (5.8,2.01);
          \draw[fill=white, draw=gray!80!black] (-0.5,-2.5) ellipse (3cm and 2.7cm);
        \end{scope}
        \node at (2.6, -2.5) {\scriptsize $S^{(i)}_{p-1}$};
        \node at (5.4, -2.5) {\scriptsize $S^{(i)}_{r}$};

        \node[vertex] (e1) at (1.9,-0.4) { };  
        \node[vertex] (e1o) at (1.5,-0.8) { };  
        \draw (e1) edge[bend left = 20, ultra thick] node[below right = -0.1cm and 0cm] {\tiny $e_{p-1}$} (e1o);
        \draw(e1) edge[decorate,decoration={snake,amplitude=.3mm,segment length=2pt,post length=0mm}, bend right=20] (e1o);

        \node[vertex] (a1) at (2.7,-1.4) { };  
        \node[vertex] (a2) at (2.9,-1.0) { };  
        \node[vertex] (a3) at (3.4,-0.8) { };  
        \node[vertex] (a4) at (3.9,-0.8) { };  
        \node[vertex] (a5) at (4.1,-1.2) { };  
        \node[vertex] (a6) at (4.1,-1.6) { };  
        \node[vertex] (a7) at (3.7,-1.7) { };  
        \node[vertex] (a8) at (3.1,-1.7) { };  
        \draw (a1) edge (a2);
        \draw (a2) edge[decorate,decoration={snake,amplitude=.3mm,segment length=4pt,post length=0mm}, ultra thick] (a3);
        \draw (a3) edge (a4);
        \draw (a4) edge[decorate,decoration={snake,amplitude=.3mm,segment length=4pt,post length=0mm}, ultra thick] (a5);
        \draw (a5) edge (a6);
        \draw (a6) edge[decorate,decoration={snake,amplitude=.3mm,segment length=4pt,post length=0mm}, ultra thick] (a7);
        \draw (a7) edge (a8);
        \draw (a8) edge[decorate,decoration={snake,amplitude=.3mm,segment length=4pt,post length=0mm}, ultra thick] (a1);
        \node[vertex] (b1) at (2.3,1.25) { };  
        \node[vertex] (b2) at (2.6,1.65) { };  
        \draw (b1) edge[bend right = 20]  (b2);
        \draw(b1) edge[decorate,decoration={snake,amplitude=.3mm,segment length=4pt,post length=0mm}, bend left=20, ultra thick] (b2);
        \node[vertex] (e2) at (4.8,-0.9) { };  
        \node[vertex] (e2o) at (5.4, -0.8) { };  
        \draw (e2) edge[bend right = 20,ultra thick] node[below right =-0.05cm and 0cm] {\tiny $e_{r}$} (e2o);
        \draw(e2) edge[decorate,decoration={snake,amplitude=.3mm,segment length=2pt,post length=0mm}, bend left=20] (e2o);
        \node[vertex] (b1) at (4.0, 2.4) { };  
        \node[vertex] (b2) at (4.5, 2.4) { };  
        \node[vertex] (b3) at (4.5, 1.9) { };  
        \node[vertex] (b4) at (4.0, 1.9) { };  
        \draw (b1) edge[ultra thick] (b2);
        \draw (b2) edge[decorate,decoration={snake,amplitude=.3mm,segment length=2pt,post length=0mm}] (b3);
        \draw (b3) edge[ultra thick] (b4);
        \draw (b4) edge[decorate,decoration={snake,amplitude=.3mm,segment length=2pt,post length=0mm}] (b1);
        \draw[rounded corners=5pt] (1.0,-3) rectangle (6,3);
        \node at (3.5, -3.3) {\small \cref{match-same-weight}};
      \end{scope}
      \begin{scope}[xshift=11cm]
        \begin{scope}
          \clip(1.5,-2) rectangle (5.8,3);
          \draw[fill=gray!30!white, draw=gray!80!black] (-0.7,-2.3) ellipse (6cm and 5.5cm);
          \draw[fill=gray!30!white, draw=gray!80!black, dashed] (-0.7,-2.3) ellipse (5.55cm and 5.05cm);
          \draw[fill=gray!30!white, draw=gray!80!black, dashed] (-0.7,-2.3) ellipse (5.0cm and 4.5cm);
          \draw[fill=gray!10!white, draw=gray!80!black] (-0.7,-2.3) ellipse (4.5cm and 4.1cm);
          \draw[fill=gray!10!white, draw=gray!80!black, dashed] (-0.7,-2.3) ellipse (4.0cm and 3.65cm);
          \draw[fill=gray!10!white, draw=gray!80!black, dashed] (-0.7,-2.3) ellipse (3.6cm and 3.20cm);
        \end{scope}
        \begin{scope}
          \clip(1.49,-2.01) rectangle (5.8,2.01);
          \draw[fill=white, draw=gray!80!black] (-0.5,-2.5) ellipse (3cm and 2.7cm);
        \end{scope}
        \node at (2.6, -2.5) {\scriptsize $S^{(i)}_{p-1}$};
        \node at (3.9, -2.5) {\scriptsize $S^{(i)}_{q}$};
        \node at (5.4, -2.5) {\scriptsize $S^{(i)}_{r}$};

        \node[vertex] (e1) at (1.9,-0.4) { };  
        \node[vertex] (e1o) at (1.5,-0.8) { };  
        \draw (e1) edge node[below right = -0.1cm and -0.05cm] {\tiny $e_{p-1}$} (e1o);

        \node[vertex] (e2) at (4.8,-0.9) { };  
        \node[vertex] (e2o) at (5.4, -0.8) { };  
        \draw (e2) edge node[below right =-0.05cm and 0cm] {\tiny $e_{r}$} (e2o);

        \node[vertex] (a1) at (3.1,-0.6) { };  
        \node[vertex] (a2) at (3.6, -0.4) { };  
        \draw (a1) edge node[below right =-0.05cm and 0cm] {\tiny $e_{q}$} (a2);

        \draw[rounded corners=5pt] (1.0,-3) rectangle (6,3);
        \node at (3.5, -3.3) {\small \cref{matching-count}};

      \end{scope}
    \end{tikzpicture}
  \end{center}
  \caption{An illustration of the claims used in the proof
    of~\cref{making-2l-contractible}. 
    \\
    \textbf{\cref{F1}}: Straight and swirly edges denote
    $M_1$ and $M_2$ respectively. The thick edges denote the alternating cycle. The dark-gray sets are $S_1, ..., S_k$.
\\
\textbf{\cref{match-same-weight}}: Straight and swirly edges denote $M_1$ and $M_2$ respectively. The thick edges denote $M_{12}$, which agrees with $M_1$ outside $U^{(i)}_{p,r}$ and with $M_2$ inside $U^{(i)}_{p,r}$.
\\
\textbf{\cref{matching-count}}: The divide-and-conquer argument is illustrated (only edges $e_{p-1}$, $e_q$, and $e_r$ are depicted). After fixing $e_{p-1}$ and $e_q$, the matching in the light-gray area is unique in the face $F_{t-1}$. Similarly, after fixing $e_q$ and $e_r$, the matching in the dark-gray area is unique in the face $F_{t-1}$.
Thus, for each choice of $e_{p-1}$ and $e_r$, there can be at most one matching inside $U^{(i)}_{p, r}$ for each possible way of fixing $e_q$.  
It follows that there are at most $n^2$ matchings inside $U^{(i)}_{p, r}$ that contain $e_{p-1}$ and $e_r$.
}
\label{fig:claims}
\end{figure}

Let $S_1, ..., S_k$ be all maximal sets $S \in \Lin$ with $S \subseteq U_p^{(i)}$.
They are those vertices of the $(\Fin, \Lin, \lambda)$-contraction which lie in $U_p^{(i)}$, and we have $S_1 \cup ... \cup S_k = U_p^{(i)}$.
Because these sets, as well as $S_{p-1}^{(i)}$ and $S_p^{(i)}$, are tight for $F_1$, any perfect matching  in $F_1$ containing $e_{p-1}$ and $e_p$ induces an almost-perfect matching on $S_1, ..., S_k$,
that is, one where only the (up to two) sets $S_i$ containing endpoints of $e_{p-1}$ and $e_p$ are unmatched (see the left part of \cref{fig:claims}).

If the matchings induced by $M_1$ and $M_2$ were different, then their symmetric difference would induce an alternating simple cycle $C$
in the graph obtained from $\suppo(\Fin)$ by contracting $S_1, ..., S_k$.
The cycle $C$ would also be present in the $(\Fin, \Lin, \lambda)$-contraction.
This is because $S_1, ..., S_k$ are maximal sets of size at most $\lambda$ in $\Lin$,
so they are indeed vertices of the contraction,
and because the edges between these vertices
are not on the boundary $\delta(T)$ of any $T \in \Lin$ with $|T| > \lambda$
(in which case they would be missing from the contraction),
which in turn follows by laminarity of $\Lin$
and the definition of $U^{(i)}_p$.

Since $C$ arises from two matchings in $F_1$, it respects $F_1$.
This follows by a similar argument as the proof of \cref{claim:y_respects},
which we repeat here for completeness.
Clearly, $\supp(C) \subseteq M_1 \cup M_2 \subseteq \suppo(F_1)$.
Let $T \in \tight(F_1)$ be a set tight for $F_1$ which is a union of the vertices of the contraction; we want to show that $\ab{\pmind{C}, \one_{\delta(T)}} = 0$.
Because $C$ is a cycle in the contraction and $|M_1 \cap \delta(T)| = |M_2 \cap \delta(T)| = 1$,
either $C$ has no edge in $\delta(T)$ or it has two, one from $M_1$ and one from $M_2$ (and they cancel out).

Moreover, since $C$ is a simple cycle inside $U^{(i)}_p$, its node-weight is at most $|S_1| + \ldots + |S_k| = |U^{(i)}_p| \le 2 \lambda$.
This contradicts our assumption.

Therefore, the induced matchings must be equal.
Moreover, the sets $S_1, ..., S_k$ are $F_1$-contractible, since they are vertices of the $(\Fin, \Lin, \lambda)$-contraction, $(\Fin,\Lin)$ is $\lambda$-good, and $F_1 \subseteq \Fin$.
This means that, given the boundary edges (i.e., the induced matching plus $e_{p-1}$ and $e_p$),
there is a unique perfect matching in $F_1$ inside each $S_i$.
It follows that $M_1$ and $M_2$ are equal inside $U_p^{(i)}$.
\end{proof}

\subsubsection{The selection of $w_t$ for $t= 2, 3, \ldots, \log_2(n)$} \label{select_wt}
In this section we show the existence of a weight function $w_t \in \cWf$ satisfying~\eqref{eqt}, i.e.,  
\begin{align*}
  \mbox{$U^{(i)}_{p,r}$ is $F_t$-contractible for all chains $i$ and $1\leq p \leq r \leq \ell_i$ with $r-p\leq 2^{t-1}-1$,}
\end{align*}
where 
$F_t = \facemin{F_{t-1}}{w_t}$.

The proof outline is as follows.
First, in~\cref{match-same-weight}, we give sufficient conditions on $w_t$ for  $U^{(i)}_{p,r}$ to be $F_t$-contractible.
They are given as a system of linear non-equalities with coefficients in $\{-1,0,1\}$.
Then, in~\cref{matching-count}, we upper-bound the number of these non-equalities by $n^{11}$.
This allows us to deduce the existence of $w_t \in \cWf$ by applying~\cref{lem23}.

The following claim gives sufficient linear non-equalities on $w_t$
for every $U^{(i)}_{p,r}$ to be $F_t$-contractible (one non-equality for each choice of $U^{(i)}_{p,r}$, $e_{p-1}$, $e_r$, $M_1$ and $M_2$). 
\begin{claim}
  \label{match-same-weight}
  Let $ U= U^{(i)}_{p,r}$ for some chain $i$ and $1\leq p\leq r \leq \ell_i$. 
  Suppose that for every two edges $e_{p-1} \in \delta(S^{(i)}_{p-1})$ and $e_{r} \in
  \delta(S^{(i)}_{r})$  defining a face $F= \{x\in F_{t-1}: x_{e_{p-1}} = 1,
  x_{e_{r}} = 1\}$ we have:
  \begin{align*}
    w_t(M_1 \cap E(U)) \neq w_t(M_2 \cap E(U))\quad \mbox{for any two matchings $M_1, M_2$ in $F$ that differ inside $U$.} 
  \end{align*}
  Then  $U$ is $F_t$-contractible. 
\end{claim}
\begin{proof}
  We prove the contrapositive.  Suppose that $U$ is not $F_t$-contractible.
  Then, by definition, there must be $e_{p-1} \in \delta(S^{(i)}_{p-1})$ and
  $e_{r} \in \delta(S^{(i)}_{r})$ that define a face $F'= \{x\in F_{t}: x_{e_{p-1}}
  = 1, x_{e_r} = 1\}$ such that there are two matchings $M_1$ and $M_2$ in $F'$
  that differ inside $U$.  Notice that $F' \subseteq F = \{ x \in F_{t-1} : x_{e_{p-1}} = 1, x_{e_{r}} = 1\}$. Therefore $M_1$ and $M_2$ are also two matchings in $F$ that differ inside $U$.

  We complete the proof of the claim by showing that
  \begin{align}
    \label{condt}
    w_t(M_1 \cap E(U)) = w_t(M_2 \cap E(U))\,.
  \end{align}
  Define
  \begin{align*}
    M_{12} &=(M_1 \setminus E(U)) \cup (M_2 \cap E(U))
  \end{align*}
 to be the perfect matching that agrees with $M_1$ on
 all edges not in $E(U)$ and agrees with $M_2$ on all edges in $E(U)$ (see the central part of~\cref{fig:claims} for an example).
 By the same argument as in the proof of \cref{contractability_downward_closed},
 $M_{12}$ is a perfect matching in $F'$.
 It differs from $M_1$ inside $U$ and agrees with $M_1$ outside $U$.

  We now use that $M_1$ and $M_{12}$ are perfect matchings in $F'$ to prove~\eqref{condt}. As $F_t = \facemin{F_{t-1}}{w_t}$ is
  the convex-hull of matchings in $F_{t-1}$ that
  minimize the objective function $w_t$, all matchings $M$
  in $F_t$ and in its subface $F'$ have the same weight $w_t(M)$. 
  In particular, 
  \begin{align*}
w_t( M_1 \setminus E(U)) + w_t (M_1 \cap E(U)) = w_t(M_1) = w_t(M_{12}) = w_t( M_1 \setminus E(U)) + w_t (M_2 \cap E(U))
  \end{align*}
  and thus $w_t (M_1 \cap E(U)) = w_t (M_2 \cap E(U))$ as required.
\end{proof}

The above claim says that it is sufficient to write down a non-equality for
each choice of $U_{p,r}^{(i)}$, $e_{p-1}$, $e_r$, $M_1$, and $M_2$. It is easy to  upper-bound the number of
ways of choosing $i$, $p$, $r$, $e_{p-1}$, and $e_r$. The following
claim bounds the number of ways of choosing $M_1$ and $M_2$. Its proof is
based on a divide-and-conquer strategy (see the right part of~\cref{fig:claims}). It uses the inductive assumption that
$U^{(i)}_{p,r}$ is
$F_{t-1}$-contractible for all chains $i$ and $1\leq p \leq r \leq \ell_i$ with
$r-p\leq 2^{t-2}-1$. 
\begin{claim}
  \label{matching-count}
  Let $U= U^{(i)}_{p,r}$ with $r-p \leq 2^{t-1}-1$ and define $q = \lfloor
  (p+r)/2 \rfloor$.  For any two edges $e_p \in \delta(S^{(i)}_{p-1})$ and $e_r \in
  \delta(S^{(i)}_{r})$ defining a face $F = \{x\in F_{t-1}: x_{e_{p-1}} = 1,
  x_{e_r} = 1\}$ we have
  \begin{align*}
    |\{M \cap E(U) : M \mbox{ is a matching in $F$}\}| \leq  |\delta(S^{(i)}_{q}) \cap E(F)| \leq n^2\,.
  \end{align*}
\end{claim}
We remark that the first inequality holds with equality, but we only need the inequality.
\begin{proof}
  The second inequality in the statement is trivial. We prove the
  first. 
  
  As we have $S^{(i)}_q \in \Lin \subseteq \tight(\Fin)$
  and $F \subseteq F_{t-1} \subseteq \Fin$,
  the set $S^{(i)}_q$ is tight for $F$.
  Thus any matching
  $M$ in $F$ must satisfy $M\cap \delta(S^{(i)}_q) = \{e_q\}$ for some edge
  $e_q \in \delta(S^{(i)}_{q}) \cap E(F)$. 
  
  We prove the statement by showing that for every choice of $e_q$, any
  matching $M$ in the face $F_{e_q} = \{x \in F: x_{e_q} = 1\}$ matches the nodes in
  $U^{(i)}_{p,r}$ in a unique way.  In other words, we show that $|\{M \cap E(U) : M\mbox{ is a matching in $F_{e_q}$}\}| \leq 1$ for every $e_q\in \delta(S^{(i)}_q) \cap E(F)$, which implies
  \begin{align*}
    |\{M \cap E(U) : M \mbox{ is a matching in $F$}\}| &\leq  \sum_{e_q \in \delta(S^{(i)}_q) \cap E(F)}|\{M \cap E(U) : M \mbox{ is a matching in $F_{e_q}$}\}|  \\
    &\leq |\delta(S^{(i)}_{q}) \cap E(F)|\,.
  \end{align*}

  To prove that $|\{M \cap E(U)  : M\mbox{ is a matching in $F_{e_q}$}\}| \leq
  1$, suppose the contrary, i.e., that $|\{M \cap E(U) : M\mbox{ is a matching
    in $F_{e_q}$}\}| \geq 2$. Take two such matchings $M_1$ and $M_2$ that
    differ inside $U$.  By the definition of $F_{e_q}$ we have $M_1 \cap \delta(S^{(i)}_q)
    = M_2 \cap \delta(S^{(i)}_q) = \{e_q\}$ and so $M_1$ and $M_2$ must differ
    inside $U^{(i)}_{p, q}$ or inside $U^{(i)}_{q+1, r}$; assume the former (the argument for the other case is the same).
    Notice that $M_1$ and $M_2$ are two matchings in $F_{e_q} \subseteq F_{t-1}$ which both contain $e_{p-1}$ and $e_q$ but differ inside $U^{(i)}_{p,q}$, which contradicts that $U^{(i)}_{p,q}$ is $F_{t-1}$-contractible.  (Note that $q-p\leq (r-p)/2 \leq 2^{t-2} - 1/2$, which
    implies that $q-p \leq 2^{t-2} -1$.)
\end{proof}

We now have all the needed tools to show the existence of a weight function
$w_t \in \cWf$ such that the face $F_t = \facemin{F_{t-1}}{w_t}$
satisfies~\eqref{eqt}, i.e., that for all chains $i$ and $1\leq p \leq r \leq
\ell_i$ with $r-p\leq 2^{t-1}-1$, $U^{(i)}_{p,r}$ is $F_t$-contractible.
By~\cref{match-same-weight}, this holds if  for any $U=U^{(i)}_{p,r}$ with $r-p
\leq 2^{t-1} - 1$ and for any $e_{p-1} \in \delta(S^{(i)}_{p-1})$ and $e_r \in
\delta(S^{(i)}_{r})$ defining a face $F= \{x\in F_{t-1}: x_{e_{p-1}} = 1,
x_{e_r} = 1\}$ we have the following:
  \begin{align*}
    w_t(M_1\cap E(U)) - w_t(M_2\cap E(U))\neq 0 \quad \mbox{for any two matchings $M_1,M_2$ in $F$ that differ inside $U$.} 
  \end{align*}

There are at most $n$ ways of choosing $i$,  $n$ ways of choosing $p$, $n$ ways of choosing $r$, 
$n^2$ ways of choosing $e_{p-1}$, $n^2$ ways of choosing
$e_r$, and by Claim~\ref{matching-count}  there are at most $n^4$ ways of
choosing $M_1$ and $M_2$. In total, we can write the sufficient
conditions on the weight function $w_t$ as a system of at most $n^{11}$ linear
non-equalities with coefficients  in $\{-1,0,1\}$. It follows  by~\cref{lem23}
that there is a weight function $w_t \in \cW(n^{14}) \subseteq \cW(n^{20}) = \cW$
satisfying these conditions. This completes the
selection of $w_t$  and the proof of~\cref{making-2l-contractible}.

\subsection{A maximal laminar family completes the proof}
\label{ocompletion}

In \cref{good-weight-fn} we have demonstrated the existence of a weight function $\wout$ that
defines a face $\Fout$ with  properties $(i)'$ and $(ii)'$. We now show that
extending $\Lin$ to a maximal laminar family $\Lout$ of $\tight(\Fout)$ yields
a $2\lambda$-good face-laminar pair. Such an extension is possible because $\Lin$ consists of sets that are tight for $\Fin$ and thus also for $\Fout$.
As explained in the beginning of~\cref{proof_of_main}, this will complete the proof of~\cref{main}.

Why a \emph{maximal} laminar family?
Part of our argument so far was about removing certain alternating circuits $C$.
In other words, we have made $C$ not respect the new face $\Fout$.
This means either not having some edge from $\supp(C)$ in the support $\suppo(\Fout)$ of $\Fout$,
or introducing a new odd-set $S$ which is tight for $\Fout$ and such that $\ab{\pmind{C}, \one_{\delta(S)}} \ne 0$.
In the latter case, we want to have an odd-set with this property also in the new laminar family $\Lout$,
so that the removal of $C$ is reflected in the new contraction (which is based on $\Lout$).
\cref{respect_L_respect_F} guarantees that this will happen
if we choose $\Lout$ to be a {maximal} laminar subset of $\tight(\Fout)$.

\begin{lemma} \label{adding_new_laminar_sets}
  Let $(\Fin, \Lin)$ be a $\lambda$-good face-laminar pair
  and $\Fout \subseteq \Fin$ be the face guaranteed by \cref{good-weight-fn}.
  Then $(\Fout, \Lout)$ is a $2 \lambda$-good face-laminar pair,
  where $\Lout$ is any maximal laminar family with $\Lin \subseteq \Lout \subseteq \tight(\Fout)$.
\end{lemma}

\begin{proof}[Proof of \cref{adding_new_laminar_sets}]

  Recall that \cref{good-weight-fn} guarantees that:
  \begin{itemize}
    \item[(i)'] each $S \in {\Lin}$ with $|S| \le 2 \lambda$ is $\Fout$-contractible,
    \item[(ii)'] in the $(\Fout, {\Lin}, 2 \lambda)$-contraction of $G$, there is no alternating circuit of node-weight at most $2 \lambda$ {which respects $\Fout$}.
  \end{itemize}
  
  We want to show that the pair $(\Fout, \Lout)$ satisfies \cref{def:face-laminar}, that is,
  \begin{enumerate}
	\item[(i)] each $S \in {\Lout}$ with $|S| \le 2 \lambda$ is $\Fout$-contractible,
	\item[(ii)] in the $(\Fout, {\Lout}, 2 \lambda)$-contraction of $G$, there is no alternating circuit of node-weight at most $2 \lambda$.
  \end{enumerate}
  
  \paragraph{Property (i).} 
    Fix a set $S \in \Lout$ with $|S| \le 2 \lambda$.
    Let $S_1, ..., S_k$ be all maximal subsets of $S$ in $\Lin$ (we have $S = S_1 \cup ... \cup S_k$).
    If $S$ is contained in a set from $\Lin$ of size at most $2 \lambda$, then that set is $\Fout$-contractible by $(i)'$, and thus $S$ is $\Fout$-contractible by~\cref{contractability_downward_closed}.
    So assume that is not the case; therefore, by laminarity, each $S_i$ is a maximal set of size at most $2 \lambda$ in $\Lin$, that is, a vertex of the $(\Fout, \Lin, 2 \lambda)$-contraction of $G$.
    By $(i)'$, each $S_i$ is $\Fout$-contractible.
    
    Now the proof proceeds as in \cref{F1}.
    We present it for completeness.
    Let $M_1$ and $M_2$ be two perfect matchings in $\Fout$
    which both contain an edge $e \in \delta(S)$.
    We want to show that $M_1$ and $M_2$ are equal inside $S$.
    Because $S$ and $S_1, ..., S_k$ are tight for $\Fout$,
    any perfect matching (on $G$) in $\Fout$ containing $e$ induces an almost-perfect matching on $S_1, ..., S_k$,
    that is, one where only the set $S_i$ containing the $S$-endpoint of $e$ is unmatched.

    If the matchings induced by $M_1$ and $M_2$ were different, then their symmetric difference would contain an alternating simple cycle
    in the $(\Fout, \Lin, 2 \lambda)$-contraction.
    Since this cycle arises from two matchings in $\Fout$, it respects $\Fout$.
    Moreover, since it is a simple cycle inside $S$, its node-weight is at most $|S_1|  + \ldots + |S_k| = |S| \le 2 \lambda$.
    This would contradict our assumption~(ii)'.
    
    Therefore the induced matchings must be equal.
    Moreover, the sets $S_1, ..., S_k$ are $\Fout$-contractible, which means that, given the boundary edges
    (i.e., the induced matching plus $e$), there is a unique perfect matching in $\Fout$ inside each $S_i$.
    It follows that $M_1$ and $M_2$ are equal inside $S$.

  \newcommand{\Hout}{\ensuremath{H_{\textsf{\tiny out}}}}
  \newcommand{\Hin}{\ensuremath{H_{\textsf{\tiny in}}}}
  \newcommand{\Cout}{\ensuremath{C_{\textsf{\tiny out}}}}
  \newcommand{\Cin}{\ensuremath{C_{\textsf{\tiny in}}}}

  \paragraph{Property (ii).}
   Let $\Hout$ be the $(\Fout, \Lout, 2 \lambda)$-contraction of $G$
   and let $\Hin$ be the $(\Fout, \Lin, 2 \lambda)$-contraction of $G$.
   Thus $\Hout$ can also be obtained by further contracting $\Hin$, as well as removing all edges that are in the boundaries of sets $S\in \Lout \setminus \Lin$ with $|S| > 2 \lambda$. 
   This will be our perspective.  Suppose towards a contradiction that there is
   an alternating circuit $\Cout$ in $\Hout$ of node-weight at most
   $2\lambda$.

  To obtain a contradiction, we are going to lift $\Cout$ back to an $\Fout$-respecting alternating
  circuit $\Cin$ in $\Hin$, which should not exist by $(ii)'$. (This is in the same spirit as the proof of \cref{killing_cycle}.)
  Namely, whenever $\Cout$ visits a vertex $S \in V(\Hout)$, we connect up the dangling endpoints of this visit inside $S$ to obtain a walk in $\Hin$.
  More precisely, let $e_1$ and $e_2$ be two consecutive edges of $\Cout$, whose common endpoint in $\Hout$ is $S$. Between them, we insert a simple path $P_{e_1e_2}$ inside the image of $S$ in $\Hin$, which is constructed as follows.
  
  Since $e_1, e_2 \in \supp(\Cout) \subseteq \suppo(\Fout)$, there exist matchings $M_1$ and $M_2$ (on $G$) in $\Fout$ containing $e_1$ and $e_2$, respectively.
  Let $S_1, ..., S_k$ be all maximal subsets of $S$ in $\Lin$ ($S_i$ are vertices of $\Hin$ and we have $S = S_1 \cup ... \cup S_k$).
  Denote by $S_{e_1}$ and $S_{e_2}$ the sets $S_i$ which contain the $S$-endpoint of $e_1$ and $e_2$, respectively.
  The sets $S$ and $S_1, ..., S_k$ are tight for $\Fout$, so $M_1$ induces a perfect matching on $\{S_1, ..., S_k\} \setminus \{S_{e_1}\}$ (and similarly for $M_2$ and $e_2$).
  The symmetric difference of these two induced matchings contains a simple path $P_{e_1e_2}$ from $S_{e_1}$ to $S_{e_2}$ in $\Hin$ which has even length (possibly $0$). For an example, see~\cref{fig:propertyii}.
  We obtain $\Cin$ by inserting such a path $P_{e_1e_2}$ between each two consecutive edges $e_1, e_2$ in $\Cout$.
  
\begin{figure}[t]
  \begin{center}
    \begin{tikzpicture}[scale=0.9]
    \tikzstyle{vertex}=[circle, fill=black, minimum size=2,inner sep=1pt]
      \draw[fill=gray!10!white] (0,0) ellipse (4cm and 2cm)node[above = 1.9cm] {\small $S\in V(\Hout)$};
     \begin{scope}[yshift=0.5cm]
       \draw[fill = none,dashed] (0, 0) ellipse (2.2cm and 1cm)node[above = 0.9cm] {\scriptsize $T$};
       \draw (-3,0) edge node[above left] {\scriptsize $e_1$} (-4.5, 0.4);
       \draw (-3, 0) edge[decorate,decoration={snake,amplitude=.3mm,segment length=4pt,post length=0mm}, ultra thick] (-1.5,0);
       \draw (-1.5, 0) edge[ultra thick] (0,0);
       \draw (0, 0) edge[decorate,decoration={snake,amplitude=.3mm,segment length=4pt,post length=0mm},ultra thick] (1.5,0);
       \draw (1.5, 0) edge[ultra thick] (3,0);
       \draw (3,0) edge[decorate,decoration={snake,amplitude=.3mm,segment length=3pt,post length=0mm}] node[above right] {\scriptsize $e_2$} (4.5, 0.4);
       \draw[fill = gray!50!white, rotate around={115:(-3,0)}] (-3, 0) ellipse (0.3cm and 0.45cm);
       \draw[fill = gray!50!white, rotate around={115:(-1.5,0)}]  (-1.5, 0) ellipse (0.2cm and 0.35cm); 
       \draw[fill = gray!50!white, rotate around={0:(0,0)}]  (0, 0) ellipse (0.4cm and 0.25cm);
       \draw[fill = gray!50!white, rotate around={45:(1.5,0)}]  (1.5, 0) ellipse (0.32cm and 0.3cm);
       \draw[fill = gray!50!white, rotate around={45:(3,0)}]  (3, 0) ellipse (0.2cm and 0.35cm); 
       \node at (-3, -0.5) {\scriptsize $S_1 = S_{e_1}$};
       \node at (-1.5, -0.5) {\scriptsize $S_2$};
       \node at (0, -0.5) {\scriptsize $S_3$};
       \node at (1.5, -0.5) {\scriptsize $S_4$};
       \node at (3, -0.5) {\scriptsize $S_5 = S_{e_2}$};
    \end{scope}
    \begin{scope}[xshift=-0.5cm, yshift=-1.2cm]
       \draw (0, 0) edge[decorate,decoration={snake,amplitude=.3mm,segment length=3pt,post length=0mm},bend right=15] (1.5,0);
       \draw (0, 0) edge[bend left=15] (1.5,0);
       \draw[fill = gray!50!white, rotate around={0:(0,0)}]  (0, 0) ellipse (0.4cm and 0.35cm);
       \draw[fill = gray!50!white, rotate around={45:(1.5,0)}]  (1.5, 0) ellipse (0.42cm and 0.25cm);
       \node at (0, -0.55) {\scriptsize $S_6$};
       \node at (1.5, -0.55) {\scriptsize $S_7$};
     \end{scope}
    \end{tikzpicture}
  \end{center}
  \caption{The construction of the path $P_{e_1e_2}$ inside a set $S \in V(\Hout)$ in the proof of Property~(ii). The dark-gray sets correspond to vertices of $\Hin$ that are subsets of $S$. The straight and swirly edges depict matchings $M_1$ and $M_2$, respectively. The path $P_{e_1e_2}$ is depicted by fat edges.   }
  \label{fig:propertyii}
\end{figure}
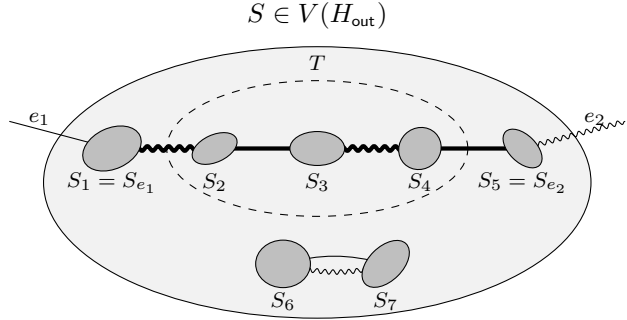

  To obtain a contradiction, we need to prove that $\Cin$ is an alternating circuit of node-weight at most $2 \lambda$ which respects $\Fout$.
  \begin{itemize}
  \item That $\Cin$ is an alternating circuit follows by construction because each path $P_{e_1e_2}$ is of even length and the alternating indicator vector of $\Cin$ is nonzero since it contains the alternating indicator vector of $\Cout$ (via the natural mapping of edges). 
  \item For the node-weight, note that in $\Cout$, the visit to $S$ (on the edge $e_1$) incurs a node-weight increase of $|S|$, whereas in $\Cin$, the visit to a certain subset of $\{S_1, ..., S_k\}$ (on $e_1$ and $P_{e_1e_2}$) incurs an increase of at most $|S_1| + \ldots + |S_k| = |S|$ because $P_{e_1e_2}$ is a simple path.
  Therefore the node-weight of $\Cin$ is at most that of $\Cout$ -- at most $2 \lambda$.
  
  \item To see that $\Cin$ respects $\Fout$, we use the assumption that $\Lout$ is a maximal laminar subset of $\tight(\Fout)$.
  This is in similar spirit as the proof of~\cref{claim:y_respects}.
  To work around the fact that we are dealing with the contraction,
  we use the following version of~\cref{respect_L_respect_F}:

  \begin{claim} \label{respect_L_respect_F_in_contraction}
	Consider a vector $z \in \bZ^{E(\Hin)}$.
	If for each $T \in \Lout$
	which is a union of sets in $V(\Hin)$
	we have
	$\ab{z, \one_{\delta(T)}} = 0$,
	then the same holds for each $T \in \tight(\Fout)$
	which is a union of sets in $V(\Hin)$.
  \end{claim}
  \vspace{-0.6em}
  
  Let us defer the proof of~\cref{respect_L_respect_F_in_contraction}
  and first use it to show that $\Cin$ respects $\Fout$.
  We verify the conditions of~\cref{def:respect_contraction}.
  First, note that $\supp(\Cin) \subseteq \suppo(\Fout)$ by construction.
  Second, let $T \in \tight(\Fout)$ be a union of vertices of $\Hin$;
  we need to show  that $\ab{\pmind{\Cin}, \one_{\delta(T)}} = 0$.
  By~\cref{respect_L_respect_F_in_contraction},
  we may assume that $T \in \Lout$.
  We consider two cases:
  \begin{itemize}
  \item If $|T| > 2 \lambda$, then all boundary edges of $T$ are absent from $\Hout$ (see \cref{def:contraction}), so $\supp(\Cout) \cap \delta(T) = \emptyset$. In this case $T$ is a union of vertices of $\Hout$ and so no path $P_{e_1e_2}$ contains any edges from $\delta(T)$ either. Hence $\supp(\Cin) \cap \delta(T) = \emptyset$.\footnote{Here we take advantage of the fact that the boundaries of large sets are erased in the definition of the contraction. If they were not erased, we would be unable to proceed, as there would be no reason for $\Cout$ (and thus $\Cin$) to respect $\Fout$.}
  
  \item If $|T| \le 2 \lambda$, then $T$ must be contained in a single set $S \in V(\Hout)$ (as depicted in~\cref{fig:propertyii}). This is because $T \in \Lout$ and the sets $S \in V(\Hout)$ are maximal sets $S \in \Lout$ with $|S| \le 2 \lambda$.
  For every path $P_{e_1e_2}$ inside $S$, the path $e_1, P_{e_1e_2}, e_2$ is a path from outside of $S$ to outside of $S$ which is part of the symmetric difference of two matchings in $\Fout$.
  If this path enters $T$, it must also leave $T$.
  Suppose it entered $T$ on an edge of the first matching. Then it must exit $T$ on an edge of the second matching,
  since $T \in \Lout \subseteq \tight(\Fout)$ is tight for $\Fout$ and both matchings are in $\Fout$,
  and so the corresponding $\pm 1$ terms cancel out.
  Abusing notation, we have $\ab{\pmind{e_1, P_{e_1e_2}, e_2}, \one_{\delta(T)}} = 0$.
  Since this holds for every path $P_{e_1e_2}$ inside $S$, we get $\ab{\pmind{\Cin}, \one_{\delta(T)}} = 0$ as required.
  \end{itemize}
  \end{itemize}
  
  Thus
  $\Cin$ is an alternating circuit of node-weight at most $2 \lambda$ which respects $\Fout$.
  Its
  existence contradicts $(ii)'$.
  We conclude with the proof of~\cref{respect_L_respect_F_in_contraction}:
  
  \begin{proof}[Proof of~\cref{respect_L_respect_F_in_contraction}]
  
  We wish to reduce our setting to that of~\cref{respect_L_respect_F}.
  Define $\Hin'$ to be the graph obtained from $(V, \suppo(\Fout))$ by contracting all maximal sets $S \in \Lin$ with $|S| \le 2 \lambda$,
  but {not} erasing the boundaries of sets $S \in \Lin$ with $|S| > 2 \lambda$.
  (If we did erase them, we would obtain $\Hin$; instead, 
  we have $V(\Hin) = V(\Hin')$ and $E(\Hin) \subseteq E(\Hin')$.)
  
  Next, we define $\Fout'$ to be the image of $\Fout$ in $\Hin'$.
  More precisely,
  since the contracted sets $S \in \Lin$ with $|S| \le 2 \lambda$
  are tight for $\Fout$,
  each perfect matching on $G$ in $\Fout$ induces a perfect matching on $\Hin'$.
  We let $\Fout'$ be the convex hull of the indicator vectors of these induced matchings.
  Note that there is a one-to-one correspondence
  between subsets of $V(\Hin)$ that are tight for $\Fout'$
  and
  subsets of $V$ that
  are tight for $\Fout$
  and
  are unions of sets in $V(\Hin)$.
  
  Finally, $\Lout$ also naturally maps to a laminar family $\Lout'$ of subsets of vertices of $\Hin'$
  since $\Lin \subseteq \Lout$. Specifically, there is a set in $\Lout'$ corresponding to each set in $\Lout$ that is a union of sets in $V(\Hin')$.
  Note that each set in $\Lout'$ is still tight for $\Fout'$,
  and that $\Lout'$ is still a {maximal} subset of $\tight(\Fout')$.
  Indeed, if it were possible to add any set in $\tight(\Fout')$ to $\Lout'$ while maintaining laminarity, then that set could be mapped back to a set in $\tight(\Fout)$ and used to enlarge $\Lout$.
  
  Now,
  by assumption,
  for each $T \in \Lout$ which is a union of sets in $V(\Hin)$
  we have
  $\ab{z, \one_{\delta(T)}} = 0$.
  This is equivalent to saying that
  $\ab{z, \one_{\delta(T')}} = 0$
  for each $T' \in \Lout'$.
  By~\cref{respect_L_respect_F}
  (applied to $\Hin'$, $\Fout'$, and $\Lout'$),
  we have the same
  for all $T' \in \tight(\Fout')$.
  Finally, that is equivalent to having
  $\ab{z, \one_{\delta(T)}} = 0$
  for each $T \in \tight(\Fout)$
  which is a union of sets in $V(\Hin)$.
  \end{proof}

\end{proof}

\fi

\vspace{-0.9em}
\section*{Acknowledgment}
We thank the anonymous FOCS 2017 reviewers for their exceptionally detailed and insightful comments.
\vspace{-0.4em}

\ifieee\else
	\appendix
\section{Proof of \cref{face_structure}} \label{proof_of_face_structure}

The proof proceeds via the primal uncrossing technique; it is adapted from \cite{LauRS11}. 
Assume without loss of generality that $E = \suppo(F)$. We can do this since including the constraint $x_e = 0$ yields the same face as removing the edge $e$ from $G$.
We begin with an uncrossing lemma.

%
%

\begin{lemma}[uncrossing] \label{uncrossing}
Let $S, T \in \cS(F)$ be two sets which are crossing (i.e., $S \cap T, S \setminus T, T \setminus S \ne \emptyset$). Then:
\begin{itemize}
	\item if $|S \cap T|$ is odd: then $S \cap T, S \cup T \in \cS(F)$ and $\one_{\delta(S)} + \one_{\delta(T)} = \one_{\delta(S \cap T)} + \one_{\delta(S \cup T)}$,
	\item otherwise: $S \setminus T, T \setminus S \in \cS(F)$ and $\one_{\delta(S)} + \one_{\delta(T)} = \one_{\delta(S \setminus T)} + \one_{\delta(T \setminus S)}$.
\end{itemize}
\end{lemma}
\begin{proof}
\textbf{Case $|S \cap T|$ odd.}
Note that we have
\[ \one_{\delta(S)} + \one_{\delta(T)} = \one_{\delta(S \cap T)} + \one_{\delta(S \cup T)} + 2 \cdot \one_{\delta(S \setminus T, T \setminus S)}. \]
For any $x \in F$, since $S, T \in \cS(F)$ and because $S \cap T$, $S \cup T$ are nonempty odd sets, we have
\[ 1 + 1 = x(\delta(S)) + x(\delta(T)) = x(\delta(S \cap T)) + x(\delta(S \cup T)) + 2 \cdot x(\delta(S \setminus T, T \setminus S)) \ge 1 + 1 + 2 \cdot 0 \]
where the inequality must be an equality, and thus $x(\delta(S \cap T)) = 1$, $x(\delta(S \cup T)) = 1$ (implying $S \cap T, S \cup T \in \cS(F)$) and $x(\delta(S \setminus T, T \setminus S)) = 0$ for all $x \in F$ (which, given that $E = \suppo(F)$, implies that $\delta(S \setminus T, T \setminus S) = \emptyset$ and thus $\one_{\delta(S \setminus T, T \setminus S)} = 0$).

\textbf{Case $|S \cap T|$ even.}
Now we have
\[ \one_{\delta(S)} + \one_{\delta(T)} = \one_{\delta(S \setminus T)} + \one_{\delta(T \setminus S)} + 2 \cdot \one_{\delta(S \cap T, V \setminus (S \cup T))}. \]
The sets $S \setminus T$ and $T \setminus S$ are odd and nonempty, and we proceed as above.
\end{proof}

\newcommand{\cross}{\mathrm{cross}}
Define $\cross(S, \cL)$ to be the number of sets in $\cL$ that cross $S$.

\begin{proposition} \label{intersects_are_smaller}
If $S \not \in \cL$ and $T \in \cL$ are crossing, then all four numbers $\cross(S \cap T, \cL)$, $\cross(S \cup T, \cL)$, $\cross(S \setminus T, \cL)$ and $\cross(T \setminus S, \cL)$ are smaller than  $\cross(S, \cL)$.
\end{proposition}
\begin{proof}
See Claim 9.1.6 in \cite{LauRS11}.
\end{proof}

Now we can prove \cref{face_structure}.
Towards a contradiction suppose that $\Span(\cL) \subsetneq \Span(\cS(F))$. Then there exists $S \in \cS(F)$ with $\one_{\delta(S)} \not \in \Span(\cL)$. Pick such a set with minimum $\cross(S, \cL)$. Clearly $\cross(S, \cL) \ge 1$, for otherwise $\cL \cup \{S\}$ would be laminar, contradicting maximality of $\cL$. Let $T \in \cL$ be a set crossing $S$. Assume that $|S \cap T|$ is odd; the other case is analogous.
Then by \cref{uncrossing}, $S \cap T, S \cup T \in \cS(F)$ and 
\begin{equation} \label{equation_on_deltas}
\one_{\delta(S)} + \one_{\delta(T)} = \one_{\delta(S \cap T)} + \one_{\delta(S \cup T)}.
\end{equation} By \cref{intersects_are_smaller} and our choice of $S$ we have $\one_{\delta(S \cap T)}, \one_{\delta(S \cup T)} \in \Span(\cL)$, and of course also $\one_{\delta(T)} \in \Span(\cL)$. This and \eqref{equation_on_deltas} implies that $\one_{\delta(S)} \in \Span(\cL)$, a contradiction.
\qedmanual

\fi

\ifieee
\else
  \bibliographystyle{alpha}
  \bibliography{references}

\begin{thebibliography}{Edm65b}

\bibitem[AHT07]{AgrawalHT07}
Manindra Agrawal, Thanh~Minh Hoang, and Thomas Thierauf.
\newblock The polynomially bounded perfect matching problem is in
  {NC}\({}^{\mbox{2}}\).
\newblock In {\em {STACS} 2007, 24th Annual Symposium on Theoretical Aspects of
  Computer Science}, pages 489--499, 2007.

\bibitem[AM08]{ArvindM08}
Vikraman Arvind and Partha Mukhopadhyay.
\newblock Derandomizing the isolation lemma and lower bounds for circuit size.
\newblock In {\em {APPROX} and {RANDOM}}, pages 276--289, 2008.

\bibitem[Bar92]{Barrington92}
D.~A.~M. Barrington.
\newblock Quasipolynomial size circuit classes.
\newblock In {\em Proceedings of the Seventh Annual Structure in Complexity
  Theory Conference}, pages 86--93, Jun 1992.

\bibitem[BCH86]{BeameCH86}
Paul~W Beame, Stephen~A Cook, and H~James Hoover.
\newblock Log depth circuits for division and related problems.
\newblock {\em SIAM J. Comput.}, 15(4):994--1003, November 1986.

\bibitem[Ber84]{Berkowitz1984}
Stuart~J. Berkowitz.
\newblock On computing the determinant in small parallel time using a small
  number of processors.
\newblock {\em Information Processing Letters}, 18(3):147--150, 1984.

\bibitem[CNN89]{ChrobakNN89}
Marek Chrobak, Joseph Naor, and Mark~B. Novick.
\newblock {\em Using bounded degree spanning trees in the design of efficient
  algorithms on claw-free graphs}, pages 147--162.
\newblock Springer Berlin Heidelberg, Berlin, Heidelberg, 1989.

\bibitem[Csa76]{Csanky76}
L.~Csanky.
\newblock Fast parallel inversion algorithm.
\newblock {\em SIAM Journal of Computing}, 5:618--623, 1976.

\bibitem[DHK93]{DahlhausHK93}
E.~Dahlhaus, P.~Hajnal, and M.~Karpinski.
\newblock On the parallel complexity of {H}amiltonian cycle and matching
  problem on dense graphs.
\newblock {\em Journal of Algorithms}, 15(3):367 -- 384, 1993.

\bibitem[DK98]{DahlhausK98}
Elias Dahlhaus and Marek Karpinski.
\newblock Matching and multidimensional matching in chordal and strongly
  chordal graphs.
\newblock {\em Discrete Applied Mathematics}, 84(1-3):79--91, 1998.

\bibitem[DKR10]{DattaKR10}
Samir Datta, Raghav Kulkarni, and Sambuddha Roy.
\newblock Deterministically isolating a perfect matching in bipartite planar
  graphs.
\newblock {\em Theory Comput. Syst.}, 47(3):737--757, 2010.

\bibitem[DS84]{DekelS84}
Eliezer Dekel and Sartaj Sahni.
\newblock A parallel matching algorithm for convex bipartite graphs and
  applications to scheduling.
\newblock {\em Journal of Parallel and Distributed Computing}, 1(2):185 -- 205,
  1984.

\bibitem[Edm65a]{Edm65}
Jack Edmonds.
\newblock Maximum matching and a polyhedron with $0,1$ vertices.
\newblock {\em Journal of Research of the National Bureau of Standards},
  69:125--130, 1965.

\bibitem[Edm65b]{edm65matching}
Jack Edmonds.
\newblock Paths, trees, and flowers.
\newblock {\em Canadian Journal of Mathematics}, 17:449--467, 1965.

\bibitem[FGT16]{FennerGT16}
Stephen~A. Fenner, Rohit Gurjar, and Thomas Thierauf.
\newblock Bipartite perfect matching is in {quasi-NC}.
\newblock In {\em Proceedings of the 48th Annual {ACM} {SIGACT} Symposium on
  Theory of Computing, {STOC}}, pages 754--763, 2016.

\bibitem[GK87]{GrigorievK87}
Dima Grigoriev and Marek Karpinski.
\newblock The matching problem for bipartite graphs with polynomially bounded
  permanents is in {NC}.
\newblock In {\em 28th Annual Symposium on Foundations of Computer Science
  ({FOCS})}, pages 166--172, 1987.

\bibitem[GT17]{GurjarT16}
Rohit Gurjar and Thomas Thierauf.
\newblock Linear matroid intersection is in quasi-{NC}.
\newblock In {\em Proceedings of the 48th Annual {ACM} {SIGACT} Symposium on
  Theory of Computing, {STOC}}, pages 821--830, 2017.

\bibitem[GTV17]{GurjarTV17}
R.~{Gurjar}, T.~{Thierauf}, and N.~K. {Vishnoi}.
\newblock {Isolating a Vertex via Lattices: Polytopes with Totally Unimodular
  Faces}.
\newblock {\em ArXiv e-prints}, August 2017.

\bibitem[Har09]{Harvey09}
Nicholas J.~A. Harvey.
\newblock Algebraic algorithms for matching and matroid problems.
\newblock {\em {SIAM} J. Comput.}, 39(2):679--702, 2009.

\bibitem[Kas67]{Kasteleyn67}
P.~W. Kasteleyn.
\newblock Graph theory and crystal physics.
\newblock In F.~Harary, editor, {\em Graph Theory and Theoretical Physics},
  pages 43--110. Academic Press, 1967.

\bibitem[KUW86]{KarpUW86}
Richard~M. Karp, Eli Upfal, and Avi Wigderson.
\newblock Constructing a perfect matching is in random {NC}.
\newblock {\em Combinatorica}, 6(1):35--48, 1986.

\bibitem[KVV85]{KozenVV85}
Dexter Kozen, Umesh~V. Vazirani, and Vijay~V. Vazirani.
\newblock {NC} algorithms for comparability graphs, interval gaphs, and testing
  for unique perfect matching.
\newblock In {\em Proceedings of the Fifth Conference on Foundations of
  Software Technology and Theoretical Computer Science}, pages 496--503,
  London, UK, 1985. Springer-Verlag.

\bibitem[Lov79]{Lovasz79}
L{\'{a}}szl{\'{o}} Lov{\'{a}}sz.
\newblock On determinants, matchings, and random algorithms.
\newblock In {\em {FCT}}, pages 565--574, 1979.

\bibitem[LPV81]{LevPV81}
G.~F. Lev, N.~Pippenger, and L.~G. Valiant.
\newblock A fast parallel algorithm for routing in permutation networks.
\newblock {\em IEEE Transactions on Computers}, C-30(2):93--100, Feb 1981.

\bibitem[LRS11]{LauRS11}
Lap~Chi Lau, Ramamoorthi Ravi, and Mohit Singh.
\newblock {\em Iterative methods in combinatorial optimization}, volume~46.
\newblock Cambridge University Press, 2011.

\bibitem[MN89]{MillerN89}
G.~L. Miller and J.~Naor.
\newblock Flow in planar graphs with multiple sources and sinks.
\newblock In {\em 30th Annual Symposium on Foundations of Computer Science},
  pages 112--117, 1989.

\bibitem[MS04]{MuchaS04}
Marcin Mucha and Piotr Sankowski.
\newblock Maximum matchings via {G}aussian elimination.
\newblock In {\em 45th Symposium on Foundations of Computer Science {(FOCS})},
  pages 248--255, 2004.

\bibitem[MV97]{MahajanV97}
Meena Mahajan and V.~Vinay.
\newblock Determinant: Combinatorics, algorithms, and complexity.
\newblock Technical report, 1997.

\bibitem[MV00]{MahajanV00}
Meena Mahajan and Kasturi~R. Varadarajan.
\newblock A new {NC}-algorithm for finding a perfect matching in bipartite
  planar and small genus graphs.
\newblock In {\em Proceedings of the Thirty-second Annual ACM Symposium on
  Theory of Computing, {STOC}}, pages 351--357, 2000.

\bibitem[MVV87]{MulmuleyVV87}
Ketan Mulmuley, Umesh~V. Vazirani, and Vijay~V. Vazirani.
\newblock Matching is as easy as matrix inversion.
\newblock {\em Combinatorica}, 7(1):105--113, 1987.

\bibitem[Nai82]{Nair82}
M.~Nair.
\newblock On {C}hebyshev-type inequalities for primes.
\newblock {\em The American Mathematical Monthly}, 89(2):126--129, 1982.

\bibitem[NSV94]{NarayananSV94}
H.~Narayanan, Huzur Saran, and Vijay~V. Vazirani.
\newblock Randomized parallel algorithms for matroid union and intersection,
  with applications to arboresences and edge-disjoint spanning trees.
\newblock {\em {SIAM} J. Comput.}, 23(2):387--397, 1994.

\bibitem[Par98]{Parfenoff98}
I.~Parfenoff.
\newblock An efficient parallel algorithm for maximum matching for some classes
  of graphs.
\newblock {\em Journal of Parallel and Distributed Computing}, 52(1):96 -- 108,
  1998.

\bibitem[PY82]{PapadimitriouY82}
Christos~H. Papadimitriou and Mihalis Yannakakis.
\newblock The complexity of restricted spanning tree problems.
\newblock {\em J. ACM}, 29(2):285--309, April 1982.

\bibitem[Sch03]{Schrijver03}
Alexander Schrijver.
\newblock {\em Combinatorial Optimization - Polyhedra and Efficiency}.
\newblock Springer-Verlag, Berlin, 2003.

\bibitem[Tut47]{Tutte}
W.~T. Tutte.
\newblock The factorization of linear graphs.
\newblock {\em Journal of the {L}ondon Mathematical Society}, 22:107--111,
  1947.

\bibitem[TV12]{TewariV12}
Raghunath Tewari and N.~V. Vinodchandran.
\newblock Green's theorem and isolation in planar graphs.
\newblock {\em Inf. Comput.}, 215:1--7, 2012.

\bibitem[Vaz89]{Vazirani89}
Vijay~V. Vazirani.
\newblock {NC} algorithms for computing the number of perfect matchings in
  \({K}_{3,3}\)-free graphs and related problems.
\newblock {\em Information and Computation}, 80(2):152 -- 164, 1989.

\end{thebibliography}
\fi

\end{document}